\pdfoutput=1
\documentclass[10pt,letterpaper]{article}  
\usepackage[utf8]{inputenc}
\usepackage{amsmath, amsfonts, amsthm, amssymb, aliascnt}
\usepackage{graphicx, xcolor}
\usepackage{tikz, pgfplots, tikz-3dplot}
\pgfplotsset{compat=1.10}
\usepackage{accents, verbatim}
\usepackage{microtype}
\usepackage[textheight = 700pt, textwidth = 460pt,footskip=30pt]{geometry}
\usepackage[nottoc]{tocbibind} 
\usepackage{hyperref,bookmark,footnotebackref}
\usepackage{relsize}
\hypersetup{linktoc = all, hidelinks, bookmarksnumbered}
\makeatletter
\renewcommand{\l@paragraph}[2]{\ifnum 4>\c@tocdepth \else
  \vskip \z@ \@plus .2\p@ {\smaller
    \leftskip 7em\relax \rightskip \@tocrmarg
    \parfillskip -\rightskip
    \parindent 7em\relax
    \@afterindenttrue
    \interlinepenalty \@M
    \leavevmode
    \@tempdima 4.1em\relax
    \advance \leftskip \@tempdima
    \null \nobreak \hskip -\leftskip {#1}\hfill\nobreak\par }\fi}
\makeatother
\usepackage{enumitem}
\setlist[2]{nosep}

\usepackage[mathcal]{euscript}
\usepackage{dsfont, mathrsfs}
\usepackage{pxfonts}
\renewcommand{\emph}{\textsl}
\makeatletter
\newcommand*\smallbullet{\mathpalette\smallbullet@{.5}}
\newcommand*\smallbullet@[2]{\mathbin{\vcenter{\hbox{\scalebox{#2}{$\m@th#1\bullet$}}}}}
\makeatother
\DeclareSymbolFont{origAMSa}{U}{msa}{m}{n}
\DeclareMathSymbol{\Box}{\mathord}{origAMSa}{3}
\renewcommand{\geq}{\geqslant}
\renewcommand{\leq}{\leqslant}
\theoremstyle{plain}
\newtheorem{theorem}{Theorem}[section]
\newcommand{\mynewtheorem}[2]{
  \newaliascnt{#1}{theorem}
  \newtheorem{#1}[#1]{#2}
  \aliascntresetthe{#1}
  \expandafter\providecommand\csname #1autorefname\endcsname{#2}
}
\mynewtheorem{definition}{Definition}
\mynewtheorem{proposition}{Proposition}
\mynewtheorem{lemma}{Lemma}
\mynewtheorem{corollary}{Corollary}
\mynewtheorem{remark}{Remark}
\numberwithin{equation}{section}
\numberwithin{figure}{section}

\newcommand{\bei}{\begin{itemize}}
\newcommand{\eei}{\end{itemize}}
\newcommand{\be}{\begin{equation}}
\newcommand{\ee}{\end{equation}} 
\newcommand{\bel}{\be \label}
\newcommand{\bse}{\begin{subequations}}
\newcommand{\ese}{\end{subequations}}
\DeclareMathOperator{\Ric}{Ric}
\DeclareMathOperator{\diag}{diag}
\DeclareMathOperator{\Vol}{Vol}

\DeclareMathOperator{\Tr}{Tr}
\DeclareMathOperator{\sgn}{sgn}
\newcommand{\Id}{\textnormal{\bf I}}
\newcommand{\abs}[1]{\lvert#1\rvert}
\newcommand{\RR}{\mathbb{R}}
\newcommand{\ZZ}{\mathbb{Z}}
\newcommand{\Tbb}{\mathbb{T}}
\newcommand{\del}{\partial}
\newcommand{\SF}{\textnormal{SF}}
\newcommand{\Kasner}{\textnormal{Kasner}}
\newcommand{\Hcal}{\mathcal{H}}
\newcommand{\Lcal}{\mathcal{L}}
\newcommand{\Mcal}{\mathcal{M}}
\newcommand{\Ncal}{\mathcal{N}}
\newcommand{\Pcal}{\mathcal{P}}
\newcommand{\Ucal}{\mathcal{U}}
\newcommand{\Ispace}{\mathbf I_{\textnormal{space}}}
\newcommand{\Ispacep}{\mathbf I_{\textnormal{space}}^{K>0}}
\newcommand{\Itime}{\mathbf I_{\textnormal{time}}}
\newcommand{\Itimep}{\mathbf I_{\textnormal{time}}^{K>0}}
\newcommand{\Ibf}{\mathbf I}
\newcommand{\Sbf}{\mathbf S}
\newcommand{\Sirc}{\mathbf S^{\textnormal{iso,rc}}}
\newcommand{\Sarc}{\mathbf S^{\textnormal{ani,rc}}}
\newcommand{\Siso}{\mathbf S^{\textnormal{iso}}}
\newcommand{\Sani}{\mathbf S^{\textnormal{ani}}}
\newcommand{\Ipoint}{\mathbf{I}_{\textnormal{point}}}
\newcommand{\Spoint}{\mathbf{S}_{\textnormal{point}}}
\newcommand{\phimin}{\phi_{\textnormal{min}}}
\newcommand{\nablamoi}{\nabla_-{}}
\newcommand{\nablapoi}{\nabla_+{}}
\newcommand{\Gammamoi}{\Gamma_-}
\newcommand{\Gammapoi}{\Gamma_+}
\makeatletter
\newcommand{\phimoi}{\phi@aux{-}}
\newcommand{\phipoi}{\phi@aux{+}}
\newcommand{\phi@aux}[1]{\@ifnextchar_{\phi@aux@{#1}}{\phi_{#1}}}
\newcommand{\phi@aux@}[3]{\phi_{#3#1}}
\makeatother
\newcommand{\ku}{\underline{k}}
\newcommand{\Ku}{\underline{K}}
\newcommand{\Kucirc}{\mathring{\Ku}}
\newcommand{\phiu}{\underline{\phi}}
\newcommand{\rmoi}{r_-}
\newcommand{\rpoi}{r_+}
\newcommand{\thetamoi}{\theta_-}
\newcommand{\thetapoi}{\theta_+}
\newcommand{\gmoi}{g_-{}}
\newcommand{\gpoi}{g_+{}}
\newcommand{\Kmoi}{K_-^{}{}}
\newcommand{\Kmoicirc}{\mathring{K}_-}
\newcommand{\Kpoi}{K_+^{}{}}
\newcommand{\Kpoicirc}{\mathring{K}_+}
\newcommand{\Kcirc}{\mathring{K}}
\newcommand{\gst}{g_*{}}
\newcommand{\Kst}{K_*^{}{}}
\newcommand{\phist}{\phi_*{}}
\newcommand{\Jst}{J_*{}}
\newcommand{\nablast}{\nabla_*{}}
\newcommand{\rhost}{\rho_*{}}
\newcommand{\Tst}{T_*{}}
\newcommand{\gpoint}{g{}}
\newcommand{\Kpoint}{K{}}
\newcommand{\phipoint}{\phi{}}
\newcommand{\gt}{\widetilde{g}}
\newcommand{\nablat}{\widetilde{\nabla}}
\newcommand{\Ht}{\widetilde{H}} 
\newcommand{\Lt}{\widetilde{L}} 
\makeatletter
\newcommand{\preprint}[1]{%
  \newcommand{\ps@preprintnumber}{\ps@plain
    \renewcommand{\@oddhead}{\hfill\begin{picture}(0,0)\put(0,-40){\llap{#1}}\end{picture}}%
    \let\@evenhead\@oddhead}%
  \thispagestyle{preprintnumber}%
}
\makeatother
\begin{document}

\title{Cyclic spacetimes through singularity scattering maps.
\\
The laws of quiescent bounces
\footnotetext{$^1$ 
Philippe Meyer Institute, Physics Department, \'Ecole Normale Sup\'erieure, PSL Research University, 24 rue Lhomond, F-75231 Paris Cedex 05, France.  
\\
$^2$ Laboratoire de Physique Théorique et Hautes Énergies, 
Centre National de la Recherche Scientifique, 
Sorbonne Universit\'e, 4 Place Jussieu, 75252 Paris, France. Email: {\tt bruno@le-floch.fr}.
\\
$^3$ Laboratoire Jacques-Louis Lions \& Centre National de la Recherche Scientifique, 
Sorbonne Universit\'e,  
4 Place Jussieu, 75252 Paris Cedex, France. Email: {\tt contact@philippelefloch.org}. 
\\
$^4$ CERN, Theory Department, CH-1211 Geneva 23, Switzerland. Email: {\tt Gabriele.Veneziano@cern.ch.} 
\\
$^5$ Coll\`ege de France, 11 Place M. Berthelot, 75005 Paris, France. 
}
} 
\author{Bruno Le Floch$^{1,2}$, Philippe G. LeFloch$^3$, and Gabriele Veneziano$^{4,5}$  
}
\date{}

\maketitle

\begin{abstract} 
For spacetimes containing quiescent singularity hypersurfaces we propose a general notion of junction conditions based on a prescribed {\sl singularity scattering map}, as we call it, and we introduce the notion of a {\sl cyclic spacetime} (also called a multiverse) consisting of spacetime domains bounded by spacelike or timelike singularity hypersurfaces, across which our scattering map is applied. A local existence theory is established here while, in a companion paper, we construct plane-symmetric cyclic spacetimes. 
We study the singularity data space consisting of the suitably rescaled metric, extrinsic curvature, and matter fields which can be prescribed on each side of the singularity, and for the class of so-called quiescent singularities we establish restrictions that a singularity scattering map must satisfy.
We obtain a full characterization of all scattering maps that are covariant and ultralocal, in a sense we define and, in particular, we distinguish between, on the one hand, {\sl three laws of bouncing cosmology} of universal nature and, on the other hand, {\sl model-dependent junction conditions.}
The theory proposed in this paper applies to spacelike and timelike hypersurfaces and without symmetry restriction.
It encompasses bouncing-cosmology scenarios, both in string theory and in loop quantum cosmology, and puts strong restrictions on their possible explicit realizations.
\end{abstract}

\preprint{\href{https://doi.org/10.1007/JHEP04(2022)095}{doi:10.1007/JHEP04(2022)095}\hspace{8.6cm}CERN-TH-2020-100}

\setcounter{tocdepth}{2}
\tableofcontents


\section{Introduction}

\subsection{Toward a theory of cyclic spacetimes} 

\paragraph{Main contribution.}

We investigate here the problem of crossing cosmological singularities in the context of the Einstein-scalar field system.
We study the nature of singularities in general relativity (without symmetry restriction) and address
the question of extending a spacetime {\sl beyond a spacelike singularity hypersurface,}
as well as whether a spacetime can contain a {\sl timelike singularity hypersurface.} Our contribution is two-fold 
and relies on notions of Lorentzian geometry and theoretical physics modeling. 

\bei 

\item {\bf A notion of cyclic spacetimes.}
We propose a notion of {\sl cyclic spacetime} (in \autoref{def:singu2} below), which is cast in a form that 
can conveniently be applied.
As a direct application, we establish a local existence theory for the initial value problem,
based on the generic power-law behaviour of the metric understood by Belinsky, Khalatnikov and Lifshitz (BKL)~\cite{BK-scalar,BKL}.
Our construction produces a broad class of spacetimes containing Big Crunch-Big Bang transitions or timelike singular interfaces. 

\item {\bf A classification of all singularity scattering maps.}
Our notion of cyclic spacetime is based on specifying a {\sl singularity scattering map} that describes how data on both sides of the singularity are related. Inspired by the ultralocality of the BKL expansion near the singularity, we focus our attention on singularity scattering maps that are ultralocal (or pointwise) and we establish a complete classification thereof.

\eei

\noindent As a consequence of our analysis, across a bounce we can distinguish between, on the one hand,
 {\sl three laws of bouncing cosmology} which are of universal nature and, on the other hand, 
 {\sl model-dependent junction conditions} which involve only a limited number of defining functions and 
 must depend upon additional physics beyond general relativity. 
In addition, in the companion paper~\cite{LLV-3}, we study the {\sl global geometry of plane-symmetric cyclic spacetimes.}
In the plane-symmetric case we solve the gravitational wave interaction problem {\sl globally.}  This global resolution to the collision problem involves both spacelike and timelike singularity hypersurfaces, which are traversed using a singularity scattering map.
The reader is also referred to \cite{LLV-2} for a brief overview of our main results.   

\paragraph{Global dynamics of self-gravitating matter.}
 
Many spacetimes satisfying the Einstein equations exhibit singularities such as curvature singularities or, at least, suffer from geodesic incompleteness as established by Penrose and Hawking~\cite{HE}. However, our theoretical knowledge about the structure of such singularities is extremely limited. 
One important issue in general relativity is deciding whether the Einstein equations provide a fully predictive theory in the sense that it uniquely determines the global evolution of the geometry and matter fields from their knowledge on a Cauchy hypersurface. Rather partial results are available and typically encompass only solutions that are globally close to Minkowski spacetime for ``small'' matter fields. 

The series of papers \cite{LeFlochLeFloch-1}--\cite{LeFlochLeFloch-5} has recently initiated a program on the mathematical study of the global dynamics of matter fields, which stems from pioneering contributions by Christodoulou on the global evolution problem in spherical symmetry and Penrose's censorship conjectures. In this direction, one outstanding question that arises naturally is whether a spacetime determined by solving the initial value problem associated with the Einstein equations can be continued so that the corresponding future globally hyperbolic Cauchy development, understood in a suitable sense, is unique.  A mathematically as well as physically consistent theory must allow for an extension beyond geometric singularities. 
For further material on singular solutions to Einstein equations, see \cite{PLFMardare,PLFRendall,PLFSormani,PLFStewart}.  

\paragraph{Bouncing through singularities.}

Another motivation for traversing geometric singularities stems from cosmology.
In the past thirty years, bouncing cosmologies and junction conditions at the bounce were proposed in many approaches:
Pre-Big Bang scenarios of Gasperini and Veneziano~\cite{GasperiniVeneziano1,GasperiniVeneziano2}
(and \cite{BV,BV1,BDV,FKV,KohlprathVeneziano,VenezianoSFD})
expyrotic models spearheaded by Steinhardt and Turok~\cite{KOST,SteinhardtTurok2004},
matter bounces of Brandenberger and Finelli~\cite{Brandenberger,FinelliBrandenberger},
as well as constructions based on
string gas cosmology of Brandenberger and Vafa~\cite{BrandenbergerVafa,NBV},
loop quantum cosmology in the Ashtekar school~\cite{AshtekarPawlowskiSingh,Ashtekar,AshtekarWilsonEwing},
and certain modified gravity theories such as~\cite{BiswasMazumdarMazumdar,BBMS,Chamseddine,Cesare-bounce}.
These approaches resolve the initial cosmological singularity through violations of null-energy conditions, modifications of Einstein gravity, or quantum gravity effects that only affect dynamics near the bounce.
We discuss some of these scenarios further in this text, and refer the reader to the review by Brandenberger and Peter~\cite{BrandenbergerPeter} on bouncing cosmologies.
An important alternative proposal is the conformal cyclic cosmology introduced by Penrose~\cite{PenroseCCC1}, followed by Tod, L\"ubbe, and others~\cite{LubbeTod,Lubbe,Tod:2002wd}. Our method should extend to Penrose's scenario, but this issue is outside the  scope of the present paper.

The Einstein equations admit solutions representing matter spacetimes that have ``quiescent'' singularities ---a class first named by Barrow~\cite{BK-scalar,Barrow,DHS}. Our aim in the present paper is to analyze
the class of such spacetimes (without symmetry restriction), which encompasses behavior generically observed in the presence of a sufficiently ``strong'' massless scalar field; in~\cite{LLV-3} we apply our theory to study plane-symmetric spacetimes in this context. In contrast, vacuum spacetimes are expected to feature spacelike singularities with an oscillating behavior~\cite{BKL} or null Cauchy horizons~\cite{DafermosLuk}.

\subsection{The notion of cyclic spacetime}  

\paragraph{Beyond standard junction conditions.}

We are interested in $4$-dimensional spacetimes $(\Mcal, g^{(4)})$ (with boundary), required to satisfy Einstein-scalar field equations of general relativity 
\be
G = 8\pi T. 
\ee
Here, $G$ denotes the Einstein tensor of~$g^{(4)}$ and $T$ the energy-momentum tensor, while the Newton constant and the light speed are normalized to unity.
We consider a massless scalar field $\phi\colon \Mcal \to \RR$ with energy-momentum tensor
\be
T = d\phi \otimes d\phi -\frac{1}{2} |d\phi|^2 g^{(4)},
\ee
which can also be used to describe an irrotational stiff fluid.
Under these conditions, the Einstein equations are equivalent to equations on the Ricci curvature $\Ric$ of the metric, that is,  
\bse\label{intro-Einstein-scalar}
\bel{intro-Einstein-Rc}
\Ric = 8\pi\, d\phi \otimes d\phi. 
\ee
In addition, the Bianchi identities imply that the scalar field satisfies the wave equation associated with the wave operator $\Box$ associated with the metric, that is,
\be
\Box \phi = 0.
\ee 
\ese

Solutions to \eqref{intro-Einstein-scalar} may exhibit singularities localized along a hypersurface in $(\Mcal, g^{(4)})$. 
The standard junction conditions (discovered by Israel \cite{IJC}, and also investigated by Darmois, Lichnerowitz, Penrose, and others)
apply to (regular) hypersurfaces when the spacetime metric and the extrinsic curvature are sufficiently regular {\sl up to the hypersurface} and only possibly suffer a jump discontinuity across it. They are derived from the ADM equations (introduced below)
by integration in an arbitrarily small neighborhood of the hypersurface, and allow for impulsive (measure) contributions contributed by the matter. 

One of the ADM equations (see below) does not involve the matter field and, assuming that the extrinsic curvature remains bounded, this ADM equation implies the continuity of the metric. On the other hand, the other ADM equation implies that the jump of the extrinsic curvature is compensated by a (possibly vanishing) matter surface term in the energy-momentum tensor.
All terms constructed from the metric and extrinsic curvature remain bounded in this regime (albeit possibly discontinuous), and only matter provides singular contributions to the ADM equations and constraints.
In contrast, our setup in the present paper concerns a foliation of hypersurfaces whose extrinsic curvature blows up for some value of the foliation parameter.  

\paragraph{Notion of singularity scattering map.}

We thus consider the Cauchy problem in the ADM formalism, in which solutions of Einstein-scalar equations are represented as an {\bf Einstein flow} ($I$ being an interval) 
\be
t \in I \mapsto (g(t), K(t), \phi(t)) 
\ee
consisting of the time-dependent three-metric~$g(t)$ and extrinsic curvature~$K(t)$ of the hypersurfaces of the foliation, and a matter field~$\phi(t)$.
We assume sufficient regularity and work with functions defined on each side of the singularity hypersurface and blowing up as one approaches it.
We emphasize that no preferred junction condition is introduced in the present paper and, rather, we find it essential to propose a framework that can {\sl  accommodate many different junctions,} 
which we describe via the notion of {\sl singularity scattering map.}
As indicated above, we concentrate on the quiescent regime and on singularity scattering maps that preserve this regime.
We refer to Sections \ref{section---2} and~\ref{section---3} for the terminology (singularity scattering maps, cyclic spacetimes, etc.) and to  \autoref{theo:391} for our explicit construction of spacetimes based on such a singularity scattering map.

\paragraph{Notion of singularity data manifold.} 

Our analysis  
is based the ADM formulation for a foliation of hypersurfaces, together with Fuchsian-type arguments in order to rigorously validate {\sl asymptotic expansions} satisfied by the main unknowns of the problem, that is, the induced metric 
$g(t)$, the extrinsic curvature $K(t)$, and the matter field $\phi(t)$. 
We follow Andersson and Rendall \cite{AnderssonRendall} who treated spacetimes with non-oscillatory singularities of spacelike nature, while we also provide a generalization to timelike hypersurfaces. 
For the huge literature existing on the Fuchsian method in mathematical general relativity, we refer to \cite{Rendall:2008,RendallWeaver} as well as  \cite{AlexakisFournodavlos,AnderssonRendall,BeyerPLF3,Damour-et-al,Fournodavlos:2016,FournodavlosLuk} 
and the references cited therein. 

Considering solutions to the Einstein equations coupled to a scalar field, we begin by neglecting all spatial derivative terms and we solve a simpler system consisting of ordinary differential equations in the (Gaussian) time variable. In turn, this provides us with an explicit Ansatz which we can validate for general solutions in the vicinity of the singularity hypersurface of interest.

In the course of our analysis, we introduce the notion of the {\sl singularity data manifold,} which we denote by $\Ispace$ for spacelike hypersurfaces. The initial value problem is then posed directly on the singularity hypersurface by prescribing a data set $(\gmoi, \Kmoi, \phimoi_0, \phimoi_1) \in \Ispace$ and solving backward in time in order to describe the past of the singularity.
The future of the singularity is likewise solved for in terms of a data set $(\gpoi, \Kpoi, \phipoi_0, \phipoi_1) \in \Ispace$, itself obtained by applying a {\sl singularity scattering map} $\Sbf\colon\Ispace\to\Ispace$ to the prescribed data set $(\gmoi, \Kmoi, \phimoi_0, \phimoi_1)$.
We use a Gaussian foliation (see below) in each regularity domain, based on a proper time function~$t$ normalized such that the singularity hypersurface is at $t = 0$.

\paragraph{Suppression of instabilities.}

We emphasize that our analysis is concerned with those spacetimes that have a non-vanishing matter field near the singularity, so that the oscillating regime identified by Belinsky, Khalatnikov, and Lifshitz~\cite{BKL} in general vacuum spacetimes is beyond the scope of the present paper.
Namely, such oscillations on a singularity generically do not arise in the presence of scalar matter, nor in vacuum spacetimes enjoying some symmetry (or high enough dimensions).

The quiescent regime has recently been shown to be stable in these contexts when evolving towards the singularity in suitable Sobolev spaces~\cite{FournodavlosRonianskiSpeck,RodnianskiSpeck:2018b,RodnianskiSpeck:2018c,Speck:2018}.
Stability of the quiescent regime starting on the singularity has also been understood earlier with analytic regularity in~\cite{AnderssonRendall,Damour-et-al}.
These works establish rigorously the physics expectation developped in~\cite{BK-scalar,Barrow,DHS} that the scalar field removes instabilities of vacuum gravity near spacelike singularities.

\paragraph{Notion of cyclic spacetime.}

The quiescent regime, and our notion of singularity scattering map, apply equally well to spacelike and timelike singularity hypersurfaces.
Generic hypersurfaces feature spacelike and timelike regions, separated by lower-dimensional transitions where the hypersurface becomes null.
This motivates us to define cyclic spacetimes as obeying the Einstein-scalar field equations away from hypersurfaces, and admitting a quiescent expansion subject to our junction conditions along these hypersurfaces except at an exceptional locus of codimension~$2$.
Despite the unconstrained behaviour at the exceptional locus, we find that our notion of cyclic spacetime is sufficiently robust to specify a \emph{unique} global development for generic plane-symmetric collisions of gravitational waves, as we establish in the companion paper~\cite{LLV-3}.

The global existence of solutions with large data is a notoriously difficult endeavour beyond $1+1$ dimensions.
Spacelike singularity hypersurfaces (for the Einstein-scalar field system) are physically understood to be generically of quiescent type away from an exceptional locus of codimension~$2$.
However, it is not clear presently whether there can be null singularity hypersurfaces, or how to extend spacetimes beyond stable null Cauchy horizons exhibited in~\cite{DafermosLuk}.
We postpone to future work the analysis of junction conditions in the null case.
Another important question is to understand whether timelike singularity hypersurfaces, equipped with the junction conditions we define, are stable under the time evolution.
Regardless of the outcome of these investigations, quiescent singularities form an important class of singularities in the Einstein-scalar field system, for which our analysis of junction conditions is crucial.

\subsection{The classification of singularity scattering maps}

\paragraph{Toward a unification of bouncing scenarios.}

Bouncing cosmologies are normally constructed by selecting some particular quantum gravity theory or modification of Einstein gravity and finding spacetimes that are well-described by Einstein gravity on both sides of a bounce, with all corrections being concentrated near the bounce.
This approach starting from an explicit microscopic theory is only completely calculable in {\sl highly symmetric} spacetimes.
Our approach is instead to observe that, regardless of the mechanism causing the bounce, the resulting scattering map must respect Einstein constraints for the asymptotic behavior before and after the bounce in the regimes well-described by Einstein's gravity theory.
While these constraints are trivial in highly symmetric spacetimes, they are {\sl very constraining for scattering maps} that apply to general $3+1$ dimensional spacetimes. This {\sl macroscopic approach to scattering maps} is the avenue that we follow in this paper:
the effect of microscopic physics is entirely encapsulated in a singularity scattering map $\Sbf\colon\Ispace\to\Ispace$.

Our method applies whenever the corrections to general relativity are subleading away from the bounce and locality is preserved during the bounce.
More precisely, we propose to focus on {\sl ultralocal scattering maps,} which stem from bounces in which the evolution at different points in space are independent from each other, in agreement with the well-known BKL analysis on each side of the bounce (see the main text below).

\paragraph{Main statement of this paper.}

The maps of interest are scattering maps for which the values of $(\gpoi,\Kpoi,\phipoi_0,\phipoi_1)$ at a point~$x$ along the singularity hypersurface only depend on $(\gmoi,\Kmoi,\phimoi_0,\phimoi_1)$ at the same point, and not on (spatial) derivatives thereof.  
This strategy allows us to single out two classes of maps:
the {\sl anisotropic} ultralocal scattering maps~$\Sani_{\Phi,c,\epsilon}$
and the {\sl isotropic} ones~$\Siso_{\lambda,\varphi,\epsilon}$, which we describe in detail momentarily.
Remarkably, these two cases exhaust the set of ultralocal scattering maps, as the following theorem states.
In both cases, one easily checks that the shear (traceless part of the extrinsic curvature) $\Kcirc\coloneqq K-\tfrac{1}{3}(\Tr K)\delta$ weighted by the volume form~$\sqrt{|g|}$ is at most multiplied by a constant when traversing the singularity.
We uncover universal, as well as model-dependent, laws (in \eqref{equa-three-laws} below) and we summarize here our main discovery in this paper, as follows. 

\begin{theorem}[Classification of singularity scattering maps in general relativity]\label{thm:intro-classification}
Any ultralocal scattering map is either an anisotropic map $\Sani_{\Phi,c,\epsilon}$ or an isotropic map $\Siso_{\lambda,\varphi,\epsilon}$.
\end{theorem}

Beyond \autoref{thm:intro-classification} classifying ultralocal scattering maps, \autoref{theorem-ultralocal} in the main text further describes several rich subclasses: maps that are quiescence-preserving, invertible, shift-covariant, momentum-preserving, etc.
The relevant restrictions depend on the application and on assumptions on the microscopic physics.
For instance, our local construction of cyclic spacetimes (cf.~\autoref{theo:391})
 involves {\sl quiescence-preserving maps}, which are defined as those that map singularity data with $\Kmoi>0$ to data with $\Kpoi>0$, thus do not generate oscillatory BKL behavior from quiescent behavior.

For our study of colliding gravitational waves in plane-symmetry (cf.~\cite{LLV-3}), it is natural to focus on the natural class of {\sl momentum-preserving maps}, defined by $\Kpoi=\Kmoi$ and $\phipoi_0=\phimoi_0$.
The name ``momentum'' stems from noticing that $(K_{\pm},\phi_{0\pm})$ are normal derivatives of the metric and scalar field at the singularity, while $(g_{\pm},\phi_{1\pm})$ pertain to values of the metric and scalar field.
Momentum-preserving ultralocal scattering maps are determined by a single function $f$ of $\phimoi_0$ and of Kasner exponents (eigenvalues of~$\Kmoi$).  They lead to the junction condition
\bel{intro-mom}
\Kpoi = \Kmoi , \qquad
\phipoi_0 = \phimoi_0 , \qquad
\phipoi_1 = \phimoi_1 + f ,
\ee
with $\gpoi$ given in full in~\cite{LLV-3}.
These maps are particular cases of the anisotropic maps~$\Sani_{\Phi,c,\epsilon}$ described below.
They are manifestly invertible, which is a useful feature for scattering maps that describe timelike singularities because it means data on either side $(g_{\pm},K_{\pm},\phi_{0\pm},\phi_{1\pm})$ is expressible in term of the other singularity data set.
Pleasantly, the maps can also be characterized (up to a sign normalization) by requesting $\Sbf$ and $\Sbf^{-1}$ to be quiescence-preserving and shift-covariant, in the sense that they respect the symmetry of the wave equation under constant shifts of~$\phi$.
By studying the collision of plane-symmetric gravitational waves in~\cite{LLV-3}, 
we discover that the evolution problem imposes an additional {\sl causality condition} on these scattering maps.
The condition expresses that gravitational waves that come out of the singular timelike interface must be determined from the incoming waves on the interface.
It constrains the function~$f$ in such a way that, for example, an identically vanishing $f=0$ is forbidden.

\paragraph{The three laws of bouncing cosmology.} 

In abstracting away all microscopic details of the physical model, we can focus on how solutions to Einstein equations should join across the bounce. Importantly, it turns out that we can distinguish between universal and model-dependent features of junction relations. From our classification we extract three universal laws obeyed by any ultralocal bounce, 
which are independent of the specific physics required in formulating the junction conditions and are summarized as follows. 

\bse
\label{equa-three-laws}

\bei

\item {\bf First law: scaling of Kasner exponents.} With a \emph{dissipation constant} $\gamma\in\RR$, we have
\bel{law1}
|g_+|^{1/2} \Kcirc_+ = - \gamma \, |g_-|^{1/2} \Kcirc_-, 
\ee
which involves the spatial metric~$g$ in synchronous gauge, its volume factor $|g|^{1/2}$, and the traceless part $\Kcirc$~of the extrinsic curvature as a $(1,1)$ tensor.
The isotropic maps~$\Siso_{\lambda,\varphi,\epsilon}$ have $\gamma=0$ while anisotropic maps~$\Sani_{\Phi,c,\epsilon}$ have $\gamma\neq 0$.

\item{\bf Second law: canonical transformation.} 
The massless scalar~$\phi$ undergoes a canonical transformation, as explicited in \autoref{def:canonical-transfo} below:
\bel{law2}
\Phi\colon(\phi_{0-},\phi_{1-}) \mapsto (\phi_{0+},\phi_{1+})
\, 
\text{ \ preserves \ } r(\phi_0)^3\,d\phi_0\wedge d\phi_1
\ee
up to a sign, in which $r(\phi_0)=(1-12\pi\phi_0^2)^{1/2}$.
The \emph{matter map}~$\Phi$ depends in addition on a scalar invariant $\chi\simeq\Tr\Kcirc_-^3/r(\phi_{0-})^3$.

\item{\bf Third law: directional metric scaling.}
The metric after the bounce is a nonlinear rescaling in each proper direction of~$K_-$, specifically 
\bel{law3}
g_{+} = \exp\bigl(\sigma_0 + \sigma_1 K_- + \sigma_2 K_-^2\bigr) g_{-} ,
\ee
in which $\sigma_0,\sigma_1,\sigma_2$ are arbitrary for isotropic scattering maps~$\Siso_{\lambda,\varphi,\epsilon}$ as explicited in~\eqref{Sani} below, and are made explicit (in~\eqref{Sani} below) for anistropic maps~$\Sani_{\Phi,c,\epsilon}$ in terms of $\Phi,\gamma$ for $\gamma\neq 0$.

\eei
\ese
\noindent The three laws are \emph{universal in the renormalization group sense:} they impose constrains on
the macroscopic aspects of all bounces
and apply to different microscopic corrections to Einstein equations.
Contrarily to field theory universality classes, which depend on finitely many parameters, ultralocal singularity scattering maps depend on a whole map, namely~$\Phi$.

\subsection{Organization of this paper}
  
In \autoref{section---1}, 
after introducing in more detail the isotropic and anisotropic scattering maps we explain  
how various physically-motivated bouncing scenarios fit in our framework.
In \autoref{section---2} we begin with the proposed definition of scattering maps for a spacelike singularity hypersurface, and next in \autoref{section---3} we present our general definition of cyclic spacetimes containing both spacelike and timelike singularity hypersurfaces. 
In \autoref{section---4} we establish the classification of all ultralocal scattering maps, while postponing to \autoref{section---5} the technical derivation. 
See also \cite{LLV-2} for a brief overview of our main results, and \cite{LLV-3} for a global construction in the class of plane-symmetric spacetimes. 

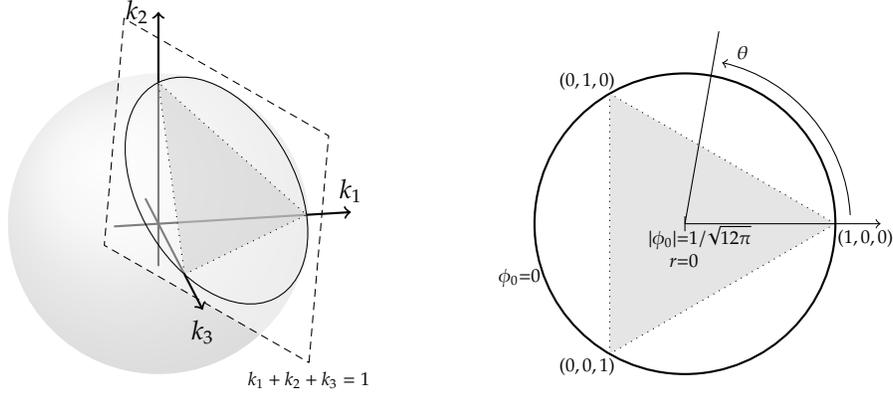
\begin{figure}\centering
  \tdplotsetmaincoords{70}{80}
  \begin{tikzpicture}
    \begin{scope}[tdplot_main_coords,scale = 2]
      \shade[ball color = black!20!white, opacity = 0.2] (0,0,0) circle (1cm);
      \tdplotsetrotatedcoords{45}{atan(sqrt(2))}{0}
      \draw[thick] (-.5,0,0) -- (1,0,0);
      \draw[thick] (0,-.3,0) -- (0,1,0);
      \draw[thick] (0,0,-.3) -- (0,0,1);
      \draw[tdplot_rotated_coords,densely dashed,fill opacity=0.5,fill=white,draw=black,text opacity=1] (1.3,0,{1/sqrt(3)}) node [below]{\scriptsize $k_1+k_2+k_3=1$} -- (0,1.3,{1/sqrt(3)}) -- (-1.3,0,{1/sqrt(3)}) -- (0,-1.3,{1/sqrt(3)}) -- cycle;
      \draw[tdplot_rotated_coords] (0,0,{1/sqrt(3)}) circle ({sqrt(6)/3});
      \draw[thick,->] (1,0,0) -- (1.7,0,0) node[below=0pt]{$k_3$};
      \draw[thick,->] (0,1,0) -- (0,1.3,0) node[above=0pt]{$k_1$};
      \draw[thick,->] (0,0,1) -- (0,0,1.5) node[left=0pt]{$k_2$};
      \draw[dotted,fill opacity=.5,fill=black!20!white] (1,0,0) -- (0,1,0) -- (0,0,1) -- cycle;
    \end{scope}
    \begin{scope}[shift={(7,0)},scale=2]
      \draw[thick] (0,0) circle (1);
      \draw[dotted,fill opacity=.5,fill=black!20!white] (0:1) -- (120:1) -- (-120:1) -- cycle;
      \node at (1.2,-.1) {\scriptsize$(1,0,0)$};
      \node at (-.65,.95) {\scriptsize$(0,1,0)$};
      \node at (-.65,-.95) {\scriptsize$(0,0,1)$};
      \draw[->] (0,0) -- (1.3,0);
      \draw (0,-.05) -- (0,.05);
      \draw (0,0) -- (80:1.3);
      \draw[->] (3:1.1) arc (3:77:1.1) node [pos=.9,above] {\scriptsize $\theta$};
      \node at (.1,-.17) {\scriptsize $\aligned|\phi_0|&{=}1/\!\sqrt{12\pi}\\[-1ex]r&{=}0\endaligned$};
      \node at (-1.1,-.36) {\scriptsize $\phi_0{=}0$};
    \end{scope}
  \end{tikzpicture}
  \caption{\label{fig:Kasner-ball}{\bf Kasner exponents allowed by the constraints.}
    The condition $\Tr K_{\pm}=1$ restricts eigenvalues $k_{1\pm},k_{2\pm},k_{3\pm}$ of $K_{\pm}$ to a plane.
    The Hamiltonian constraint $\Tr K_{\pm}^2=1-8\pi\phi_{0\pm}^2\leq 1$ restricts them to the shaded ball, specifically to a sphere of radius controlled by~$\phi_{0\pm}$.
    The plane and ball intersect along the Kasner disk, depicted on the right, which is conveniently parametrized by the Kasner radius $r(\phi_{0\pm})$ and angle~$\theta$.
    The center of the disk is $K_{\pm}=\tfrac{1}{3}\delta$, for which $\phi_{0\pm}=1/\sqrt{12\pi}$.  More generally, $r(\phi_{0\pm})^2=1-12\pi\phi_{0\pm}^2$.
    We also shade the Kasner triangle, subset of the disk in which all Kasner exponents are positive.
    Its corners are $(k_{1\pm},k_{2\pm},k_{3\pm})=(1,0,0)$ and permutations thereof.}
\end{figure}

\section{Outlook and applications}
\label{section---1}

\subsection{The anisotropic and isotropic maps}

\paragraph{Anisotropic ultralocal scattering.}

\bse\label{intro-Sani}
While omitting a few technical aspects (discussed in full detail later in Example 4; cf.~\eqref{Sani} below), we describe first our anisotropic scattering maps,
written as 
\be
\Sani_{\Phi,c,\epsilon} \colon (\gmoi,\Kmoi,\phimoi_0,\phimoi_1) \mapsto (\gpoi,\Kpoi,\phipoi_0,\phipoi_1). 
\ee
General covariance imposes that $(\phipoi_0,\phipoi_1)$ are functions of scalar invariants of the data, only, and there are a priori {\sl five such invariants} not involving any derivatives which are the matter components $\phimoi_0,\phimoi_1$ and the {\bf three Kasner exponents,} namely the eigenvalues of the extrinsic curvature~$\Kmoi$.  However, the Einstein constraints only allow for Kasner exponents to lie in a circle (with a $\phimoi_0$-dependent radius~$r(\phimoi_0)$), which we parametrize by a {\bf Kasner angle}~$\thetamoi$ as depicted in \autoref{fig:Kasner-ball} (and likewise~$\thetapoi$ for the image data on the other side of the singularity).
Altogether, these fields are described by a map
\be
\Phi\colon(\thetamoi,\phimoi_0,\phimoi_1)\mapsto(\phipoi_0,\phipoi_1) .
\ee

Regarding the extrinsic curvature, our main tool is the asymptotic version of the ADM momentum constraint, which expresses the divergence of~$K_{\pm}$ in terms of the scalar fields~$\phi_{0\pm},\phi_{1\pm}$.
Based on the fact that the scattering map must preserve this momentum constraint, we prove that the extrinsic curvature~$\Kpoi$ depends at most linearly on~$\Kmoi$.
In other words the traceless part of the extrinsic curvature is simply scaled as
\bel{intro-Kcirc-scaling}
(\Kpoi - \tfrac{1}{3}\delta)
= \epsilon\,\Omega(\phipoi_0,\phimoi_0) \,(\Kmoi - \tfrac{1}{3}\delta) ,
\ee
for some sign $\epsilon=\pm 1$ and a conformal factor $\Omega(\phipoi_0,\phimoi_0)=r(\phipoi_0)/r(\phimoi_0)$ determined by radii of the circles on which Kasner exponents lie.
The Kasner exponents (minus their average $1/3$) are scaled by a positive coefficient ($\epsilon=+1$) or negative coefficient ($\epsilon=-1$) that depends on $\phimoi_0$ and~$\phipoi_0$, and additionally the corresponding {\sl eigenvectors of the extrinsic curvature} are preserved.
Returning to our parametrization of Kasner exponents we learn that the scattering map either preserves the Kasner angle or shifts it by~$\pi$, that is, 
\be
\thetapoi = \thetamoi \quad \text{ if } \epsilon=+1,
\qquad \quad 
\thetapoi = \thetamoi+\pi \quad \text{ if } \epsilon=-1.
\ee
We then prove that $\Omega$ is a constant multiple of $\sqrt{\gmoi}/\sqrt{\gpoi}$, hence is identically vanishing (which leads to the isotropic maps $\Siso$ discussed next) or nowhere vanishing (which leads to anisotropic maps~$\Sani$ discussed presently).
In the anisotropic case $\Omega\neq 0$ we determine that the metric is scaled differently in each eigenspace of~$\Kmoi$ (or~$\Kpoi$) and reads 
\bel{intro-Saniso-g}
\gpoi = c^2 \Omega^{-2/3} \exp\Big(
16\pi\epsilon\xi\cos\Theta_-
- 16\pi\epsilon\bigl( \del_{\thetamoi} \xi + \frac{\phipoi_0}{r(\phipoi_0)} \del_{\thetamoi} \phipoi_1 \bigr) \sin\Theta_-
\Big) \, \gmoi. 
\ee
Here, $c>0$ is a constant parameter, the tensor $\Theta_-=\diag(\theta_-,\theta_-+2\pi/3,\theta_-+4\pi/3)$ is such that $\Kmoi=\frac{1}{3}\delta+\frac{2}{3}r(\phimoi_0)\cos\Theta_-$, while $\xi$~is an auxiliary function of $\thetamoi,\phimoi_0,\phimoi_1$ given explicitly as an integral formula in terms of~$\Phi$.
Finally, we prove that for each angle~$\thetamoi$, the map $\Phi(\thetamoi,\,.\,,\,.\,)$ is a {\sl canonical transformation} for a measure~\eqref{phase-space-measure} defined on the phase space of all matter data $(\phimoi_0,\phimoi_1)$, and prove suitable boundary conditions on~$\Phi$.

We emphasize the following features of the map $\Sani_{\Phi,c,\epsilon}$. 
\bei 

\item The singularity scattering map depends essentially on the prescription of a {\sl single scalar function} $\Phi$.

\item This function $\Phi = (\thetamoi,\phimoi_0,\phimoi_1)$ {\sl depends upon the Kasner angle and matter field} before the bounce, only, and can be chosen (almost) arbitrarily. 

\item The (trace-free part of) extrinsic curvature is {\sl conformally transformed}, by a conformal factor that is explicit in terms of $\Phi$.
In fact, the densitized trace-free extrinsic curvature $(K_{\pm}-\tfrac{1}{3}\delta)\sqrt{g_{\pm}}$ is unchanged up to a constant factor~$\epsilon c$.

\item The metric is rescaled anisotropically, differently along each eigenvector of~$K_{\pm}$; indeed, we stress that $\Theta_-$~is a {\sl matrix}.

\eei 
\ese
The aforementioned momentum-preserving maps~\eqref{intro-mom} correspond to the case with $\epsilon=+1$ and $\Phi=(\phimoi_0,\phimoi_1+f(\thetamoi,\phimoi_0))$.

\paragraph{Isotropic ultralocal scattering.}

The second class of ultralocal scattering maps we discover is obtained by taking $\Omega=0$ in~\eqref{intro-Kcirc-scaling}, hence $\Kpoi=\tfrac{1}{3}\delta$ and $r(\phipoi_0)=0$, which fixes $\phipoi_0$ up to a sign.  The momentum constraint then forces the scalar field~$\phipoi_1$ to be constant, while the metric is arbitrary.
The isotropic scattering map is written (in the spacelike case) as
\be
\Siso_{\lambda,\varphi,\epsilon} \colon (\gmoi,\Kmoi,\phimoi_0,\phimoi_1) \mapsto (\gpoi,\Kpoi,\phipoi_0,\phipoi_1) 
= 
\Big( \lambda(\Theta_-,\phimoi_0,\phimoi_1)^2 \gmoi, \ \frac{1}{3}\delta, \ {\epsilon}/{\sqrt{12\pi}}, \ \varphi \Big) 
\ee
for any constant $\varphi\in\RR$, any sign $\epsilon=\pm 1$, and any function $\lambda = \lambda(\theta_-,\phimoi_0,\phimoi_1)$ that is positive, $2\pi$-periodic and even in~$\theta_-$, and obeys suitable boundary conditions in~$\phimoi_0$.
Here, the tensor~$\Theta_-$ is as defined below~\eqref{intro-Saniso-g} and $\lambda$ is applied to each of its (diagonal) entries independently.
The sign $\epsilon=\pm 1$ and the constant $\varphi\in\RR$ can be normalized away using symmetries of the wave equation for~$\phi$ away from the singularity.
At first sight, $\Siso$ is obtained as a degenerate case of the anisotropic maps~$\Sani$ above: take $\Phi$~to be a constant map, specifically $\phipoi_0=\epsilon/\sqrt{12\pi}$ and $\phipoi_1=\varphi$, so that $\Kpoi=\tfrac{1}{3}\delta$.
However, these limits of~$\Sani$ do not give rise to the most general choice of function~$\lambda$.
The metric is less constrained in the isotropic case than the anisotropic case because obeying the momentum constraint is trivial in the isotropic case.

After the bounce under the map $\Siso_{\lambda,\varphi,\epsilon}$, we have the following features. 
\bei 

\item The scattering map depends essentially on the prescription of a {\sl single scalar function} $\lambda$.

\item This function $\lambda = \lambda(\theta_-,\phimoi_0,\phimoi_1)$ {\sl depends upon the Kasner angle and matter field} before the bounce, only, and can be chosen (almost) arbitrarily. 

\item The metric is rescaled differently by the bounce along the different eigenvectors of the extrinsic curvature.
 
\item The extrinsic curvature is a constant multiple of the identity, leading to an {\sl isotropic and homogeneous evolution} after the bounce: (the asymptotic profile of) the metric after the bounce simply entails a time-dependent conformal factor that is constant along leaves of the foliation.

\item The two components of the matter field after the bounce are overall constants. 

\eei

\paragraph{Vacuum case.}

Our scattering maps are defined for any values of the data compatible with Einstein constraints, in particular in regions of spacetime that may be vacuum.  To avoid creation of matter by the scattering, one may want to impose $\Phi(\thetamoi,0,0) = (0,0)$.
In that case, and restricting them to vacuum data, only, the scattering maps we define above reduce to (with $\epsilon=\pm 1$ and $c>0$ constant and $\xi$ an essentially arbitrary periodic function of~$\thetamoi$)
\[
\Kpoicirc = \epsilon \Kmoicirc , \qquad
\gpoi = c^2 \exp\bigl(16\pi\epsilon(\xi\cos\Theta_- - (\del_{\thetamoi}\xi)\sin\Theta_-)\bigr)\gmoi .
\]
It would be interesting to determine more generally what scattering maps exist in vacuum, without the restriction that the maps be defined in the presence of scalar fields as well.
While in vacuum the ultralocal scattering maps are likely much simpler than our classification \autoref{thm:intro-classification}, solutions to the Einstein equations may involve BKL oscillations that are not directly covered by our analysis. Furthermore, our classification method should also apply to spacetimes containing stiff fluids, an important class of spacetimes in order to deal with ultra-dense matter that can appear in cosmology; see Zel'dovich~\cite{Zeldovich}.

\subsection{Applications: collisions, string theory, and loop quantum cosmology}
\label{ssec:applications}

\paragraph{Microscopic versus macroscopic approach.}

Geometric singularities in solutions to Einstein equations suggest that general relativity should receive corrections in regions with high curvature, so as to avoid singularities.
In particular, various cosmological models exist where the Big Bang is replaced by a singular or non-singular bounce, achieved for example through quantum gravity effects, a modification of the Einstein--Hilbert action, or simply matter violating the null energy condition.
Our macroscopic approach abstracts away details of the bounce by approximating both sides as a solution of general relativity and, from the Einstein constraints, deducing strong a priori restrictions on possible bounces regardless of microscopic details.

The microscopic approaches are mostly studied in the cosmological literature for very symmetrical spacetimes such as Bianchi (homogeneous) spacetimes, and perturbations thereof, for which calculations are analytically tractable.  For our approach, in contrast, it is essential to consider general spacetimes, in which preserving Einstein constraints is a very restrictive condition on scattering maps.

Our method applies whenever a microscopic theory produces bounces that are well-described by the BKL solutions to Einstein equations on both sides of the bounce and whose behavior is dominated by time derivatives rather than spatial gradients: such a bounce must be described by one of our ultralocal scattering map, which depends solely on the chosen microscopic theory and not on details of the bounce.
The relevant singularity scattering map can be identified simply by working out bounces in Bianchi~I (homogeneous but anisotropic) spacetimes.
Our scattering map approach then predicts features of bounces in {\sl arbitrarily inhomogeneous spacetimes}.
After validating these predictions (hence the ultralocality assumption) in simplified setups where first principles microscopic derivations are possible, such as linearized perturbations around Bianchi~I spacetimes, one can start applying our general tools to learn about cosmological features after bounces with arbitrary inhomogeneities in the chosen microscopic theory.

\paragraph{Pre-Big Bang scenario in string cosmology.} 

Let us outline the situation for the pre-Big Bang scenario in a spatially homogeneous setting, ignoring various constants and postponing a more detailed analysis to later work.
We keep the dimension~$d$ of spatial slices unspecified in this paragraph, to ease comparison with available literature.  The reader can substitute $d=3$ to match the rest of this paper.

In the string frame~(SF), the homogeneous metric-dilaton equations of motion (at tree level and truncated to the lowest order in derivatives) admit Bianchi~I solutions of the form: 
\[
g_{\SF} = - dt_{\SF}^2 + \sum_{i=1}^d |t_{\SF}|^{2\beta_i}dx^idx^i , \qquad
\phi_{\SF} = (\Sigma -1) \log|t_{\SF}| , \qquad
\text{with} ~ \Sigma = \sum_{i=1}^d \beta_i, ~\quad \sum_{i=1}^d \beta_i^2 = 1 .
\]
Thus, any given solution (i.e.\ any given choice of the $\beta_i$) belongs to a set  of $2^{d+1}$ choices, corresponding to the possibility of flipping the sign of $t_{\SF}$ as well as the one of any $\beta_i$. This possibility is guaranteed by a symmetry (scale-factor duality~\cite{VenezianoSFD}) of the string-cosmology equations in the presence of $d$ abelian isometries.
The idea of the pre-Big Bang scenario~\cite{VenezianoSFD,GasperiniVeneziano1} is to combine, in a single cosmology valid from $t_{\SF} = - \infty$ to $t_{\SF}= + \infty$, two solutions in this set that differ for both the sign of $t_{\SF}$ and for that of each $\beta_i$, so that each Hubble parameter $\beta_i/t_{\SF}$ does not change sign from  $t_{\SF} < 0$ to $t_{\SF}  >0$.
Each solution becomes singular at $t_{\SF}=0$ but it is conjectured that higher derivative and/or higher loop corrections will remove the singularity and allow for a smooth joining of the two solutions.

In the present context we then write, for all $t_{\SF} \ne0$,
\[
g_{\SF} = - dt_{\SF}^2 + \sum_{i=1}^d |t_{\SF}|^{2\beta_{i\pm}}dx^idx^i , \qquad
\phi_{\SF} = (\Sigma_{\pm}-1) \log|t_{\SF}| , \qquad
\text{with} \quad \Sigma_{\pm} = \sum_{i=1}^d \beta_{i\pm} , \quad \sum_{i=1}^d \beta_{i\pm}^2 = 1 ,
\]
where the subscripts~$\pm$ are the sign of~$t_{\SF}$. 
This sign distinguishes two sides of the bounce. As mentioned above, a solution for $t_{\SF}<0$ with some values of the exponents~$\beta_{i-}$ is joined to a solution with $t_{\SF}>0$ with all $\beta_{i+}= -\beta_{i-}$, hence $\Sigma_+=-\Sigma_-$.

The Einstein-frame metric is $\exp(-2\phi_{\SF}/(d-1))g_{\SF}$, and the corresponding proper time coordinate~$t$ (vanishing at the bounce) is
\[
t = \pm \frac{d-1}{d-\Sigma_{\pm}} |t_{\SF}|^{(d-\Sigma_{\pm})/(d-1)} \quad
\text{for } \pm t_{\SF} > 0 .
\]
The Einstein-frame metric~$g$ then takes the form
\bel{intro-PBB-gEF}
g = - dt^2 + \sum_{i=1}^d g_{\pm\,ii}|t|^{2k_{i\pm}}dx^idx^i , \qquad
\text{with} \quad k_{i\pm} = \frac{1}{d} + \frac{d-1}{d-\Sigma_{\pm}} \biggl(\beta_{i\pm} - \frac{\Sigma_{\pm}}{d}\biggr) , \quad
g_{\pm\,ii} = \biggl(\frac{d-\Sigma_{\pm}}{d-1}\biggr)^{k_{i\pm}} .
\ee
We wrote the Kasner exponents $k_{i\pm}$ in a form that makes manifest that $\sum_ik_{i\pm}=1$, since $\Sigma_{\pm}/d$ is the average of the~$\beta_{i\pm}$.
In addition, we readily translate the junction condition $\beta_{i+}=-\beta_{i-}$ (and $\Sigma_+=-\Sigma_-$) to a rescaling of all shears $k_{i\pm}-1/d$ and of the volume factor by inverse amounts:
\[
k_{i+} - \frac{1}{d} = - \frac{d-\Sigma_-}{d+\Sigma_-} \biggl( k_{i-} - \frac{1}{d} \biggr) , \qquad
\sqrt{|\gpoi|} = \frac{d+\Sigma_-}{d-\Sigma_-} \sqrt{|\gmoi|} .
\]
This is precisely as predicted by the first law~\eqref{law1} of ultralocal scattering maps, suggesting that the pre-Big Bang scenario bounce is described by one of our maps.
If so, the map must be an anisotropic map $\Sani_{\Phi,c,\epsilon}$ with $c=1$ and $\epsilon=-1$, because $(k_{i\pm} - \frac{1}{d})\sqrt{|\gpoi|}$ simply changes sign.

One can in principle determine~$\Phi$ by studying how the Einstein-frame canonically normalized dilaton jumps.  Its leading coefficient $\phi_{0\pm}$ (in units where Newton's constant is $G=1$) is given by
\[
\phi = \phi_{0\pm} \log|t| + O(1) , \qquad
|\phi_{0\pm}| = \sqrt{\frac{d-1}{8\pi}} \, \frac{\Sigma_{\pm}-1}{d-\Sigma_{\pm}} ,
\]
and one easily checks $\sum_ik_{i\pm}^2=1-8\pi\phi_{0\pm}^2$.
Clearly, $\Sigma_+=-\Sigma_-$ allows us to express $\phipoi_0$ as a function of~$\phimoi_0$, only, and not of individual Kasner exponents.
In the language of~\eqref{intro-Sani} this means that $\phipoi_0$ does not depend on the Kasner angle (or angles, in dimension $d>3$).
More precise calculations suggest that $\phipoi_1$ also does not depend on these angles, so that our expression of the metric~\eqref{intro-Saniso-g} simplifies to an expression of the form
\[
g_+ = \Omega^{-2/d} \exp(\lambda(K_--1/d))g_-
\]
where $\lambda$ may a priori depend on~$\phimoi_0,\phimoi_1$.
This is consistent with the junction condition we found on~$g$ in~\eqref{intro-PBB-gEF}, with
$\lambda = \log\frac{d-1}{d-\Sigma_-} + \frac{d-\Sigma_-}{d+\Sigma_-}\log\frac{d-1}{d+\Sigma_-}$.
Note that, in principle, one can try to construct alternative bouncing cosmologies by matching, across the singularity, any two of the $2^d$ duality-related Kasner cosmologies. It is easy to check, however, that in other cases our junction conditions are not satisfied: thus, the only bounce consistent with ultralocality is the one where all $\beta_+=-\beta_-$. 
A specific example of this in the context of the plane-symmetric case is presented in~\cite{LLV-3}.

\paragraph{Modified gravity theories.}

Bounces were also considered in a class of modified gravity theories including metric or Palatini $f(R)$ gravity, Brans-Dicke theory, and more general scalar-tensor theory in~\cite{Cesare-bounce}.
The set of gravity theories under consideration is too general to obtain a specific scattering map.
Nevertheless, the densitized trace-free extrinsic curvature $\sqrt{|g|}(K-\tfrac{1}{3}\Tr(K)\delta)$ was shown in this setting to remain conserved throughout the bounce, hence to be the same on both sides of the bounce.
In terms of the singularity scattering data $(g_{\pm},K_{\pm})$ this gives\footnote{The extrinsic curvatures $K_{\pm}$ are defined with respect to unit normals pointing away from the singularity, while in a smooth bounce one more naturally works with the normals pointing in the same direction on both sides of the bounce.  This leads to a sign in the scattering map from $\Kmoi$ to~$\Kpoi$.} $\sqrt{|\gpoi|}\,\Kpoicirc=-\sqrt{|\gmoi|}\,\Kmoicirc$, which is consistent with the first law~\eqref{law1} above.
Combined with our classification of ultralocal scattering maps (in \autoref{theorem-ultralocal}), this suggests that bounces in rather general modified gravity theories are governed by an anisotropic scattering map of the form~$\Sani_{\Phi,1,-}$, as in the pre-Big Bang scenario. It would be interesting to extend our arguments to 
yet further models of gravity such as the one studied mathematically in \cite{LuzMena}. 

\paragraph{Loop quantum cosmology.}

Loop quantum cosmology following the Ashtekar school~\cite{AshtekarSingh,AshtekarPawlowskiSingh} leads quite generically to cosmological bounces. A different standpoint by Bojowald~\cite{Bojowald:2020wuc} was analyzed and opposed in \cite{Ashtekar:2009,CorichiSingh,KaminskiPawloski}. 
 
In loop quantum cosmology~\cite{AshtekarWilsonEwing,Wilson-Ewing-LQC}, and in some classical gravity theories such as limiting curvature mimetic gravity~\cite{Chamseddine}, the junction condition for the extrinsic curvature in a Bianchi I bouncing spacetime with a stiff fluid or massless scalar field is $\Kpoi=\frac{2}{3}\delta-\Kmoi$.
Assuming that bounces in these modifications of general relativity respect the ultralocality expected from the BKL analysis, they must be described by a scattering map listed in our classification in \autoref{thm:intro-classification} above.
As we explicitize near~\eqref{scattering-map-antipodal-K} below, the only scattering maps that give rise to this sign flip $\Kpoicirc=-\Kmoicirc$ are $\Sani_{\Phi,c,\epsilon}$ with $\Phi(\thetamoi,\phimoi_0,\phimoi_1)=\pm(-\phimoi_0,f(\thetamoi,\phimoi_0)+\phimoi_1)$ and $\epsilon=-1$.
These maps are parametrized by a single function $f\colon \RR\times I_0\to\RR$ (periodic in~$\theta$) and an unimportant constant $c>0$ and sign~$\pm$.
We call these maps momentum-reversing, in analogy to the momentum-preserving case that we discussed above.

In this way, our method provides an explicit form of the scattering map applicable to general spacetimes, starting only from the map of Kasner exponents in a homogeneous spacetime.
It would be interesting to test our assumption of ultralocality by checking whether the scattering map~\eqref{scattering-map-antipodal-K} (see below) is compatible with results in loop quantum cosmology with Gowdy symmetry~\cite{Brizuela:2009nk} or with linearized perturbations around homogeneous spacetimes in limiting curvature mimetic gravity.

\paragraph{Further generalizations.}

\bei 

\item{\bf Bounces with no classical description.} In some other quantum gravity approaches such as quantum reduced loop gravity~\cite{Alesci:2019sni}, the solutions do not admit a classical description after the bounce, which makes our techniques inapplicable.

\item{\bf On non-ultralocal scattering maps.} More generally, we could also consider singularity scattering maps that are not ultralocal, namely for which the values of $(\gpoi,\Kpoi,\phipoi_0,\phipoi_1)$ at $x\in\Hcal$ can depend on values of $(\gmoi,\Kmoi,\phimoi_0,\phimoi_1)$ and their derivatives at that point.  While in principle an approach similar to the one we take in the ultralocal case might lead to a classification of singularity scattering maps involving derivatives of a given order, the calculations appear intractable.

\eei 


\section{Spacelike singularity hypersurfaces in \((3+1)\)-dimensional spacetimes} 
\label{section---2}

\subsection{The \(3+1\) ADM formulation} 

\paragraph{Gaussian foliation.}

We describe here the geometry near a spacelike singularity hypersurface~$\Hcal_0$.
In the following we shall make use of a {\sl local} Gaussian foliation emanating from the singular hypersurface and constructed as follows.
Geodesics normal to the hypersurface~$\Hcal_0$ cover a neighborhood of that hypersurface, so that a time coordinate~$s$ can be defined as the proper time along such geodesics, with $s=0$ at~$\Hcal_0$.
Level sets of~$s$ form a local spacelike foliation of spacetime
\[
\Mcal^{(4)} = \bigcup_{s \in [s_{-1}, s_1]} \Hcal_s, 
\]
by a time coordinate denoted by $s\colon \Mcal^{(4)} \mapsto [s_{-1}, s_1]$
for two parameters $s_{-1} < 0 < s_1$,
consisting of a past region $s\in[s_{-1}, 0)$ and a future region $s\in (0, s_1]$.
These two regions are pasted at $s=0$ along a spacelike singularity hypersurface $\Hcal_0$ on which curvature invariants may blow up.
Each slice $\Hcal_s$, $s\neq 0$ is endowed with a Riemannian metric $g(s) = (g_{ab}(s))$ and an extrinsic curvature tensor (or second fundamental form) $K(s) =  (K_a^b(s))$. Here, both tensor fields are symmetric, thus $g_{ab}=g_{ba}$ and $K_{ab} = K_{ba}$ where, as usual, indices are lowered (or raised) with the metric~$g$.  In our notation, local coordinate indices are written with Latin letters $a,b, \ldots = 1,2,3$. 
The trace $\Tr(K) = K_b^b = g^{ab} K_{ab}$ represents the mean curvature of the slices within the spacetime and, in our setup, blows up at $s=0$.

Locally, in addition to defining a proper time coordinate~$s$, the geodesics emanating from~$\Hcal_0$ and normal to it provide a diffeomorphism from each leaf~$\Hcal_s$ to~$\Hcal_0$.
The shift vector is then identically~$0$, and the lapse function is identically~$1$ by construction, so that the foliation is a {\bf Gaussian foliation}.
Then the four-dimensional metric in $(\Mcal^{(4)}, g^{(4)})$ is expressed in terms of the three-dimensional one as\footnote{This gauge choice $g_{00}=-1$ and $g_{0a}=0$ is also called synchronous gauge, but we avoid this terminology, as it is not applicable to the case of timelike foliations we consider later on.  In the ADM formalism the gauge choice sets the lapse to~$1$ and the shift to~$0$.  Such a choice of coordinates can only be made locally, as there are typically obstructions to the existence of a {\sl global} synchronous gauge coordinate system.  Note additionally that the synchronous gauge (Gaussian foliation) does not guarantee a simultaneous singularity, but that one can {\sl choose} to set up the foliation starting from the singularity hypersurface (as we do) to ensure that the singularity indeed happens simultaneously at $s=0$.}
\[
g^{(4)} = \big(g^{(4)}_{\alpha\beta}\big)
= - ds^2 + g(s), \quad \text{ with } g(s) = g_{ab}(s) dx^a dx^b. 
\]
Here, Greek indices $\alpha, \beta, \ldots$ range from $0$ to $3$, while for Latin indices we take $a,b, \ldots= 1,2,3$. 
We sometimes call $s$ a {\bf Gaussian time coordinate}.
Here and throughout this paper, we use Greek indices for spacetime indices $\alpha, \beta = 0,1,2,3$. 
In such a foliation, $K_{ab} = -(1/2) \del_s g_{ab}$.

\paragraph{Gravitational field equations.}  

The metric and extrinsic curvature tensor fields are assumed to satisfy the ADM (Arnowitt-Deser-Misner) first-order formulation of {\bf Einstein's evolution equations}, i.e. 
\bel{equa:ADMsysytem}
\del_s g_{ab}  = - 2 \, K_{ab},
\qquad
\del_s K^a_b - \Tr(K)  K^a_b =  R^a_b - 8 \pi \, M^a_b, 
\qquad M^a_b =  \frac{1}{2} \rho \delta^a_b + \big( T^a_b - \frac{1}{2} \Tr(T)  \delta^a_b\big). 
\ee
Here, $R^a_b$ denotes the (intrinsic, $3$-dimensional) Ricci curvature of the slices, while the mass-energy density 
$\rho = T^{(4)}_{00} = T^{(4)}(n,n)$, the momentum vector $J = -T^{(4)}_{0\,\smallbullet} = -T^{(4)}(n, \smallbullet\,)$ and the stress tensor $T =(T^a_b)$ are components of the spacetime energy-momentum tensor $T^{(4)}_{\alpha\beta}$ specified below, where $n$ is the future-oriented, unit normal to the foliation.

In addition, the equations \eqref{equa:ADMsysytem} are supplemented with {\bf Einstein's constraint equations}
\bel{equa-constr}
R + (\Tr K)^2 - \Tr (K^2) = 16 \pi \rho,
\qquad 
\nabla_a K^a_b  - \del_b (\Tr K) = 8 \pi J_b,
\ee
in which $R = R_b^b$ denotes the trace of the Ricci tensor. These latter two equations are referred to as the Hamiltonian and momentum equations, respectively, and provide one with a restriction of the initial data set that can be prescribed (on any given regularity hypersurface, say). In the regions $s<0$ and $s>0$ of regularity, it is well-known that they hold on any hypersurface $\Hcal_s$ provided they hold on any other one.

\paragraph{Coupling with the matter field.}

The right-hand sides of the equations \eqref{equa:ADMsysytem}--\eqref{equa-constr} contain contributions whose explicit expression requires a modeling assumption about the matter content of our spacetime. Here, we work with a massless scalar field $\phi$ whose energy-momentum tensor is quadratic in the first-order derivatives of $\phi$, namely  
\[
T^{(4)}_{\alpha\beta} \coloneqq \del_\alpha \phi \del_\beta \phi 
- \frac{1}{2} \Big(g^{(4)\gamma\delta} \del_\gamma \phi \del_\delta \phi  
\Big) g^{(4)}_{\alpha\beta} .
\]
After projection on the slices of the foliation, the matter components are found to read 
\bel{equa:rhoandJ}
\rho = \frac{1}{2} \Big( (\del_s \phi)^2 + | d\phi|^2_g \big), 
\qquad
J = - \del_s \phi \, d \phi, 
\qquad
T = d\phi \otimes d\phi + \frac{1}{2} \Big( (\del_s \phi)^2 - | d\phi|^2_g \big) g. 
\ee
By virtue of the Euler equations $\nabla^{(4)}_\alpha T^{(4)\alpha}_\beta = 0$, where $\nabla^{(4)}$ is the connection associated with the spacetime metric, the field $\phi$ is determined by solving the wave equation $\nabla^{(4)}_\alpha \nabla^{(4)\alpha} \phi = 0$, that is, the {\bf matter evolution equation} 
\bel{equa-wavephi}
- \del_s^2 \phi + \Tr(K) \, \del_s \phi + \Delta_g \phi = 0
\ee 
with $\Delta_g \phi  = \nabla_b \nabla^b \phi$. 
This is a linear wave equation which, of course, is coupled to \eqref{equa:ADMsysytem}. For this matter model, the prescription of two scalar fields, that is, the restrictions of $\phi$ and $\del_s \phi$, are required as part of the initial data set on a (regularity) hypersurface. 
Furthermore, we emphasize that the term involving $\Tr(K)$ accounts for the expanding or contracting nature of the spacetime.

\paragraph{Local Cauchy developments from regularity hypersurfaces.} 

It is a standard matter than the system \eqref{equa:ADMsysytem}--\eqref{equa-wavephi} admits a unique local-in-time solution defined on an interval $[s_{-1}, s_0)$, provided a sufficiently regular initial data set\footnote{Here $\Lcal$ denotes the Lie derivative operator.}
$\big(g(s_{-1}), K(s_{-1}), \phi(s_{-1}), \Lcal_n \phi(s_{-1}) \big)$ is prescribed on a regularity hypersurface $\Hcal_{s_{-1}}$ and $s_0$ is sufficiently close to $s_{-1}$.  In general, a solution initiating at $s= s_{-1}$ may not exist over a sufficient long time interval and may not reach the singularity hypersurface. An alternative and more natural approach, which we investigate in the rest of this section, consists of prescribing data directly on the singularity hypersurface and evolving {\sl away} from it.

\subsection{Singularity data and asymptotic profile} 

\paragraph{BKL behaviors of quiescent or oscillating type.}

The BKL conjecture~\cite{BKL} describes how, near a spacelike singularity, the evolution at different points in space generically decouples. Depending on dimensionality and on the matter content, one expects two possible regimes~\cite{BKL}:

\bei

\item The {\bf quiescent regime} (studied by Barrow and others~\cite{BK-scalar,Barrow,DHS} and which is of main interest to the present study) where the metric is close to a Bianchi~I metric (with well-defined Kasner exponents) {\sl at each point} near the singularity hypersurface~$\Hcal_0$ (as we describe below). 

\item The {\bf oscillating regime,} where the spacetime has successive epochs each being described by a Bianchi~I metric {\sl at each point}, separated by rapid transitions during which the Kasner exponents and the directions transform non-trivially.
\eei
In our setting with a massless scalar field in $3+1$ dimensions there are generically no oscillation and the metric can be approximated by a Bianchi~I metric at each point of the singularity hypersurface.  More precisely, in our existence theory (cf.\ \autoref{theo:391}) dealing with solutions with a prescribed asymptotic behavior on the singularity, we are able to treat {\bf quiescent bounces,} obtained when the second fundamental form has a definite sign in the sense that all of the Kasner exponents are positive.
Furthermore, in each of the two generic regimes above, singularities may additionally feature spikes~\cite{RendallWeaver} in co-dimension~$1$.
This motivates us, later on in this text, to work away from a two-dimensional exceptional locus.

\paragraph{Evolution equations for the asymptotic profile.}

We consider first the time interval $s \in [s_{-1}, 0)$ and we investigate the behavior of the solutions $(g, K, \phi)$ to the coupled system \eqref{equa:ADMsysytem}--\eqref{equa-wavephi}, as $s \to 0$.  We seek an {\bf asymptotic profile} denoted by $(\gst, \Kst, \phist)$ that accurately approximates a general solution as one approaches the singularity. Such an asymptotic profile (cf.~the review in Rendall's textbook \cite{Rendall:2008}) should be determined by solving the so-called {\bf velocity-dominated evolution equations\footnote{The terminology ``velocity dominated'' refers to the fact that time-differentiated terms (interpreted as ``velocity'' terms) are dominant.}}, obtained by removing all spatial derivatives in the evolution equations, as follows. 

Namely, from the evolution equations \eqref{equa:ADMsysytem} and \eqref{equa-wavephi} we formally deduce the following equations with unknowns $\gst, \Kst, \phist$, respectively, 
\begin{subequations}
\label{ADM-zero}
\bel{ADM-zero-a}
\del_s \gst_{ab} = - 2 \, \Kst_{ab},
\ee
\bel{ADM-zero-b}
\del_s \Kst_a^b - \Tr(\Kst) \Kst_a^b = 0, 
\ee
\bel{ADM-zero-c} 
\del_s^2 \phist - \Tr(\Kst) \del_s \phist = 0.
\ee
\end{subequations}
The system can be solved explicitly, as follows. 
\bei 

\item By taking the trace of \eqref{ADM-zero-b}, we find $\del_s \Tr(\Kst) = (\Tr\Kst)^2$ and, provided we normalize the singularity to take place at the time $s=0$, it follows that 
\[
\Tr(\Kst)(s) = - \frac{1}{s}, 
\]
so that this asymptotic profile consists of a CMC (constant mean-curvature) foliation. 

\item Consequently, the same equation in \eqref{ADM-zero-b} tells us that $(-s) \, \Kst_a^b$ is a constant in time, which we denote by $\Kmoi_a^b$. Hence, we find (with the spatial variable $x$ describing $\Hcal_s \simeq \Hcal_0$): 
\[
\Kst_a^b(s, x) = \frac{-1}{s} \Kmoi_a^b(x),
\qquad 
\Tr(\Kmoi(x)) = 1, 
\qquad x \in \Hcal_0. 
\]

\item Next, the metric equation \eqref{ADM-zero-a} reads 
$s \, \del_s \gst_{ab} = 2 \, \Kmoi_a^c \, \gst_{cb}$ and leads us to  
\[
\gst_{ab}(s,x) = \bigl( |s|^{2\Kmoi(x)} \bigr)_a^c \, \gmoi_{cb}(x), 
\qquad x \in \Hcal_0, 
\]
in which the two-tensor $|s|^{2\Kmoi} = e^{2\Kmoi\log |s|}$ is defined by exponentiation. 

\item Finally, from the matter equation \eqref{ADM-zero-c} we obtain
\[
\phist(s,x) = \phimoi_0(x) \log|s| + \phimoi_1(x), 
\qquad x \in \Hcal_0, 
\]
in which the fields $\phimoi_0, \phimoi_1$ are arbitrary. 
\eei 

As we will observe in the proof of \autoref{theo:391}, the asymptotic system~\eqref{ADM-zero} is a controlled approximation of the Einstein-scalar field equations if $\Kmoi$~is positive definite.  Beyond this so-called quiescent regime, the asymptotic profile is generically unstable, with a well-understood transition~\cite{BKL} to another value of the exponents~$\Kmoi$.  Providing the definitions for general exponents remains useful nevertheless, because non-quiescent singularities are stable in certain symmetry classes, for instance the plane-symmetric spacetimes that we explore in~\cite{LLV-3}.

Altogether, an asymptotic profile is uniquely determined from the prescription, on the singularity hypersurface $\Hcal_0$, 
of an arbitrary Riemannian metric $\gmoi_{ab}$ and a symmetric $2$-tensor field $\Kmoi_{ab}$ satisfying $\Tr\Kmoi = 1$, 
together with two scalar fields $\phimoi_0, \phimoi_1$.  We observe that the condition $\Tr\Kmoi = 1$ implies that the determinant $\abs{\gst}$ of $\gst_{ab}$ is proportional to $s^2$ and, more precisely, 
\[
\sqrt{\abs{\gst(s, x)}} = \abs{s} \sqrt{\abs{\gmoi(x)}}.  
\]
We also observe that the asymptotic profile can be extended to $s\in(-\infty,0)$ using the same formulas, and that the data $(\gmoi,\Kmoi,\phimoi_0,\phimoi_1)$ coincide with the asymptotic profile $(\gst,\Kst,s\del_s\phist,\phist)$ at $s=-1$.

It is important to check that $\gst$ and $\Kst$ have the desired symmetry provided $\gmoi_{ab} = \gmoi_{ba}$ and $\Kmoi_a^b \gmoi_{bc} = \Kmoi_c^b \gmoi_{ba}$.  By applying the second identity $n$ times one easily checks that $(K_-^n)_a^b\gmoi_{bc}$ is symmetric too, thus for any entire function~$f$ we have that $f(\Kmoi)_a^b\gmoi_{bc}$ is symmetric.  Since $\gst_{ac}$ and $\Kst_a^b \gst_{bc}$ both have this form they are symmetric.

\paragraph{Constraint equations for the asymptotic profiles.}

The above data are {\sl not independent} and we also require the following asymptotic version of Einstein's constraint equations~\eqref{equa-constr}: 
\bse
\label{equa:reducedconstraints}
\begin{align}
\label{equa:asympconstr-a} 
(\Tr \Kst)^2 - \Kst^a_b \Kst_a^b & = 16 \pi \, \rhost,
\\[.5ex]
\label{equa:asympconstr-b} 
\nablast_a \Kst_b^a - \del_b (\Tr \Kst) & = 8 \pi \, \Jst_b, 
\end{align}
referred to as the {\bf velocity-dominated constraint equations.} 
Here, we have neglected the scalar curvature term and, in addition, space derivatives are  neglected in the matter components~\eqref{equa:rhoandJ}. Precisely, we set  
\bel{equa:asympconstr-c} 
\rhost \coloneqq \frac{1}{2} (\del_s \phist)^2,  
\qquad
\Jst \coloneqq - \del_s \phist d \phist,
\qquad 
\Tst \coloneqq \frac{1}{2} (\del_s \phist)^2  \gst.
\ee
\ese

We denote by $C$ the left-hand side minus the right-hand side of \eqref{equa:asympconstr-a} and by $D_b$ the same difference for the second constraint~\eqref{equa:asympconstr-b}. A calculation shows us that the evolution equations~\eqref{ADM-zero} imply  
\[
\del_s C = 2 (\Tr\Kst) \, C,
\qquad
\del_s D_b = (\Tr\Kst) D_b - \frac{1}{2} \del_b C.
\]
The first equation is a first-order differential equation for $C$, while --once the coefficient $C$ is known from the first equation-- the second equation can also be seen as a first-order differential equation for each component $D_b$.  
Therefore, these evolution equations imply that if constraints are satisfied (that is, $C=0$ and $D_b = 0$) on a hypersurface $\Hcal_s$ for some fixed time $s$, then they are satisfied for all $s<0$.

\begin{figure}\centering
\begin{tikzpicture}
    \draw[very thick] (-2,0) .. controls +(40:2) and +(-160:1.3) .. +(4,.5);
    \foreach \y in {-.6,-.4,-.2,.2,.4,.6} {
      \draw (-2,\y) .. controls +(40:2) and +(-160:1.3) .. +(4,.5); }
    \node at (-2,.3) [left] {$s>0$};
    \node at (-2,-.3) [left] {$s<0$};
    \node at (0,1.4) {$(\gst,\Kst,\phist)(s)$};
    \draw [->] (2.2,.7) node [right] {$(\gpoi,\Kpoi,\phipoi_0,\phipoi_1)$} -- (1.7,.7) -- (1.7,.45);
    \draw [->] (2.2,.15) node [right] {$(\gmoi,\Kmoi,\phimoi_0,\phimoi_1)$} -- (1.7,.15) -- (1.7,.4);
    \node at (0,-.4) {$(\gst,\Kst,\phist)(s)$};
  \end{tikzpicture}
  \caption{\label{fig:spacetimefoliation}{\bf Spacetime foliation by spacelike hypersurfaces~$\Hcal_s$.}
A singularity hypersurface~$\Hcal_0$ along which past and future singularity data $(g_\pm,K_\pm,\phi_{0\pm},\phi_{1\pm})$ are prescribed.
These data specify asymptotic profiles $(\gst,\Kst,\phist)(s)$ that solve~\eqref{ADM-zero} for $s<0$ and $s>0$, with explicit expressions given in \eqref{asymp-profile-m} and~\eqref{asymp-profile-p}, respectively.}
\end{figure}
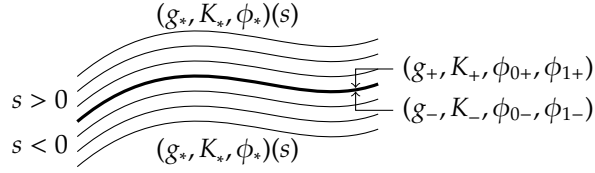

\paragraph{Initial data set on a singularity.}

We now translate the constraints on the asymptotic profile into constraints on the singularity data set $(\gmoi,\Kmoi,\phimoi_0,\phimoi_1)$.  This is simply a matter of setting $s=-1$: as we just saw, imposing the constraints at that time ensures that they hold at all times $s\in(-\infty,0)$. In addition, at this time, the tuplet $(\gst,\Kst,s\del_s\phist,\phist)$ is equal to $(\gmoi,\Kmoi,\phimoi_0,\phimoi_1)$.  Using additionally $\Tr\Kmoi = 1$, the asymptotic version of the Einstein constraints \eqref{equa:reducedconstraints} read as follows in terms of the singularity data: 
\bel{equa-const-singu} 
1  - \Kmoi^a_b  \Kmoi_a^b  = 8\pi \, (\phimoi_0)^2, 
\qquad
\nabla^-_a \Kmoi^a_b
= 8 \pi \, \phimoi_0 \, \del_b \phimoi_1.
\ee

At this stage of our general definitions, we do not need to make specific regularity assumptions (with respect to the spatial variable). In our main result below we actually work in the analytic class (but, when plane symmetry is assumed, a much weaker regularity can be handled). Throughout, $\Hcal$ denotes a $3$-manifold. 

\begin{definition}\label{def:singularity-data}
1. A set $(\gmoi, \Kmoi, \phimoi_0, \phimoi_1) = (\gmoi_{ab}, \Kmoi_a^b, \phimoi_0, \phimoi_1)$ consisting of two tensor fields and two scalar fields defined on $\Hcal$ is called a {\bf spacelike singularity initial data set} provided: 
\bel{equ:sing-space-T}
\aligned
& (i)  && \gmoi \text{ is a Riemannian metric on $\Hcal$.}
\\
& (ii)  && \Kmoi \text{ is symmetric, that is, } \gmoi_{ac} \Kmoi_b^c = \gmoi_{bc} \Kmoi_a^c. 
\\
& (iii)  && \Hcal \text{ has unit mean curvature, that is, } \Tr\Kmoi = 1 \text{ on $\Hcal$.}
\\
& (iv) && \text{The asymptotic Hamiltonian and momentum constraints \eqref{equa-const-singu} hold on $\Hcal$.}
\endaligned
\ee
The set of all such data is referred to as the {\bf space of spacelike singularity data} and is denoted by $\Ispace(\Hcal)$. 

\vskip.15cm

2. The data $(\gmoi, \Kmoi, \phimoi_0, \phimoi_1)$ are {\bf quiescent} if $\Kmoi>0$.  The space of such quiescent data is denoted by $\Ispacep(\Hcal)$.

\vskip.15cm

3. The {\bf spacelike asymptotic profile} associated with the data set $(\gmoi, \Kmoi, \phimoi_0, \phimoi_1) \in \Ispace(\Hcal)$ is the flow $s \in (-\infty,0) \mapsto \big(\gst(s), \Kst(s), \phist(s) \big)$ defined on $\Hcal$ by 
\bel{asymp-profile-m}
\gst(s) = |s|^{2 \Kmoi} \gmoi, 
\qquad
\Kst(s) = \frac{-1}{s} \Kmoi, 
\qquad
\phist(s) = \phimoi_0 \log|s|  + \phimoi_1. 
\ee
\end{definition}

For a discussion of the properties of the space $\Ispace(\Hcal)$ we refer to \autoref{subsec:param-data}. So far, we have discussed the direction toward the singularity but, clearly, a similar definition can be given in order to evolve {\sl away} from the singularity hypersurface toward the future. For the corresponding data we use the notation $(\gpoi, \Kpoi, \phipoi_0, \phipoi_1) \in \Ispace(\Hcal)$ and we define the corresponding asymptotic profile over the time interval $(0,\infty)$ by
\bel{asymp-profile-p}
\gst(s) = \abs{s}^{2\Kpoi} \gpoi, 
\qquad
\Kst(s) = \frac{-1}{s} \Kpoi, 
\qquad
\phist(s) = \phipoi_0 \log|s|  + \phipoi_1. 
\ee
We emphasize that our sign conventions in \eqref{asymp-profile-m}--\eqref{asymp-profile-p} are such that $\Tr K_{\pm}=1$ and $\Kst$ is the extrinsic curvature measured using the {\sl future-pointing} unit normal to the foliation, which explains the opposite sign of $\Tr\Kst$ for $s\lessgtr 0$.
Note that while (asymptotic profiles $\Kst$ of) the extrinsic curvatures change sign if one changes the sign of~$s$, hence of the unit normal $\del_s$, the normalized tensors $K_{\pm}$ have unambiguous signs, as exemplified by the condition $\Tr K_{\pm}=1$.
These notations are summarized in \autoref{fig:spacetimefoliation} (above), while in \autoref{fig:aspectscyclic} (below) we depict some aspects of light-cones near a singularity hypersurface.

\paragraph{An important example.} 

As an illustration of our definitions, let us consider a particular class of data sets and asymptotic profiles, in which for simplicity $\gmoi$ is chosen to be the Euclidean metric on $\Hcal \simeq \RR^3$ and $\Kmoi$ has constant eigenvectors.  In suitable coordinates, we can write $\Kmoi \equiv \diag(k_1,k_2,k_3)$ for three functions $k_1, k_2, k_3$ defined on~$\RR^3$.
This choice leads us the following {\bf generalized Kasner metric:} 
\bel{eq:metric-BianchiI}
\aligned
\gst_\Kasner 
& = (- s)^{2k_1(x)} (dx^1)^2 + (- s)^{2k_2(x)} (dx^2)^2 +  (- s)^{2k_3(x)} (dx^3)^2, \qquad s<0,
\\
g_{*\Kasner}^{(4)}
& = - ds^2 + \gst_\Kasner. 
\endaligned
\ee
This is an asymptotic profile included in the general framework above, {\sl provided suitable restrictions} are put on the data functions $k_1, k_2, k_3$. Namely, the CMC requirement $\Tr\Kmoi =1$ reads 
\bse
\label{eq:KasConditions}
\be 
k_1(x) + k_2(x) + k_3(x) = 1, 
\ee
and from the Hamiltonian constraint in \eqref{equa-const-singu} we get  
\bel{psquareleq1}
(k_1(x))^2 + (k_2(x))^2 + (k_3(x))^2 \leq 1.
\ee
We also have three differential constraints
\bel{phimoi0squared}
\del_a k_a(x) = 8 \pi \phimoi_0(x) \del_a \phimoi_1(x), 
\qquad 
\phimoi_0(x)^2 = \frac{1}{8 \pi} \Big( 1 - (k_1(x))^2 + (k_2(x))^2 + (k_3(x))^2 \Big). 
\ee
\ese

For instance, if $\phimoi_1$ is chosen to be a constant, then from the equations $\del_1k_1=\del_2k_2=\del_3k_3=0$ together with $k_1+k_2+k_3=1$, we conclude that $\del_1\del_2 k_3=0$. Hence, for this class of singularity data, $k_3$ is the sum of a function of~$x^1$ and a function of~$x^2$. Using again $k_1+k_2+k_3=1$ we arrive at the family of solutions
\[
k_1(x) = \frac{1}{3} + f_2(x^2) - f_3(x^3) , \quad
k_2(x) = \frac{1}{3} + f_3(x^3) - f_1(x^1) , \quad
k_3(x) = \frac{1}{3} + f_1(x^1) - f_2(x^2), 
\]
parametrized by three functions on~$\RR$ up to an overall shift, subject only to the inequality~\eqref{psquareleq1}, easily satisfied for example by functions with all $\abs{f_a(x^a)}<1/\sqrt{12}$.
We also recall that $\phimoi_0$ is given by~\eqref{phimoi0squared}.

Furthermore, we observe that, when the $k_a$ are chosen to be constant, the metric~\eqref{eq:metric-BianchiI} is not only an asymptotic profile but, in fact, a genuine solution to the Einstein equations. It is a vacuum solution (the Kasner solution~\cite{Kasner}) only if moreover $\phi_0^-$ vanishes.

\begin{figure}\centering
  \begin{tikzpicture}
    \fill[opacity=.5,black!20!white] (1.8,-1.44)
    parabola bend (3,0) (3,0)
    parabola (3.75,1)
    parabola bend (4.5,0) (4.5,0)
    parabola (5.7,-1.44) -- cycle;
    \draw[very thick] (6,0) -- (1.5,0) node [left] {$s=0$};
    \foreach \x in {3,4.5} {
      \filldraw (\x,0) circle (.05);
      \draw (\x,0) parabola +(.9,1.44);
      \draw (\x,0) parabola +(-.9,1.44);
      \draw (\x,0) parabola +(-1.2,-1.44);
      \draw (\x,0) parabola +(1.2,-1.44);
    }
    \draw[dashed] (2,-1) -- (5.5,-1);
    \node at (3.75,1) {$\times$};
  \end{tikzpicture}
  \caption{\label{fig:aspectscyclic}{\bf Aspects of light-cones near a singularity hypersurface.}  Past and future light-cones of two points on the $s=0$ singularity hypersurface~$\Hcal_0$, and domain of dependence (in gray) of a spacetime point (cross).  Kasner exponents $k_{a\pm}$ before and after the singularity are all less than~$1$, except in the special case $k_{1\pm}=k_{2\pm}=0,k_{3\pm}=1$ (and permutations thereof).  Null geodesics then travel by a finite amount $\sim\int ds/|s|^{k_{a-}}$ in the three spatial directions before reaching the singularity, and likewise after the singularity, hence the domain of determinacy of a sufficiently large region (such as depicted by the dashed line) can include parts of the spacetime after the singularity.  The fact that null rays can ``traverse'' the singularity enables us to set up null coordinates globally in the plane-symmetric gravitational collision problem treated in \cite{LLV-3}.} 
\end{figure}
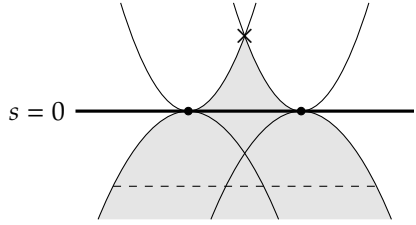

\subsection{Singularity scattering maps}

\paragraph{Beyond Israel's junction conditions.}

In order to construct a solution to the Einstein equations that crosses over a singularity hypersurface,  
some prescription has to be found for connecting data reached from both sides. The standard approach to tackle this problem, in principle, is offered by the Israel (also called Israel--Darmois) junction conditions~\cite{IJC}. However, these conditions were introduced under the assumption that, near the hypersurface, the local geometry on each side is sufficiently \emph{regular}. The conditions found by Israel were the continuity of the metric, as well as the continuity of the extrinsic curvature unless a surface matter term is present which then introduces a jump discontinuity in the extrinsic curvature. A suitable generalization of Israel's junction conditions is required in order to encompass singularity hypersurfaces such that the metric and extrinsic curvature are both blowing up. We expect that the junctions will require extra physical input and possibly a matter source of impulsive type. 

In fluid dynamics, junction conditions are necessary when two flows of different materials are separated by a moving interface or a fixed membrane \cite{LeFloch-Oslo}. For instance, for car traffic flows or other flows through a network one often introduces jump conditions that are not a consequence of the first principles of the physical theory. 
Connecting a contracting spacetime with an expanding one may be thought as analogous to a fluid flow in a converging-diverging nozzle (a so-called De Laval nozzle), and some loss might be observed (and be described by suitable small-scale physics modeling) at the throat of the nozzle; see for instance \cite{Ebersohn}.

\paragraph{The notion of scattering map.}

We regard a singularity hypersurface as an interface between two ``phases'' across which the geometric and matter fields may encounter a jump, due to small-scale physics that we are not modeling here.  We are only interested in the ``average'' effect rather than the detailed physics that may take place within this interface. This is a standard strategy in fluid dynamics or material science when small parameters like viscosity, capillarity, heat conduction, etc., are neglected in the modeling, yet have a macroscopic effect that is captured by imposing suitable jump conditions. 

\begin{definition}
\label{def:singu}
A {\bf spacelike singularity scattering map} on a $3$-manifold $\Hcal$ is a diffeomorphism-covariant map
\bel{scattering-def}
\Sbf\colon (\gmoi, \Kmoi, \phimoi_0, \phimoi_1) \in \Ispace(\Hcal) 
\mapsto (\gpoi, \Kpoi, \phipoi_0, \phipoi_1) \in \Ispace(\Hcal)
\ee
defined over the space of singularity data $\Ispace(\Hcal)$ and satisfying the following {\bf locality property}: for any open set $U \subset \Hcal$, the restriction 
\[
\Sbf(\gmoi, \Kmoi, \phimoi_0, \phimoi_1)\big|_U \text{ depends only on } (\gmoi, \Kmoi, \phimoi_0, \phimoi_1) \big|_U. 
\]
\end{definition}

\begin{remark}
The locality property ensures that $\Sbf$~is entirely determined by its restriction to any small open ball, which is independent of the ball thanks to diffeomorphism invariance. This restriction is an arbitrary singularity scattering map on the ball.
Specifying a singularity scattering map~$\Sbf$ on~$\Hcal$ is thus equivalent to specifying one on a unit ball, and it is therefore natural to identify singularity scattering maps~$\Sbf$ on all $3$-manifolds and suppress the dependence on~$\Hcal$.
\end{remark}

A map $\Sbf$ is said to be {\bf ultralocal (or pointwise)} if 
\[
\Sbf(\gmoi, \Kmoi, \phimoi_0, \phimoi_1)(x) \text{ depends only on } (\gmoi, \Kmoi, \phimoi_0, \phimoi_1)(x), 
\qquad x \in \Hcal. 
\]
By diffeomorphism invariance, the restrictions $\Sbf_x$ to every point~$x$ are the same.
The ultralocality condition is motivated by the fact that the dynamics at different spatial points decouple near a spacelike singularity.  As we see in \autoref{secti-24}, below, the class of ultralocal singularity scattering maps is rich while still being amenable to classification.
We classify in \autoref{theorem-ultralocal} all ultralocal scattering maps for a self-gravitating scalar field.

We say that $\Sbf$ is a {\bf quiescence-preserving singularity scattering map} if it preserves positivity of~$K$ in the sense that
\[
\text{if } \Kmoi > 0 \text{ then } \Kpoi > 0 , \text{ where } \Kpoi \text{ is defined by \eqref{scattering-def}.}
\]
In asymptotic profiles with $\Kmoi,\Kpoi>0$ all distances decrease to zero as $s\to 0^-$, then increases back to finite values for $s>0$: such profiles describe a ``bounce''. (This positivity condition is also motivated by the absence of BKL oscillations in the presence of matter.) 

In view of the earlier literature (for instance \cite{LubbeTod}), we can also single out singularities across which the metric jumps by a conformal transformation: a singularity scattering map~$\Sbf$ is called {\bf rigidly conformal} if
\[
\gpoi=\lambda^2\gmoi
\]
for some scale factor~$\lambda$.  We introduce a more general and natural concept of {\bf conformal} maps in \autoref{subsec-example-comparison}.

\paragraph{Further properties of scattering maps.}

Another natural physical requirement on scattering maps is to respect symmetries under constant shifts of~$\phi$.  This means that shifting $\phimoi_1\to\phimoi_1+\varphi$ for some constant $\varphi\in\RR$ should only affect the result by a shift $\phipoi_1\to\phipoi_1+\psi$ that is constant in space.  Locality only allows~$\psi$ to depend on~$\varphi$, and composing two shifts shows that $\psi$ is simply a multiple of~$\varphi$.  Thus, we say that $\Sbf$ is {\bf shift-covariant} if there exists a coefficient $a\in\RR$ such that
\[
\Sbf(\gmoi,\Kmoi,\phimoi_0,\phimoi_1+\varphi) = (\gpoi,\Kpoi,\phipoi_0,\phipoi_1+a\varphi)
\]
for any constant $\varphi\in\RR$, any singularity data set $(\gmoi,\Kmoi,\phimoi_0,\phimoi_1)$, and its image $(\gpoi,\Kpoi,\phipoi_0,\phipoi_1)$ under~$\Sbf$.
The cases $a=\pm 1$ are particularly interesting: at least for ultralocal maps classified in \autoref{theorem-ultralocal}, we find {\bf momentum-preserving} maps ($\Kpoi=\Kmoi$ and $\phipoi_0=\pm\phimoi_0$) or {\bf momentum-reversing} maps ($\Kpoi=\frac{2}{3}\delta-\Kmoi$ and $\phipoi_0=\pm\phimoi_0$).
This is somewhat unsurprising, from Noether's theorem applied to whichever microscopic theory governs the bounce if it is shift-covariant.

Finally, a map $\Sbf$ is called {\bf idempotent} if
\[
\Sbf \circ \Sbf \text{ is the identity map on } \Ispace(\Hcal) .
\]
The condition states that the two sides of the singularity play the same role.
A weaker requirement is that $\Sbf$~be an {\bf invertible} map, or equivalently that either singularity data set $(g_{\pm},K_{\pm},\phi_{0\pm},\phi_{1\pm})$ can be expressed in terms of the other.
This becomes a very natural requirement when extending our definitions to timelike singularities, where one may want both sides of the singularity to play the same role.

As an aside, it is easy to check that composing two scattering maps gives a scattering map and that if the two scattering maps are both ultralocal, quiescence-preserving, rigidly conformal, or shift-covariant, then their composition also has the same property.  The sets of scattering maps with any of these properties thus forms a semigroup under composition.


\section{Cyclic \((3+1)\)-dimensional spacetimes based on a scattering map} 
\label{section---3}

\subsection{Spacetimes with spacelike singularity hypersurfaces}

\paragraph{Main contribution of this section.}

In order to provide a first application of the formalism we propose in the present paper, we now 
establish the following result, together with the more general statement in \autoref{theo:391} below.  In addition, for a global construction scheme within the class of plane-symmetric  cyclic spacetimes, we refer to the companion paper \cite{LLV-3}. 

We are going to show in the following result. 
Given any a quiescence-preserving singularity scattering map defined on a three-manifold with boundary~$M^3$, 
there exists a large class of spacetimes diffeomorphic to $M^{3+1} \simeq [t_{-1}, t_1] \times M^3$, satisfying the Einstein-scalar field system and containing a spacelike singularity hypersurface that separates the two regions of regularity $[t_{-1}, 0) \times M^3$ and $(0,t_1] \times M^3$.
These spacetimes are expressed in Gaussian coordinates (also called synchronous gauge) in which the singularity is simultaneous, while the past and future limits at the singularity hypersurface $t=0$ are related by the prescribed scattering map. Moreover, these solutions are parametrized by the expected degrees of freedom for the Cauchy problem, that is, the induced metric, extrinsic curvature, and matter field on one of the foliation hypersurfaces.

\paragraph{Generalization.}

An analogous result holds with timelike hypersurfaces, as we will state in \autoref{theo:391}. 
Interesting, our results admit several direct extensions. Using the recent advances in \cite{AlexakisFournodavlos,Fournodavlos:2016,FournodavlosLuk,FournodavlosRonianskiSpeck,RodnianskiSpeck:2018b,RodnianskiSpeck:2018c,Speck:2018} it is straightforward to reformulate our conclusions at the Sobolev regularity level. 
Furthermore, for any initial data set in a large class of data in the sub-critical regime,
 as described in \cite{FournodavlosRonianskiSpeck}, the initial value problem can be solved from a spacelike hypersurface {\sl toward} a spacelike singularity hypersurface; next, to the corresponding initial data set on the singularity hypersurface we can then apply our singularity scattering map and, finally,
we can evolve toward the future by Fuchsian techniques following~\cite{FournodavlosLuk}. 

\subsection{Timelike singularity scattering maps}

\paragraph{A generalization of the ADM formalism.} 

The treatment of timelike hypersurfaces is formally analogous to the one of spacelike hypersurfaces, and we now outline the necessary modifications that are required in the definitions and results above.  After defining singularity scattering maps for (quiescent) timelike singularities in this section, we will introduce the notion of cyclic spacetimes with singularity hypersurfaces, and apply Fuchsian techniques to construct such spacetimes locally.

For spacelike singularity hypersurfaces we worked with a Gaussian foliation (also called synchronous gauge) such that the four-dimensional metric takes the form $g^{(4)}=-ds^2+g(s)$.  The analogous setup in the timelike case starts with a local foliation of a spacetime by hypersurfaces $\Hcal_s$, $s\in[s_{-1},s_1]$ with $s_{-1}<0<s_1$, endowed with a (symmetric) {\sl Lorentzian} metric $g(s)=(g_{ab}(s))$ and an extrinsic curvature $K(s)=(K_a^b(s))$ such that $K_a^b g_{bc}$ is symmetric.  Here, indices $a,b,\dots$ are local coordinate indices on slices of the foliation.  Without loss of generality locally, we assume the foliation to be a {\bf proper distance foliation}, in the sense that one has diffeomorphisms $\Hcal_s\simeq\Hcal_0$ such that the four-dimensional metric reads $g^{(4)} = ds^2 + g(s)$.
Such a foliation can be constructed in a neighborhood of~$\Hcal_0$ by defining $s$ as the proper distance along geodesics normal to~$\Hcal_0$.

It is useful to treat both spacelike and timelike hypersurfaces together by writing the four-dimensional metric in a proper time or proper distance foliation as
\be
  g^{(4)} = \epsilon \, ds^2 + g(s),
\ee
where $\epsilon=-1$ for a spacelike foliation and $\epsilon=+1$ for a timelike foliation.
Taking into account the signature, the matter evolution equation~\eqref{equa-wavephi} for the massless scalar field~$\phi$ becomes
\bel{equa:time-wave}
- \del_s^2 \phi + \Tr(K)\, \del_s \phi = \epsilon \, \Delta_g \phi, 
\ee
where $\Delta_g\phi = \nabla_a\nabla^a\phi$ is the Laplacian operator on spacelike slices or the D'Alembertian on timelike slices.  The ADM formulation \eqref{equa:ADMsysytem}--\eqref{equa-constr} for the Einstein equations now reads 
\bel{equa:time-ADM}
\aligned
\del_s g_{ab} + 2 \, K_{ab} & = 0 , \\
\del_s K_a^b - (\Tr K)  K_a^b  & =  -\epsilon \, R_a^b + 8 \pi \epsilon \del_a\phi \del^b\phi, \\
(\Tr K)^2 - \Tr (K^2) - 8 \pi (\del_s\phi)^2 & = \epsilon \,R - 8\pi\epsilon\del_a\phi\del^a\phi, \\
\nabla_a K^a_b  - \del_b (\Tr K) + 8 \pi \del_s\phi \del_b\phi & = 0.
\endaligned
\ee
While in the spacelike case the first two equations are evolution equations and the last two are constraints on the initial data, no such interpretation is available in the timelike case since $\del_s$ is then a spatial derivative.

\paragraph{Data for timelike hypersurfaces.}

In both spacelike and timelike cases, asymptotic profiles are found by neglecting derivatives along leaves of the foliation compared to $s$~derivatives.  This turns out to exactly remove the $\epsilon$-dependent terms (all right-hand sides) in the system \eqref{equa:time-wave}--\eqref{equa:time-ADM}.  In both cases asymptotic profiles thus take the form
\bel{equa:time-asymptotic-profile}
  \gst(s) = \abs{s}^{2\Kmoi}\gmoi , \qquad \Kst(s) = \frac{-1}{s} \Kmoi , \qquad
  \phist(s) = \phimoi_0\log|s| + \phimoi_1 ,
\ee
in terms of singularity data $(\gmoi,\Kmoi,\phimoi_0,\phimoi_1)$ such that $\Tr\Kmoi=1$ and the asymptotic Hamiltonian and momentum constraints~\eqref{equa-const-singu} hold.  This leads to a natural extension of \autoref{def:singularity-data} to the case of timelike singularity hypersurfaces.
We keep the notion of quiescent data defined as in the spacelike case since the same positivity condition appears in both cases in our main existence theorem.

\begin{definition}\label{def:time-singularity-data}
  1. A {\bf timelike singularity initial data set} on $\Hcal$ is a set $(\gmoi, \Kmoi, \phimoi_0, \phimoi_1)$ consisting of a Lorentzian metric $\gmoi$ on~$\Hcal$, a two-tensor $\Kmoi=(\Kmoi_a^b)$ that is symmetric (namely $\gmoi_{ac} \Kmoi_b^c = \gmoi_{bc} \Kmoi_a^c$) and obeys $\Tr\Kmoi=1$, and two scalar fields such that the constraints~\eqref{equa-const-singu} hold.
  The {\bf space of timelike singularity data}, denoted by $\Itime(\Hcal)$, is the set of all such data.

\vskip.15cm

2. Data $(\gmoi, \Kmoi, \phimoi_0, \phimoi_1)$ are {\bf quiescent} if $\Kmoi>0$.  The space of such quiescent timelike singularity data is denoted by $\Itimep(\Hcal)$.

\vskip.15cm

3. The {\bf timelike asymptotic profile} associated with a data set $(\gmoi, \Kmoi, \phimoi_0, \phimoi_1) \in \Itime(\Hcal)$ is the flow $s \in (-\infty,0) \mapsto \big(\gst(s), \Kst(s), \phist(s) \big)$ defined on $\Hcal$ by~\eqref{equa:time-asymptotic-profile}.
\end{definition}

\paragraph{An example.}

It is useful to consider again the class of asymptotic profiles with Kasner behavior.
These, now, only depend on the proper distance from a timelike hypersurface labelled $s=0$. Specifically, the {\bf generalized Kasner profile with timelike singularity} is defined as 
\[
\aligned
  g_{*\Kasner}^{(4)} & =  ds^2 + g_{*\Kasner} , 
  \qquad  \qquad
  g_{*\Kasner}  = - (- s)^{2k_1(t,x)} dt^2 + (- s)^{2k_2(t,x)} (dx^2)^2 +  (- s)^{2k_3(t,x)} (dx^3)^2, \qquad s<0.
\endaligned
\]
The discussion of~\eqref{eq:metric-BianchiI} applies verbatim, apart from renaming the coordinate $x^1$ to $t$ to emphasize that it is now a time coordinate.  Namely, $g_{*\Kasner}$ is an asymptotic profile included in our framework provided $k_1+k_2+k_3=1$, $k_1^2+k_2^2+k_3^2\leq 1$, and the three differential constraints~\eqref{phimoi0squared} are obeyed.  As in the spacelike case, the profile is an exact solution of Einstein's equations when the exponents $k_a$~are constants.

\paragraph{Junction along timelike hypersurfaces.}

Our \autoref{def:singu} of singularity scattering maps extends straightforwardly from the spacelike to the timelike case by changing $\Ispace$ to~$\Itime$.

\begin{definition}
  A {\bf timelike singularity scattering map} on a $3$-manifold~$\Hcal$ is a local diffeomorphism-covariant map $\Sbf\colon\Itime(\Hcal)\to\Itime(\Hcal)$.
\end{definition}

As in the spacelike case, the singularity hypersurface on which the (local) scattering map is defined is irrelevant.
Likewise, a timelike singularity scattering map is defined to be ultralocal, quiescence-preserving, or rigidly conformal under the same conditions as for the spacelike case.
Since the only difference between $\Ispace$ and $\Itime$ is the signature of the metric, it is natural to combine these two spaces into the {\bf space of singularity data}
\[
\Ibf(\Hcal) = \Ispace(\Hcal) \sqcup \Itime(\Hcal) ,
\]
whose elements are tuples $(\gmoi,\Kmoi,\phimoi_0,\phimoi_1)$ in which $\gmoi$ is a Riemannian or Lorentzian metric, $\Kmoi$ is symmetric, $\Tr\Kmoi=1$, and the Hamiltonian and momentum constraints are obeyed.

\begin{definition}
  A {\bf singularity scattering map} is a local diffeomorphism-covariant map $\Sbf\colon\Ibf(\Hcal)\to\Ibf(\Hcal)$ that maps $\Ispace(\Hcal)$ to itself and $\Itime(\Hcal)$ to itself.\footnote{We tacitly assume that the scattering maps are regular.  Specifically, we require that smooth data are mapped to (at least) continuous data.}
\end{definition}

\subsection{The notion of cyclic spacetimes}

\paragraph{The main definition.}

We are now ready to introduce a notion of spacetimes with singularities that encompasses the common construction of cyclic spacetimes made of successive epochs separated by bounces.  We describe each bounce using a singularity scattering map.

In our case-study of colliding plane symmetric gravitational waves in \cite{LLV-3} we naturally construct a spacetime with intersecting singularity hypersurfaces, and with singularity hypersurfaces whose spacelike or timelike nature generically changes along a two-dimensional locus.  Away from that locus, singularity hypersurfaces have a fixed nature and it is natural to match the asymptotic descriptions of the metric on both sides using a singularity scattering map, as described by the following definition.
Another motivation to exclude a two-dimensional exceptional locus in the definition below is that it allows for the presence of non-generic ``spikes'', where derivatives parallel to the singularity are not negligible compared to derivatives transverse to it~\cite{RendallWeaver}. 

\begin{definition}
\label{def:singu2}
A {\bf cyclic spacetime} $(\Mcal^4, \Ncal^3, \Pcal^2, g^{(4)}, \phi)$ based on a singularity scattering map $\Sbf$ is a smooth oriented $4$-manifold $\Mcal^4$, endowed with a Lorentzian metric~$g^{(4)}$ and a scalar field~$\phi$, both defined outside a singular locus $\Ncal^3 \subset \Mcal^4$ consisting of the union of a collection of oriented and smooth hypersurfaces with boundary and an exceptional $2$-dimensional locus $\Pcal^2\subset\Ncal^3$, with the following properties.

\bei 

\item {\bf Einstein equations.} 
The Einstein-scalar field evolution and constraint equations $G^{(4)}_{\alpha\beta}=8\pi T^{(4)}_{\alpha\beta}$ and the matter evolution equation $g^{(4)\alpha\beta}\nabla_\alpha^{(4)}\nabla_\beta^{(4)}\phi=0$ hold outside the singular locus~$\Ncal^3$.

\item {\bf Local foliations.} 
Every point in $\Ncal^3 \setminus\Pcal^2$ admits a neighborhood~$\Ucal$ that can be endowed with a foliation by hypersurfaces $\Hcal_s$ for $s$ in an interval $(s_{-1},s_1)$ containing $s=0$, such that $\Hcal_0 = \Ncal^3 \cap \Ucal$, the hypersurfaces $\Hcal_s$ are all diffeomorphic to~$\Hcal_0$, and the metric reads $g^{(4)}=\pm ds^2+g(s)$ for a one-parameter family of metrics $g(s)$ defined on~$\Hcal_s \simeq \Hcal_0$. (The orientation of $\del_s$ is chosen to be compatible with the orientation 
of $\Mcal^4$ and the hypersurfaces.) 

\item {\bf Singularity behavior.} 
Near each such~$\Hcal_0$, the singularity data from both sides are well-defined, as the limits 
\bel{equa-limits}
\aligned
(\gpoi, \Kpoi, \phipoi_0, \phipoi_1) 
& \coloneqq \lim_{\substack{s \to 0\\s > 0}} \Big(|s|^{2 s K} g,\; -s K,\; s \del_s \phi,\; \phi - s \log |s|  \del_s \phi \Big)(s), 
\\
(\gmoi, \Kmoi, \phimoi_0, \phimoi_1) 
& \coloneqq \lim_{\substack{s \to 0\\s < 0}} \Big(|s|^{2 s K} g,\; -s K,\; s \del_s \phi,\; \phi - s \log |s|  \del_s \phi \Big)(s). 
\endaligned
\ee

\item {\bf Scattering conditions.}
On each such~$\Hcal_0$, the following junction conditions hold: 
\bel{equa:thejunction}
(\gpoi, \Kpoi, \phipoi_0, \phipoi_1) = \Sbf(\gmoi, \Kmoi, \phimoi_0, \phimoi_1).
\ee 

\eei 
\end{definition}

To motivate the choice of limits in~\eqref{equa-limits}, we remark that for any asymptotic profile $s \in (-\infty,0) \mapsto (\gst(s), \Kst(s), \phist(s) )$, these limits coincide with the corresponding singularity data $(\gmoi, \Kmoi, \phimoi_0, \phimoi_1)$ and, in fact, the arguments of the limits are {\sl independent} of $s$.
In \autoref{theo:391} below we check that, when data are imposed on one side of the singularity with positive definite~$\Kmoi$, suitable solutions as described in \autoref{def:singu2} do exist and admit well-defined limits~\eqref{equa-limits}.
\autoref{fig:singu2} summarizes some notations.

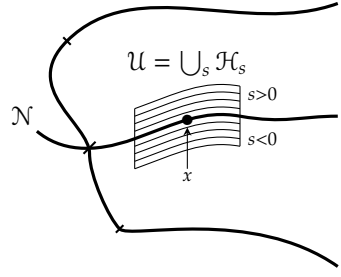
\begin{figure}\centering
  \begin{tikzpicture}
    \draw[very thick] (-2,-.2)
    .. controls +(-40:.4) and +(-160:.6) .. (-.7,-.3)
    .. controls +(20:.6) and +(170:.5) .. (.7,0)
    .. controls +(-10:.6) and +(180:.5) .. (2,0);
    \foreach \y in {-.7,-.6,-.5,-.4,-.2,-.1,0,.1} {
      \draw (-.7,\y) .. controls +(20:.6) and +(170:.5) .. +(1.4,.3); }
    \draw (-.7,-.7) -- (-.7,.1);
    \draw (.7,-.4) -- (.7,.4);
    \node at (0,.7) {$\Ucal=\bigcup_s\Hcal_s$};
    \node at (0,-.065) {$\bullet$};
    \draw[-{stealth}] (0,-.7) node [below=-.5ex] {$\scriptstyle x$} -- (0,-.15);
    \node at (1,.3) {$\scriptstyle s>0$};
    \node at (1,-.3) {$\scriptstyle s<0$};
    \node at (-2.2,0) {$\Ncal$};
    \draw[very thick] (2,1.5)
      .. controls +(-160:1.3) and +(45:1.3) .. (-1.6,1)
      .. controls +(-135:.9) and +(90:.4) .. (-1.3,-.5)
      .. controls +(-90:.3) and +(135:.2) .. (-.9,-1.5)
      .. controls +(-45:.2) and +(145:1.4) .. (2,-2);
    \draw[thick] (-1.55,.95) -- (-1.65,1.05);
    \draw[thick] (-.95,-1.55) -- (-.85,-1.45);
    \draw[thick] (-1.38,-.5) -- (-1.22,-.34);
  \end{tikzpicture}
  \caption{\label{fig:singu2}{\bf Singular locus in a cyclic universe.}  In this example, the singular locus~$\Ncal$ consists of two singularity hypersurfaces, one of which is spacelike.  The codimension~$2$ locus~$\Pcal$ (depicted by three ticks along~$\Ncal$) consists of their intersection and of the two-dimensional locus where a hypersurface changes between timelike and spacelike nature.  It could have additional components near which the approximation of the metric by an asymptotic profile breaks down (so-called spikes).  By definition, every $x\in\Ncal\setminus\Pcal$ admits a neighborhood $\Ucal$ with a (spacelike or timelike) Gaussian foliation such that $\Hcal_0=\Ncal\cap\Ucal$.}
\end{figure}

\paragraph{Geometric properties near a singularity.} 

We explain now some geometric consequences of the definition.
In the more constrained setting of \autoref{theo:391}, we will see that somewhat stronger statements (whose proof requires Fuchsian techniques) hold.
For the time being we study the curvature of a general cyclic spacetime $(\Mcal^4, \Ncal^3, \Pcal^2, g^{(4)}, \phi)$, based only on the definition.
Let $x\in\Ncal\setminus\Pcal$ be a point on the singularity locus such that $\phi_{0\pm}$ defined by~\eqref{equa-limits} are \emph{non-zero.}
  Consider a local foliation $\Hcal_s$, $s\in(s_{-1},s_1)$ whose existence is entailed by \autoref{def:singu2}, namely a foliation such that $x\in\Hcal_0\subset\Ncal$ while other leaves do not intersect~$\Ncal$, and such that $g^{(4)}=\epsilon ds^2+g(s)$ ($\epsilon=\pm 1$) for some diffeomorphisms $\Hcal_s\simeq\Hcal_0$.
Restrict the foliation to a smaller neighborhood if necessary so that 
\bel{equa-nondeg} 
\phi_{0\pm}\neq 0 \text{ throughout } \Hcal_0.
\ee
Observe that this condition is for instance ensured if we assume quiescent data ($K_{\pm}>0$), since $\Tr K_{\pm}=1$ and the asymptotic Hamiltonian constraint then imply
$4\pi \phi_{0\pm}^2 = k_{1\pm} k_{2\pm} + k_{2\pm} k_{3\pm} + k_{3\pm} k_{1\pm} > 0$,
where $k_{a\pm}$ are the eigenvalues of $K_{\pm}$.

 \bei 
 
\item{\bf Behavior of the curvature.}
Under the assumption~\eqref{equa-nondeg}, we now show that the spacetime curvature component $R^{(4)}_{00}$ along the unit normal to the foliation blows up with a uniform power of~$s$, namely 
$
\lim_{s \to 0^{\pm}} s^2 R^{(4)}_{00}(s) = 8\pi \phi_{0\pm}^2$
on $\Hcal_0,
$
so that $\Hcal_0$ is a curvature singularity.
This is checked as follows. By the Einstein equations and the form of the massless scalar stress-energy tensor,
we have 
$
R^{(4)}_{00} = 8\pi \, (\del_s\phi)^2$.
Since $s\del_s\phi\to\phi_{0\pm}$ as $s\to 0^{\pm}$, this term behaves as $\phi_{0\pm}^2/s^2$, as announced.
On the other hand, for general cyclic spacetimes we cannot show that the spacetime \emph{scalar curvature} $R^{(4)}$ blows up, since 
we are not assuming any control of the spatial derivatives, such as the property $\del_a\phi=O(\log|s|)$ for solutions that we construct in \autoref{theo:391}.

\item{\bf Behavior of the second fundamental form.}
We see that the mean curvature $H=\Tr K$ of the leaves blows up uniformly:
\bse\label{sHseqm1}
\be
\lim_{s \to 0} s \, H(s) = -1 \quad \text{ on } \Hcal_0.
\ee
We check this rather easily by noting that $-sK\to K_{\pm}$ implies $-sH(s)\to \Tr K_{\pm}=1$.

\item{\bf Behavior of the volume.} 
In the spacelike case ($g^{(4)}=-ds^2+g(s)$ with $g$ Riemannian) the volume of co-moving regions vanishes on the singularity.
Namely, consider a compact subset $C\subset\Hcal_0$ and let $V_C(s)$ be the volume of its image under the diffeomorphism $\Hcal_0\simeq\Hcal_s$.  This volume shrinks to zero
\be
  \lim_{s \to 0} V_C(s) = 0 ,
\ee
which, together with $H(s)\to-\infty$, implies that $\Hcal_0$ is (by definition) a crushing singularity.
\ese

This is checked as follows. By definition, $|s|^{2sK}g\to g_{\pm}$ as $s\to 0^{\pm}$, so $|s|^{s\Tr K}|g|^{1/2} \to |g_{\pm}|^{1/2}$.  Since $s\Tr K\to -\Tr K_{\pm}=-1$, for sufficiently small~$s$ we have $s\Tr K<-1/2$ (say) so
\[
|g|^{1/2} \lesssim |s|^{1/2} |g_{\pm}|^{1/2} \to 0 \text{ as } s\to 0^{\pm} .
\]
We conclude that the volume of any compact region shrinks as $s\to 0$. 

\eei

\subsection{Existence theory and qualitative behavior}

\paragraph{Solutions generated by the Fuchsian method.}

We are now in a position to show the existence of a large class of spacetimes with prescribed singularity data on a spacelike or timelike hypersurface.
Motivated by \cite{AnderssonRendall} our result is restricted to the regime where the extrinsic curvature is positive. The regime where some of the Kasner exponents (eigenvalues of $K_{\pm}$ as introduced above)
are negative is not amenable to the theory of Fuchsian equations since 
the more involved BKL oscillation mechanism generically takes place.
Despite a local existence theory being available only in the all-positive regime, we have allowed our definitions of singularity data and singularity scattering maps to cover the general case where some Kasner exponents may be negative, as this case appears in our study of plane-symmetric spacetimes in~\cite{LLV-3}. 

\begin{theorem}[A class of cyclic spacetimes based on arbitrary Fuchsian data] 
\label{theo:391} 
Consider an analytic three-manifold $\Hcal_0$ together with a singularity scattering map $\Sbf\colon \Ibf(\Hcal_0) \mapsto \Ibf(\Hcal_0)$ defined over the space of singularity data on $\Hcal_0$.
Assume that the map $\Sbf$ is quiescence-preserving in the sense that it preserves positivity of the extrinsic curvature.

\bei

\item {\bf Existence theory.}  
Then, given any singularity data $(\gmoi, \Kmoi, \phimoi_0, \phimoi_1)$ defined and analytic on $\Hcal_0$, that is quiescent in the sense that
$
\Kmoi > 0,
$
there exists a cyclic spacetime $(\Mcal^{(4)}, g^{(4)})$ 
in which $\Hcal_0$ embeds as a single singularity hypersurface such that  the initial conditions \eqref{equa-limits} hold on the two sides of the singularity, with the singularity data $(\gpoi, \Kpoi, \phipoi_0, \phipoi_1)$ determined from $(\gmoi, \Kmoi, \phimoi_0, \phimoi_1)$ via the singularity scattering map $\Sbf$, as stated in \eqref{equa:thejunction}. In particular, for every compact subset $L \subset \Hcal_0$ there exists an $s_*>0$ such that every geodesic originating from $L$ normal to the singularity hypersurface
exists for a proper time or distance $s_*$. 

\item {\bf Crushing curvature singularity property.}
Furthermore, in the spacelike case and assuming for definiteness that $\Hcal_0$ is compact, the mean curvature $H(s)$ and the volume $V(s) = \Vol_{\Hcal_s}$ of the slices  (as pointed out in  \eqref{sHseqm1}) satisfy
\[
\lim_{s \to 0} s \, H(s) = -1 \quad \text{ on } \Hcal_0, 
\qquad
\lim_{s \to 0} V(s) = 0, 
\]
so that $s=0$ is a crushing singularity, while the spacetime scalar curvature blows up
in a uniform way, namely 
$
\lim_{s \to 0^{\pm}} s^2 R^{(4)}(s) = 8\pi \, \epsilon \, \phi_{0\pm}^2$ on $\Hcal_0$. 

\eei
\end{theorem}

\paragraph{Step 1. Spacelike hypersurface.} 

We rely on the existence theorem established in \cite{AnderssonRendall}, which treats spacelike hypersurfaces only.
  With our terminology, the main theorem therein states that given any singularity initial data set such that extrinsic curvature has a definite sign, there exists an actual solution to the Einstein-scalar field system that enjoys the same asymptotic behavior as the associated asymptotic profile. Recall that analyticity of the Fuchsian data is assumed in \cite{AnderssonRendall} and throughout the present discussion. 

  The sign condition is that the initial data has negative definite extrinsic curvature tensor~$K$ (defined using a normal pointing \emph{away} from the singularity), which holds in our case on both sides of the singularity because of the signs in \eqref{asymp-profile-m}, \eqref{asymp-profile-p} and our convention to use a normal pointing toward increasing values of~$s$.
  Equivalently, the corresponding Kasner exponents (eigenvalues of~$-K$) are all positive.
  (Other behaviors are in principle possible but the Fuchsian-type arguments in \cite{AnderssonRendall} would not apply.) 

  The given singularity data $(\gmoi,\Kmoi,\phimoi_0,\phimoi_1)$ has everywhere positive-definite~$\Kmoi$ by assumption, hence their existence theorem provides a local solution defined on the $s<0$ side\footnote{We recall from \autoref{def:singu2} that this may be the past or future side of the singularity depending on the cyclic spacetime's orientation.} of the singularity.
  By the same token, the image $(\gpoi,\Kpoi,\phipoi_0,\phipoi_1)$ under the singularity scattering map~$\Sbf$ also has $\Kpoi>0$ (because $\Sbf$ is quiescence-preserving), so that applying the same existence theorem, but forward in time, yields a local solution for small $s>0$.
  Pasting these two solutions together along the singularity hypersurface at $s=0$ completes our construction in the spacelike case.

  To prove that we constructed a cyclic spacetime in the sense of \autoref{def:singu2}, there remains to show that $(|s|^{2sK}g,-sK,s\del_s\phi,\phi-s\log|s|\del_s\phi)$ tends pointwise to the given singularity data $(g_{\pm},K_{\pm},\phi_{0\pm},\phi_{1\pm})$ as $s\to 0^{\pm}$ on both sides of the singularity.
  We fix once and for all a side $\pm s>0$ and a specific spatial point~$x_0$ at which we study the behavior.
  We rely on detailed estimates of \cite{AnderssonRendall} on the difference between the solution $(g,K,\phi)$ and its asymptotic profile $(\gst,\Kst,\phist)$ as $s\to 0^{\pm}$.
  They fix a positive $\alpha_0<\min(k_{a\pm}(x_0))/10$ and construct a frame on a neighborhood of~$x_0$ such that at each point~$x$ in this neighborhood, $K_{\pm}(x)$ is close to a tensor $Q(x)=\diag(q_1(x),q_2(x),q_3(x))$ that is diagonal in that frame, in the sense that the spectrum of $K_{\pm}-Q$ stays in a small interval $(-\alpha_0/8,\alpha_0/8)$.
  One can arrange for $K_{\pm}(x_0)=\diag(k_{1\pm}(x_0),k_{2\pm}(x_0),k_{3\pm}(x_0))$ to be diagonal in that frame at~$x_0$, and the closeness condition implies $|k_{a\pm}(x_0)-q_a(x_0)|<\alpha_0/8$.
  For each $a,b=1,2,3$ they prove in particular the following estimates:
  \bel{ARproof-0}
  \Bigl(
  |s|^{-\alpha^a{}_b} (g_*^{-1} g-\delta)^a_{\ b} , \quad
  |s|^{1-\alpha^a{}_b}(K-\Kst)^a_{\ b} , \quad
  \phi - \phist , \quad
  |s|\log|s|\,\del_s(\phi-\phist) , \quad
  \del_a (\phi - \phist)
  \Bigr)
  \xrightarrow{s\to 0^{\pm}} 0 ,
  \ee
  where $\alpha^a{}_b=\alpha_0+2\max(0,q_b-q_a)>0$.
  These imply our desired limits immediately except for the metric, which we now study.

  By symmetry of~$K$ we have $|s|^{2sK}g=g|s|^{2sK}$, so it is enough to show that $g_{\pm}^{-1} g|s|^{2sK}\to\delta$.
  Using that $g_{\pm}=\gst|s|^{-2K_{\pm}}$ by construction of the asymptotic profile, we find
  \bel{ARproof-1}
  g_{\pm}^{-1} g|s|^{2sK} = \bigl(|s|^{2K_{\pm}} \gst^{-1} g \, |s|^{-2K_{\pm}}\bigr) \bigl(|s|^{2K_{\pm}} |s|^{2sK}\bigr) .
  \ee
  The first factor is close to the identity:
  \[
    \bigl(|s|^{2K_{\pm}} \gst^{-1} g \, |s|^{-2K_{\pm}}\bigr)^a_{\ b} = |s|^{2k_{a\pm}-2k_{b\pm}} (\gst^{-1}g)^a_{\ b} = |s|^{2k_{a\pm}-2k_{b\pm}} \bigl(\delta^a_b + o\bigl(|s|^{\alpha^a{}_b}\bigr)\bigr) = \delta^a_b + o(|s|^{\alpha_0/2})
  \]
  where we first used that $K_{\pm}$ is diagonal, then~\eqref{ARproof-0}, then $\alpha^a{}_b+2k_{a\pm}-2k_{b\pm} \geq \alpha_0+2(q_b-k_{b\pm})-2(q_a-k_{a\pm}) \geq \alpha_0-4\alpha_0/8$.
  Next, we write the second factor as $\Delta(1)$ with $\Delta(\lambda)\coloneqq|s|^{2\lambda K_{\pm}}|s|^{2\lambda sK}$.
  We observe that
  \[
  \tfrac{1}{2} \del_\lambda\Delta(\lambda) = |s|^{2\lambda K_{\pm}}(K_{\pm}+sK)|s|^{2\lambda sK} = |s|^{2\lambda K_{\pm}}(K_{\pm}+sK)|s|^{-2\lambda K_{\pm}}\Delta(\lambda) ,
  \]
  whose solution with $\Delta(0)=\delta$ is explicitly given by the series (a path-ordered exponential)
  \bel{ARproof-3}
  \Delta(\lambda)
  = \sum_{n\geq 0} \int_{0\leq\mu_n\leq\dots\leq\mu_1\leq\lambda} \biggl(\prod_{i=1}^n 2 |s|^{2\mu_iK_{\pm}}(K_{\pm}+sK)|s|^{-2\mu_iK_{\pm}}\biggr)\, d^n\mu .
  \ee
  Indeed, it is easy to check this is formally a solution, while convergence of the series is checked as follows.
  Their estimate~\eqref{ARproof-0} on~$K$ reads $(K_{\pm}+sK)^a{}_b=o(|s|^{\alpha^a{}_b})$, hence, for $0\leq\mu\leq 1$, we get
  \[
  \bigl(|s|^{2\mu K_{\pm}}(K_{\pm}+sK)|s|^{-2\mu K_{\pm}}\bigr)^a_{\ b}
  = |s|^{2\mu(q_a-q_b)} (K_{\pm}+sK)^a{}_b
  = o\bigl( |s|^{2\min(0,q_a-q_b)+\alpha^a{}_b} \bigr) = o(|s|^{\alpha_0}) .
  \]
  Thus the matrix norm of all factors $2 |s|^{2\mu_iK_{\pm}}(K_{\pm}+sK)|s|^{-2\mu_iK_{\pm}}$ in~\eqref{ARproof-3} is bounded by $C|s|^{\alpha_0}$ for some~$C$, so that the $n$-th term in the sum is bounded by $C^n|s|^{n\alpha_0}/n!$ (where the $1/n!$ factor comes from the volume of $\{0\leq\mu_n\leq\dots\leq\mu_1\leq 1\}$).  In addition, we learn from this bound that $\Delta(\lambda)-\delta=o(\exp(|s|^{\alpha_0})-1)=o(|s|^{\alpha_0})$.
  This concludes the proof that~\eqref{ARproof-1} tends to the identity matrix as $s\to 0$.
  Altogether, the spacetime we constructed using the result of~\cite{AnderssonRendall} in the spacelike case is indeed a cyclic spacetime in the sense of \autoref{def:singu2}, as we expected.

\paragraph{Step 2. Timelike hypersurface.}  

By applying the same strategy as in the spacelike case, and using that the scattering map~$\Sbf$ preserves the signature of the metric and is quiescence-preserving, we reduce the problem to proving a timelike counterpart to the existence theorem of~\cite{AnderssonRendall}.
For definiteness we work on the $s>0$ side.
We summarize their proof and explain along the way how to modify it to include the signature $\epsilon=\pm 1$ (we recall our convention that $g^{(4)}=\epsilon ds^2+g(s)$ so that $\epsilon=-1$ is the spacelike case).
The key point is that~$\epsilon$ arises only in source terms in~\eqref{equa:time-ADM}, so that it does not change the analysis of the Fuchsian differential equations and, especially, their behavior in powers of~$s$.

Following~\cite{AnderssonRendall}, we fix a point $x_0\in\Hcal_0$ and prove existence and uniqueness in a small neighborhood of~$x_0$, the full solution being easily patched up from these local ones.
We fix a positive number $\alpha_0<\min(k_{a+}(x_0))/10$.
On a sufficiently small neighborhood~$U$ of~$x_0$ we can find a frame $\{e_a\}$ on~$U$, and a diagonal tensor $Q=\diag(q_1,q_2,q_3)$ in this frame, such that the spectrum of $\Kpoi-Q$ is in the small interval $(-\alpha_0/8,\alpha_0/8)$ and such that all $k_{a+}>5\alpha_0$ everywhere.\footnote{\label{foot:ARalpha0}As a side-note we point out an easily fixed mistake in~\cite{AnderssonRendall}.  In their Section~5, $\alpha_0/8$~is implicitly chosen to be smaller than any {\sl non-zero} difference $|k_{a+}(x_0)-k_{b+}(x_0)|$ so that their case distinction between ``near Friedmann'', ``near double eigenvalue'', and ``diagonalizable'' exactly matches with the number of coincident eigenvalues of $K_{\pm}$ at~$x_0$.  As we discuss momentarily, this additional restriction on~$\alpha_0$ would break the proof that constraint equations are obeyed.  Thankfully, this additional restriction on~$\alpha_0$ and its consequence $K_{\pm}(x_0)=Q(x_0)$ are not actually used in their paper.}

We parametrize the differences between exact solutions and asymptotic profiles using the same Ansatz as~\cite{AnderssonRendall}:
\bel{AR-proof5}
\aligned
(\gst^{-1}g-\delta)^a{}_b & = s^{\alpha^a{}_b} \gamma^a{}_b , \quad &
e_c(\gamma^a{}_b) & = s^{-\zeta} \lambda^a{}_{bc} , \quad &
s(K-\Kst)^a{}_b & = s^{\alpha^a{}_b} \xi^a{}_b , \\
\phi-\phist & = s^{2\zeta} \psi , &
e_c(\psi) & = s^{-\zeta} \omega_c , &
s\del_s(\phi-\phist) & = s^{2\zeta} \chi ,
\endgathered
\ee
where $\alpha^a{}_b=\alpha_0+2\max(0,q_b-q_a)>0$, $\zeta=\alpha_0/800$, and we seek solutions with $\gamma^a{}_b,\lambda^a{}_{bc},\xi^a{}_b,\psi,\omega_c,\chi=o(1)$.
Let us turn to the equations obeyed by these six {\sl remainder functions}.
Injecting the Ansatz for $(g,K,\phi)$ into the Einstein-scalar field equations in ADM formalism \eqref{equa:time-wave}--\eqref{equa:time-ADM} yields equations with first order and second order derivatives of $\gamma^a{}_b,\xi^a{}_b,\psi$.
As anticipated in \eqref{AR-proof5} we give names $\omega_c,\chi$ to weighted derivatives of $\psi$~along and across leaves of the foliation, and likewise $\lambda^a{}_{bc},\xi^a{}_b$ to weighted derivatives of~$\gamma^a{}_b$.
Then first order $s\del_s$ derivatives of the six remainder functions can be expressed in terms of their first order derivatives along leaves of the foliation:
generalized to either signature $\epsilon=\pm 1$ their equations (48) and (47) read
\bel{AR-proof6}
\aligned
s\del_s\gamma^a{}_b + \alpha^a{}_b \gamma^a{}_b + 2\xi^a{}_b - 2[\gamma,\Kpoi]^a{}_b & = - 2 s^{\alpha^a_c+\alpha^c_b-\alpha^a_b} (\gamma\xi)^a{}_b , \\
s\del_s\lambda^a{}_{bc} & = s^\zeta e_c(s\del_s\gamma^a{}_b) + \zeta s^\zeta e_c(\gamma^a{}_b) , \\
s\del_s\xi^a{}_b + \alpha^a{}_b \xi^a{}_b + \Kpoi^a{}_b \Tr\xi & = s^{\alpha_0} (\Tr\xi) \xi^a{}_b - \epsilon\, s^{2-\alpha^a{}_b} \bigl({}^{S\!}R^a{}_b - 8\pi g^{ac} e_c(\phi) e_b(\phi)\bigr) , \\
s\del_s\psi + 2\zeta \psi - \chi & = 0 , \\
s\del_s\omega_a & = s^\zeta e_a(\chi - \zeta \, \psi) , \\
s\del_s\chi + 2\zeta \chi & = s^{\alpha_0-2\zeta} (\phipoi_0 + s^{2\zeta} \chi) \Tr\xi - \epsilon \, s^{2-2\zeta} \bigl( \Delta_g \phist + \nabla_g^a \omega_a \bigr) .
\endaligned
\ee
The only differences between the spacelike and timelike cases are the sign of two source terms in~\eqref{AR-proof6} as expected from our~\eqref{equa:time-ADM}, and the signature of~$g$, which does not affect their discussion at all.

Observe that the system takes the form 
$
s\del_s U+AU=F(s,x,U,U_x),
$
where $U=(\gamma^a{}_b,\lambda^a{}_{bc},\xi^a{}_b,\psi,\omega_c,\chi)$ is the vector of remainder functions, $A$~is $s$-independent, and the source term~$F$ only involves at most first-order $x$~derivatives of~$U$.
To be precise, $s\del_s\gamma^a{}_b$ in the right-hand side of the second equation should be replaced by its value according to the first equation.
The Ricci curvature tensor ${}^{S\!}R^a{}_b$ of ${}^{S\!}g_{cd}\coloneqq \frac{1}{2}(g_{cd}+g_{dc})$ must likewise be expressed in terms of first-order derivatives of $\lambda^a{}_{bc}$ instead of second-order derivatives of~$\gamma^a{}_b$.
Finally, while $\Delta_g\phist+\nabla_g^a\omega_a$ involves second order derivatives of the asymptotic profile~$\phist$, it only involves first order derivatives of the remainder functions $\gamma^a{}_b,\omega_a$.
The most technical part of~\cite{AnderssonRendall} is to prove that~\eqref{AR-proof6} are Fuchsian equations, in the sense that the matrix~$A$ is the direct sum of a zero matrix and a matrix with positive spectrum, and that $F$~is suitably analytic and tends to zero as some positive power of~$s$.
The matrix~$A$ is does not depend on~$\epsilon$ so their analysis of its spectrum applies verbatim to the timelike case.
The bounds they derive on~$F$ as $s\to 0$ do not rely on cancellations between terms so the same bounds apply to both signs $\epsilon=\pm 1$.

Once the system is shown to be Fuchsian, a general theorem implies the existence and uniqueness of $U=(\gamma^a{}_b,\lambda^a{}_{bc},\xi^a{}_b,\psi,\omega_c,\chi)$ on a small interval $s\in (0,s_1]$ such that $U$ tends to zero as $s\to 0$.
Such a solution automatically has symmetric $g_{ab}$ and $K_{ab}$; indeed, $s^{1-\alpha_0}(K_{ab}-K_{ba})$ turns out to obey a {\sl linear} Fuchsian equation and tends to zero as $s\to 0$ hence vanishes, which then implies that $g_{ab}-g_{ba}$ is $s$-independent yet bounded by a positive power of~$s$ as $s\to 0$, hence vanishes.
At this point, we know that $(g,K,\phi)$ given in terms of $\gamma^a{}_b,\xi^a{}_b,\psi$ by the Ansatz~\eqref{AR-proof5} is a solution of the Einstein-scalar field ``evolution'' (meaning~$\del_s$) equations, but not necessarily the constraints.
If the prescribed singularity data set has ``isotropic enough'' $\Kpoi\simeq\frac{1}{3}\delta$, then (suitable rescalings of) the Hamiltonian and momentum constraints are shown to obey {\sl linear} Fuchsian equations and tend to zero as $s\to 0$ hence they vanish identically.
More precisely, this argument goes through provided all $|k_{a+}-1/3|<\alpha_0/10$.
Concretely, this establishes the existence and uniqueness in neighborhoods of points~$x_0$ at which all $|k_{a+}(x_0)-1/3|<1/500$ (say) because this inequality ensures that $\alpha_0=1/50<\min(k_{a+}(x_0))/10$ fulfills all the necessary inequalities.\footnote{Even when all $|k_{a+}-1/3|$ are very small it is not necessarily possible to fulfill the additional restriction on~$\alpha_0$ explained in \autoref{foot:ARalpha0}.  Indeed, consider $k_{1+}=1/3-a+a^2$, $k_{2+}=1/3-a-a^2$ and $k_{3+}=1/3+2a$ for a tiny $a>0$.  This is arbitrarily close to~$1/3$, but the condition $|k_{a+}-1/3|<\alpha_0/10$ of applicability of the argument for constraints requires $\alpha_0>20a$; then $|k_{1+}-k_{2+}|<\alpha_0/8$, which fits in Case II of~\cite{AnderssonRendall} even though $\Kpoi(x_0)$ is non-degenerate.  As we point out in \autoref{foot:ARalpha0} this minor oversight is corrected by simply ignoring the idea that the number of equal eigenvalues of~$\Kpoi(x_0)$ controls whether the data set should be treated as Case~I, II, or~III in their case distinction.}

The general case $\Kpoi>0$ relies on an analytic continuation argument.
Consider a singularity data set $(\gpoi,\Kpoi,\phipoi_0,\phipoi_1)$ that is quiescent ($\Kpoi>0$), and a point~$x_0$ along the singularity hypersurface.
In the neighborhood~$U$ of~$x_0$ considered previously we have that all $k_{a+}>5\alpha_0$.
From this and the constraints $\sum_a k_{a+}=1$ and $\sum_a k_{a+}^2 = 1 - 8\pi\phipoi_0^2$ one can work out that $12\pi\phipoi_0^2 > 15\alpha_0$.
For a (constant) parameter $a\in(0,1]$ we consider the data set
\bel{ARproof-rescaled-data}
\bigl(\mu^2\gpoi,\tfrac{1}{3}\delta+\mu^{-3}\Kpoicirc,a^{-1}\mu^{-3}\phipoi_0,a\phipoi_1\bigr) , \qquad
\mu = \bigl(1 + 12\pi \phipoi_0^2(a^{-2}-1)\bigr)^{1/6} .
\ee
For $a=1$ this is the original data set $(\gpoi,\Kpoi,\phipoi_0,\phipoi_1)$.
One readily checks that this choice of conformal rescaling of the metric, trace-free extrinsic curvature, and~$\phipoi_0$ by powers of~$\mu$ leaves the momentum constraint $\nablapoi_a\Kpoi^a_b=8\pi\phipoi_0\del_b\phipoi_1$ invariant, and that the expression of~$\mu$ in terms of~$a$ leaves the Hamiltonian constraint $1-\Tr\Kpoi^2=8\pi\phipoi_0^2$ invariant.
For $a\in(0,1]$ we have $\mu\geq 1$ so $\tfrac{1}{3}\delta+\mu^{-3}\Kpoicirc$ is closer than~$\Kpoi$ to the center $\tfrac{1}{3}\delta$ of the Kasner disk, hence is positive.
We have thus constructed an analytic family (parametrized by~$a$) of singularity data sets that are quiescent.
Solutions of Fuchsian equations are known to depend analytically on parameters in this context, and in particular the constraints depend analytically on~$a$.
For $a\in(0,1]$ sufficiently small we have $15\alpha_0(a^{-2}-1)>10^6$ so $\mu>10$, which makes the rescaled data~\eqref{ARproof-rescaled-data} ``isotropic enough'' in the sense above.
The constraints thus vanish identically for small~$a$, and depend analytically on~$a$, so they vanish for all~$a$.
This concludes our proof of the timelike analogue of Andersson and Rendall's result on Einstein-scalar field equations.

The analytic continuation argument in~\cite{AnderssonRendall} does not rely on our construction~\eqref{ARproof-rescaled-data}.
Instead, they use a much simpler construction that applies to stiff fluids to conclude for this type of matter, then they remark that solutions with a massless free scalar~$\phi$ give rise to (particular) solutions with stiff fluids whose velocity is the normalized gradient $\nabla\phi/|\nabla\phi|$.
Our nontrivial construction~\eqref{ARproof-rescaled-data} makes for a conceptually clearer argument, as it makes the proof for the massless scalar completely independent of that for stiff fluids.
Another motivation for our approach is that the normalized gradient of a scalar field is spacelike near timelike singularities, so that it cannot be interpreted physically as the velocity of a fluid.

\paragraph{Step 3. Behavior of the curvature.}

By construction of the foliation, $g^{(4)00}=\epsilon$ and $g^{(4)0a}=0$ so we compute
\[
R^{(4)} = - 8\pi \Tr^{(4)}(T) = 8\pi g^{(4)\alpha\beta} \del_\alpha \phi \del_\beta \phi
= 8\pi \bigl(\epsilon(\del_s\phi)^2 + g^{ab} \del_a \phi \del_b \phi \bigr).
\]
Since $s\del_s\phi\to\phi_{0\pm}$ as $s\to 0^{\pm}$, the first term behaves as $\epsilon\phi_{0\pm}^2/s^2$ as announced in the theorem.
There remains to prove that the second term does not contribute to the limit of $s^2 R^{(4)}$.
On the other hand, the asymptotic Hamiltonian constraint implies an upper bound on each eigenvalue $k_{a\pm}$ of~$K_{\pm}$:
\[
k_{a\pm}^2 \leq \Tr(K_{\pm}^2) = 1 - 8\pi\phi_{0\pm}^2 \leq (1 - 4\pi\phi_{0\pm}^2)^2 .
\]
Since $-sK\to K_{\pm}$ as $s\to 0^{\pm}$, we deduce that, for sufficiently small~$s$, eigenvalues of $-sK$ are less than $1 - 2\pi\phi_{0\pm}^2$ (say).  This, and the fact that $|s|^{2sK} g$ has a finite limit~$g_{\pm}$, enables us to bound the inverse metric as
$
g^{-1} = |s|^{2sK} \bigl(|s|^{2sK} g\bigr)^{-1} = O\bigl(|s|^{4\pi\phi_{0\pm}^2-2}\bigr) .
$
On the other hand, $\del_a\phi=\del_a\phist+o(1)=O(\log|s|)$ so we deduce as desired (provided $\phi_{0\pm}\neq 0$)
\[
g^{ab} \del_a\phi \del_b\phi = O\bigl(|s|^{4\pi\phi_{0\pm}^2-2}(\log|s|)^2\bigr) = o(s^{-2}). \qedhere 
\]

\subsection{Parametrization of the set of singularity data}
\label{subsec:param-data}

\paragraph{The space of spacelike singularity data.} 

To get a better handle on the space of singularity data, we now turn to parametrizing it, first in the spacelike setting.
Consider the conditions \eqref{equ:sing-space-T}, together with \eqref{equa-const-singu}, defining the space of spacelike singularity data $\Ispace(\Hcal)$. We drop the ``$-$'' subscripts for brevity. Hence, $g$ is an arbitrary Riemannian metric on $\Hcal$ while $K$ is a symmetric two-tensor satisfying on $\Hcal$
\[
\Tr K = 1,
\qquad
1  - \Tr(K^2)  = 8\pi \, \phi_0^2,
\qquad
\nabla_a K^a_b  = 8 \pi \, \phi_0 \del_b \phi_1. 
\]
For definiteness, we treat the case where $\Hcal$ is a compact three-manifold.
We parametrize the subspace of singularity data defined by the restriction that $\phi_0>0$ everywhere.

We adapt here the so-called conformal method, originally proposed by Lichnerowicz and recently generalized by several authors; see \cite{Maxwell} and the references therein.
It turns out to be convenient to scale the metric~$g$ and trace-free part of~$K$, by introducing a metric $\gt$ and a tensor $\Ht$ as follows: 
\bel{asym-scale-phi0}
g_{ab} = \phi_0^{-2/3} \gt_{ab}, 
\qquad K^a_b - \frac{1}{3} \delta^a_b = \phi_0 \Ht^a_b .
\ee
Then, the symmetry of $K$, the trace condition $\Tr K = 1$, and the Hamiltonian constraint in~\eqref{equa-const-singu} read
\bel{constr-tilde}
\gt_{ab} \Ht^b_c = \gt_{cb} \Ht^b_a, 
\qquad
\Ht^a_a = 0, 
\qquad
\phi_0 = \sqrt{\frac{2/3}{8\pi+\Tr\Ht^2}}, 
\ee
while the differential constraint in~\eqref{equa-const-singu} simplifies to
$
\nablat_a \Ht^a_b = 8\pi \, \del_b\phi_1.
$
Here $\nablat$ is the Levi-Civita connection of~$\gt$.  Let us also define $\nablat^b$ by $\nablat_a=\gt_{ab}\nablat^b$.

Any symmetric traceless tensor, such as $\gt \, \Ht$, can be decomposed into a symmetric tranverse-traceless tensor $\sigma$ and a vector field part $W$: 
\bel{gtKt-decom}
\gt \, \Ht = \sigma + \frac{1}{2N} \Lt W, 
\qquad \nablat^a \sigma_{ab} = 0, 
\ee
in which $N>0$ is any prescribed function and $\Lt$ denotes the conformal Killing operator of the metric~$\gt$
\[
\bigl(\Lt W\bigr)_{ab} = \nablat_a W_b + \nablat_b W_a -\frac{2}{3} \bigl(\nablat^c W_c\bigr) \gt_{ab}. 
\]
Its dual $\Lt^*$ act on symmetric, traceless tensors $A_{ab}$ and is defined as $(\Lt^* A)_b = -2 \nablat^a A_{ab}$. 

Plugging this Ansatz into our momentum equation, we obtain the elliptic system
\bel{asymptotic-constraints-elliptic}
\Bigl(\Lt^* \Bigl( \frac{1}{2N} \Lt W \Bigr)\Bigr)_b = -16 \pi \, \del_b \phi_1. 
\ee
This is a system of three equations for a vector field $W$ defined on $\Hcal$ (assumed to be compact). A unique solution $W$ exists (see Section 6.1 in \cite{HMM}) provided the right-hand side is $L^2$-orthogonal to any conformal Killing field on $\Hcal$. For instance, this is always true whenever the metric~$\gt$, or equivalently the metric~$g$, has no conformal Killing field.

Given the solution~$W$ of~\eqref{asymptotic-constraints-elliptic} and any chosen transverse-traceless tensor~$\sigma$ we obtain $\Ht$ from~\eqref{gtKt-decom}, then deduce $\phi_0$ from~\eqref{constr-tilde}.  Finally, we scale $(\gt,\Ht)$ to get $(g,K-\delta/3)$.

Choosing $N=1$ in the above for definiteness, we find the following proposition.

\begin{proposition}[Parametrization of the space of singularity data in the spacelike case] 
  On a compact $3$-manifold $\Hcal$, there is a one-to-one correspondence between the singularity data sets $(g,K,\phi_0,\phi_1)$ with $\phi_0>0$ and the triples $(\gt,\sigma,\phi_1)$ consisting of a Riemannian metric, a symmetric transverse-traceless (TT) tensor field $\sigma$ on $(\Hcal,\gt)$, and a scalar field $\phi_1$ such that $\nabla\phi_1$ is $L^2$-orthogonal to the conformal Killing fields of $(\Hcal,\gt)$ (if any exists).
\end{proposition}

\paragraph{The space of timelike singularity data.}

For timelike hypersurfaces we can proceed in a similar way as above, but a significant difference arises: the equation \eqref{asymptotic-constraints-elliptic} for the (vector-valued) unknown $W$ is now a {\sl coupled system of wave equations}. Hence, it is natural to assume that the hypersurface topology is $\Hcal \simeq I \times \Sigma^2$ where $I \subset \RR$ is an interval containing $0$, say, and $\Sigma^2$ is a two-surface.

From a singularity data set $(g,K,\phi_0,\phi_1)$ with $\phi_0>0$, we scale the metric as in~\eqref{asym-scale-phi0} to define $\gt$ and~$\Ht$.
To construct a solution~$W$ of the wave equation~\eqref{asymptotic-constraints-elliptic}, suitable initial data should be prescribed on the two-dimensional slice $\Sigma_0 = \{0\} \times \Sigma^2$, that is 
\[
W|_{\Sigma_0} = W_0, \qquad \Lcal_\nu W|_{\Sigma_0} = W_1,
\]
in which $\nu$ is a unit (for $\gt$) normal vector field along~$\Sigma_0$.
Then $\sigma$ is defined by~\eqref{gtKt-decom} as in the spacelike case.
We can thus expect a one-to-one correspondence between the singularity data sets $(g,K,\phi_0,\phi_1)$ with $\phi_0>0$ and the tuples $(\gt,\sigma,W_0, W_1, \phi_1)$ consisting of a Riemannian metric, a symmetric transverse-traceless (TT) tensor field $\sigma$ on $(\Hcal,\gt)$, and a scalar field $\phi_1$ on $\Hcal$, as well as two scalar fields $W_0, W_1$ prescribed on the surface $\Sigma^2$.

The above parametrizations in the spacelike and timelike cases are not directly used below, since we prefer to describe the scattering maps $\Sbf$ as maps defined for all singularity data including configurations where $\phi_0$~may vanish.
To use these parametrizations, we would need to require scattering maps to map any data set with a positive matter field $\phi_0$ to an image with positive matter field.
However, the conformal rescaling method used in~\eqref{asym-scale-phi0} is useful at various points later on.


\section{Classification of ultralocal scattering maps} 
\label{section---4}

\subsection{Preliminaries}

\paragraph{Organization of this section.}

Interestingly, many choices of junction are allowed by our definitions above, and it is only after {\sl additional physical input} is specified that one can decide which junction conditions are actually achieved. The same phenomenon occurs with phase interfaces in fluids undergoing phase transitions:
 an augmented physical model is required which provides us with the ``internal structure'' of the interfaces (or shock waves) and, in turn, a complete description of the global dynamics of the fluid. See \cite{LeFloch-Oslo} for a review. 
Motivated by the observations in Belinsky, Khalatnikov, and Lifshitz~\cite{BKL} that, along a singularity,  the dynamics typically decouples completely at different points,
our aim in the present section is to parametrize the class of {\sl ultralocal} scattering maps;

After some more preliminaries on scalar invariants of singularity data in this section, we introduce in \autoref{secti-23} the class of {\sl rigidly conformal scattering maps,} as we call them, which are defined as those for which $\gpoi$ and~$\gmoi$ have the same conformal class.  We then give the full classification of ultralocal maps in \autoref{secti-24}, specifically \autoref{theorem-ultralocal}.
We end in \autoref{subsec-example-comparison} by defining a more ``robust'' notion of conformality, that is, the class of {\sl conformal scattering maps} which (in contrast with rigid conformality) are independent of the implicit normalization we made when defining~$g_{\pm}$.

\autoref{theorem-ultralocal}, which classifies ultralocal maps, is established in \autoref{section---5}.
Let us highlight here some steps of our proof, to give the reader some insight on why ultralocal maps are so restricted.
Ultralocality means that the value of $(\gpoi,\Kpoi,\phipoi_0,\phipoi_1)=\Sbf(\gmoi, \Kmoi, \phimoi_0, \phimoi_1)$ at a point~$x$ only depends on $(\gmoi, \Kmoi, \phimoi_0, \phimoi_1)$ at the same point, and not on their derivatives.
As we shall show, general covariance then only allows a {\sl finite number of tensor structures}  for $\Kpoi$ (namely $\delta$, $\Kmoi$, and~$K_-^2$), and likewise for~$\gpoi$, with scalar coefficients.
Scattering maps must respect the momentum constraint, which expresses the divergence $\nabla_{\pm}\cdot K_{\pm}$ in terms of the scalar fields.  We work out that if the expression of~$\Kpoi$ includes the tensor structure~$K_-^2$, then $\nablapoi\cdot\Kpoi$ involves not only $\nablamoi\cdot\Kmoi$ but also other derivatives of~$\Kmoi$ that cannot be expressed in terms of scalar fields, thus violating the momentum constraint.  This entails a most crucial property: the trace-free part of~$\Kpoi$ is proportional to that of~$\Kmoi$.  The momentum constraint then further restricts various scalar fields appearing in the construction, which leads us to a complete classification of all ultralocal scattering maps.

When describing scattering maps in the rest of this section we generally omit the subscript~``$-$'' for brevity but keep it and the subscript~``$+$'' where necessary.
Scattering maps must be sufficiently regular to preserve the regularity of the data $(g_{\pm}, K_{\pm}, \phi_{0\pm}, \phi_{1\pm})$ in any chosen functional space.  In particular, we tacitly require our scattering maps to be sufficiently regular so that $\gpoi$ is at least continuous when $(\gmoi,\Kmoi,\phimoi_0,\phimoi_1)$ are smooth.

\paragraph{Kasner radius and angle.}

Ultralocal scattering maps are conveniently described in terms of the following parametrization of the eigenvalues $k_1,k_2,k_3$ of~$K$ at a given point $x\in\Hcal$.
Since $K$ is symmetric with respect to the quadratic form~$g$ (at~$x$), it admits eigenvectors $v_1,v_2,v_3$ that are orthogonal with respect to~$g$.  In this basis, $K$ and $g$~are diagonal.
For spacelike singularities, $g$~is Riemannian so up to rescaling the eigenvectors we obtain $g(v_a,v_b)=\delta_{ab}$, and the three eigenvalues $k_1,k_2,k_3$ are indistinguishable.  On the other hand, for timelike singularities $g$~is Lorentzian, so $g(v_a,v_b)$ can be normalized to $\diag(-1,1,1)$: the eigenvalue~$k_1$ (say) is singled out as the one with (at least) one timelike eigenvector.

In terms of the eigenvalues $k_1,k_2,k_3$ of $K$, the constraints $\Tr K=1$ and $\Tr K^2\leq 1$ describe a unit disk.
We essentially use the Jacobs parametrization~\cite{Jacobs} of this disk by polar coordinates: a {\bf Kasner radius} $r \in [0,1]$ and a {\bf Kasner angle}~$\theta$, such that
\bel{Kasner-rtheta}
k_1 - \frac{1}{3} = \frac{2}{3} r \cos\theta,
\qquad
k_2 - \frac{1}{3} = \frac{2}{3} r \cos\Bigl(\theta+\frac{2\pi}{3}\Bigr),
\qquad
k_3 - \frac{1}{3} = \frac{2}{3} r \cos\Bigl(\theta+\frac{4\pi}{3}\Bigr).
\ee
For spacelike singularities, eigenvalues are indistinguishable, so the angle $\theta$ has periodicity~$2\pi/3$ by definition.  For timelike singularities, the eigenvalue~$k_1$ is special due to the timelike eigenvector, so $\theta$ has periodicity~$2\pi$.  In both cases, mapping $\theta\to-\theta$ simply exchanges $k_2\leftrightarrow k_3$, which are indistinguishable.  When describing scattering maps it is nevertheless more convenient to keep the somewhat redundant parametrization with $\theta\in\RR$ and suitable symmetry and periodicity requirements.  As usual for polar coordinates, the value of~$\theta$ is meaningless when $r=0$.

For a $2\pi$-periodic even function $f\colon\RR\to\RR$ we introduce the notation
\bel{fTheta-notation}
f(\Theta) \coloneqq \diag\biggl(f(\theta),f\Bigl(\theta+\frac{2\pi}{3}\Bigr),
f\Bigl(\theta+\frac{4\pi}{3}\Bigr)\biggr) \qquad \text{in the basis } v_1,v_2,v_3,
\ee
which (as we explain momentarily) is well-defined for $r\neq 0$: when $r=0$, $\theta$~is completely ambiguous.
Here, $\Theta$~stands schematically for $\diag(\theta,\theta+2\pi/3,\theta+4\pi/3)$, which is ill-defined for two reasons.
\bei 

\item 
First, the Kasner angle is only defined up to changing $\theta\to-\theta$ and $\theta\to\theta+2\pi/3$ or $2\pi$ depending on signature, and such changes, together with the corresponding permutations of eigenvectors, only preserve $\Theta$ modulo $2\pi$ shifts and overall sign changes.  These ambiguities do not affect $f(\Theta)$ thanks to evenness and periodicity of~$f$.

\item Second, when (exactly) two eigenvalues coincide ($\theta=0\bmod{\pi/3}$) the basis $v_1,v_2,v_3$ is ambiguous.  This is cured since at these values of~$\theta$ the corresponding eigenvalues of~$f(\Theta)$ coincide, so that changing the basis does not affect~$f(\Theta)$.

\eei 
Using the notation~\eqref{fTheta-notation}, the equation~\eqref{Kasner-rtheta} is simply
\bel{Kasner-rtheta-2}
K = \frac{1}{3}\delta + \frac{2}{3} r \cos\Theta ,
\ee
in which the factor $r$ suppresses the ambiguity in $\cos\Theta$ when $r=0$.
It will be useful to compute various powers of the traceless extrinsic curvature $\Kcirc=K-\frac{1}{3}\delta$:
\bel{Kcirc-powers}
\Kcirc = \frac{2r}{3} \cos\Theta , \qquad
\Kcirc^2 = \frac{2r^2}{9} \bigl(\delta+\cos(2\Theta)\bigr) , \qquad
\Kcirc^3 = \frac{2r^3\cos(3\theta)}{27} \delta + \frac{r^2}{3} \Kcirc ,
\ee
and their traces
\bel{Kcirc-powers-tr}
\Tr\Kcirc = 0 , \qquad
\Tr\Kcirc^2 = \frac{2r^2}{3} , \qquad
\Tr\Kcirc^3 = \frac{2r^3\cos(3\theta)}{9} .
\ee

Finally, the Kasner radius can in fact be determined solely from $\phi_0$ thanks to the asymptotic Hamiltonian constraint in ~\eqref{equa-const-singu} and $\Tr K=1$:
\[
r^2 = \frac{3}{2} \Tr(\Kcirc^2)
= \frac{3}{2} \Tr K^2 - \Tr K + \frac{1}{2} = 1-12\pi\phi_0^2 .
\]
In particular, $\phi_0$ lies in a bounded interval and it is natural to define
\bel{r-of-phi0}
r(\phi_0) \coloneqq \sqrt{1-12\pi\phi_0^2} \qquad \text{for } \phi_0 \in I_0 \coloneqq \bigl[-1/\sqrt{12\pi},1/\sqrt{12\pi}\bigr] .
\ee
Observe that the relation cannot be inverted since the sign of~$\phi_0$ cannot be deduced from~$r$.

\paragraph{Scalar invariants of singularity data.}

An important building block in the classification is to understand what ultralocal scalar invariants the singularity data $(g,K,\phi_0,\phi_1)$ admits, namely what functions of the data at a point~$x$ (and not its derivatives) are invariant under changes of coordinates.

As explained above, the symmetry of~$K$ with respect to~$g$ ensures that $K=\diag(k_1,k_2,k_3)$ and $g=\diag(\pm 1,1,1)$ in some basis $v_1,v_2,v_3$, where the sign depends on the signature of~$g$.
Any ultralocal scalar is therefore determined by its value for such diagonal matrices, and can thus only depend on $k_1,k_2,k_3,\phi_0,\phi_1$.  Expressing the eigenvalues in terms of $(r,\theta)$, and $r$ in terms of~$\phi_0$ using~\eqref{r-of-phi0} we obtain the following lemma (recall that $\theta$ is undefined when $r=0$ namely $\phi_0=\pm 1/\sqrt{12\pi}$).

\begin{lemma}[Ultralocal scalars]\label{lem:scalars}
  Any $GL(3,\RR)$-invariant function of the singularity data $(g,K,\phi_0,\phi_1)$ at a point
  can be written as a function of the scalars $\theta,\phi_0,\phi_1$ defined in~\eqref{Kasner-rtheta-2}
  that is an even and periodic function in~$\theta$ with period $2\pi/3$ (spacelike case) or $2\pi$ (timelike case), and that is $\theta$-independent for $\phi_0=\pm 1/\sqrt{12\pi}$.
\end{lemma}

In particular, the scalar fields $(\phipoi_0,\phipoi_1)$ obtained after applying an ultralocal scattering map are ultralocal scalar invariants.
They are thus described by a function $\Phi\colon \RR\times I_0\times\RR\to I_0\times\RR$ suitably even and periodic in~$\theta$, such that $(\phipoi_0,\phipoi_1)=\Phi(\thetamoi,\phimoi_0,\phimoi_1)$.
The function~$\Phi$ plays a key role in describing the most general ultralocal scattering in \autoref{secti-24}.

\subsection{Rigidly conformal scattering maps}
\label{secti-23}

\paragraph{The notion of rigid conformality.}

As a warm-up before giving the most general ultralocal scattering map, we describe in this section all {\bf rigidly conformal scattering maps,} in the sense that $\gmoi$ and~$\gpoi$ are in the same conformal class.
Recall that $\gmoi$ and~$\gpoi$ are the values of asymptotic profiles for proper times (or proper distances) $s=\pm 1$ around the singularity hypersurface.
We introduce later a more flexible notion of conformal scattering map, defined by comparing the whole asymptotic profiles, rather than specifically their values $\gmoi$ and~$\gpoi$ at $s=\pm 1$.

\vskip.15cm

{\bf Example 1. Isotropic rigidly conformal scattering.} 
  For any ultralocal {\sl scale factor} $\lambda>0$, namely a function of $(\theta,\phi_0,\phi_1)\in\RR\times I_0\times\RR$ that is even and periodic in~$\theta$ with period $2\pi/3$ (spacelike case) or $2\pi$ (timelike case), any constant $\varphi\in\RR$ and either sign $\epsilon=\pm 1$, we introduce the map
\bel{Sirc}
  \Sirc_{\lambda,\varphi,\epsilon} \colon (g, K, \phi_0, \phi_1)
  \mapsto 
  \Big(\lambda^2 g, \ \frac{1}{3}\delta,\ {\epsilon}/{\sqrt{12\pi}},\ \varphi \Big). 
\ee
  Observe that after the bounce the three Kasner exponents are equal, hence the expansion (in the spacelike case) is isotropic.  The scalar field and extrinsic curvature are completely shielded by the singularity, except that they make the scale factor~$\lambda$ of the metric space-dependent.

  The constant $\varphi$ and sign $\epsilon$ are mostly irrelevant since they simply affect the overall sign and constant part of the asymptotic profile~$\phi_*$, and the Einstein-scalar field equations are invariant under mapping $\phi\to-\phi$ or $\phi\to(\phi+\text{constant})$.
  One could thus focus on the special case $\Sirc_\lambda\coloneqq\Sirc_{\lambda,0,+}$, but to have a complete classification we keep all parameters.

\vskip.15cm

{\bf Example 2. Anisotropic rigidly conformal scattering.}
For any differentiable function $F\colon\RR\to\RR$ with nowhere vanishing derivative, any constant $c>0$, and any sign $\epsilon=\pm 1$,
we introduce the map
\bse\label{Sarc}
\be
  \Sarc_{F,c,\epsilon} \colon (g, K, \phi_0, \phi_1)
  \mapsto \Biggl(c^2 \mu^2 g, \ \epsilon\mu^{-3}(K - \tfrac{1}{3}\delta) + \tfrac{1}{3}\delta,\
  \epsilon\mu^{-3} \frac{\phi_0}{F'(\phi_1)},\ F(\phi_1) \Biggr), 
\ee
in which  
\bel{ex2-mu}
\mu = \mu(\phi_0,\phi_1) = \bigl(1+12\pi \phi_0^2 \bigl(F'(\phi_1)^{-2} - 1\bigr) \bigr)^{1/6}.
\ee
\ese
Observe that $1-12\pi\phi_0^2 = \tfrac{3}{2}\Tr((K-\delta/3)^2) \geq 0$ and $12\pi\phi_0^2 F'(\phi_1)^{-2}\geq 0$ with equality when $\phi_0=0$, so that their sum is positive and $\mu$ is indeed well-defined and nonzero.
As we will see, if $\epsilon=+1$ and $F$~is contracting ($\abs{F'}\leq 1$), then $\mu\geq 1$ so the scattering map brings the Kasner exponents closer to the isotropic case~$\tfrac{1}{3}$.

As in the isotropic case, changing $F\to -F$ or shifting it by a constant is mostly irrelevant due to the Einstein-scalar field equations being invariant under mapping $\phi\to-\phi$ or $\phi\to(\phi+\text{constant})$.  In contrast, the sign $\epsilon$ in $\Sarc_{F,c,\epsilon}$ affects $K$ hence has a very strong effect on the asymptotic profile $g_*(s)=|s|^{2\Kpoi}\gpoi$.

\paragraph{Special cases, limits, and regularity.}

\bei
\item A special case is that $\Sarc_{\Id,1,+}$ (where $\Id$ denotes here the identity map $\RR\to\RR$) is the identity map, which we dub {\bf continuous scattering}.  Another interesting case that plays a role in \autoref{prop-conformal-ultralocal} below is the momentum reversing
\[
\Sarc_{\Id,1,-}\colon(g,K,\phi_0,\phi_1)\mapsto\bigl(g,\,\tfrac{2}{3}\delta-K,\,-\phi_0,\,\phi_1\bigr) .
\]

\item
  The function~$F$ must be {\sl monotonic}.
  For the anisotropic scattering $\Sarc_{F,c,\epsilon}$ to lead to a continuous~$\gpoi$, one needs $\mu$ to be continuous, which requires in particular $F'$ to be continuous.
  Since $F'$ is nowhere vanishing, it must be either positive or negative everywhere, hence forcing $F$ to be monotonic.
  The precise regularity condition to impose on~$F$ for $\Sarc_{F,c,\epsilon}$ (likewise $\lambda$ for $\Sirc_{\lambda,\varphi,\epsilon}$) depends on the chosen regularity of singularity data sets.

\item Some~$\Sirc$ can essentially be obtained as limits of~$\Sarc$.
  For this, consider the limit of $\Sarc_{\varphi+c^3F,c,\epsilon}$ as $c\to 0$ for some monotonic $F\colon\RR\to\RR$ with nowhere vanishing derivative.  For $\phi_0>0$ or $\phi_0<0$, the limit is well-defined and coincides with an isotropic rigidly conformal scattering
\[
  \lim_{c\to 0} \Sarc_{\varphi+c^3F,c,\epsilon}(g, K, \phi_0, \phi_1)
  = \Sirc_{\lambda,\varphi,\epsilon\sgn\phi_0\sgn F'}(g, K, \phi_0, \phi_1)
\]
with $\lambda^6 = 12\pi \phi_0^2 F'(\phi_1)^{-2}$.  Observe that the limit is discontinuous and ill-defined whenever $\phi_0$ vanishes.
\eei

\paragraph{Classification of rigidly conformal ultralocal singularity scattering maps.} 

In fact, the following proposition states that examples 1 and 2 {\sl cover all possible classes of such scattering maps.} Interestingly, the classification is the {\sl same} for spacelike and for timelike maps, except for the different $\theta$ periodicity of~$\lambda$ in~$\Sirc_{\lambda,\varphi,\epsilon}$.  Later on, in \autoref{theorem-ultralocal}, we state the more general classification of ultralocal scattering maps that are not necessarily rigidly conformal.

\begin{proposition}[Rigidly conformal ultralocal scattering maps for self-gravitating scalar fields]
\label{prop-conformal-ultralocal}
A spacelike or timelike scattering map~$\Sbf$ that is rigidly conformal and ultralocal 
is either $\Sirc_{\lambda,\varphi,\epsilon}$ or $\Sarc_{F,c,\epsilon}$ defined in~\eqref{Sirc} and~\eqref{Sarc} above.
Among these maps, one distinguishes several subclasses:
\bei
\item Quiescence-preserving maps are $\Sirc_{\lambda,\varphi,\epsilon}$ and $\Sarc_{F,c,+}$ with $0<\abs{F'}\leq 1$ identically.

\item Idempotent maps are $\Sarc_{F,1,\epsilon}$ with $F\circ F=\Id$ and nowhere vanishing~$F'$,
  which implies that $F=\Id$ or $F'<0$ everywhere.

\item Shift-covariant maps are $\Sirc_{\lambda,\varphi,\epsilon}$ with $\phi_1$-independent~$\lambda$, and $\Sarc_{F,c,\epsilon}$ with $F''=0$.

\item Quiescence-preserving idempotent maps are  $\Sarc_{\Id,1,+}=\Id$ and $\Sarc_{(y\mapsto\varphi-y),1,+}\!\colon(g,K,\phi_0,\phi_1)\mapsto(g,K,-\phi_0,\varphi-\phi_1)$.
  They are automatically shift-covariant.
\eei 
\end{proposition}

\paragraph{Proof of which maps are quiescence-preserving, idempotent and/or shift-covariant.}

We defer to \autoref{section---5} the proof that rigidly conformal ultralocal scattering maps are \eqref{Sirc} or~\eqref{Sarc}.  For now we determine which of these maps are quiescence-preserving, idempotent, and/or shift-covariant.  Idempotence is primarily relevant for the case of timelike singularities, but the classification is independent of signature.
It is convenient to re-introduce in this proof the notation $g_\pm$, etc.\@ to distinguish between the two sides of the singularity hypersurface.
All tensors are considered at a given point $x\in\Hcal$, which we omit from notations.

Recall that $\Sbf$ is {\sl quiescence-preserving} if $\Kmoi>0$ (quiescent data) implies $\Kpoi>0$.
Manifestly, $\Sirc_{\lambda,\varphi,\epsilon}$ is quiescence-preserving since $\Kpoi=\tfrac{1}{3}\delta>0$ regardless of~$\Kmoi$.  To show that $\Sarc_{F,c,\epsilon}$ is quiescence-preserving when $\epsilon=+1$ and $0<\abs{F'}\leq 1$ identically, note that $F'(\phimoi_1)^{-2}-1\geq 0$ so
\[
\mu(\phimoi_0,\phimoi_1)^{-3} = \bigl(1+12\pi\phimoi_0^2(F'(\phimoi_1)^{-2}-1)\bigr)^{-1/2} \leq 1
\]
identically, thus $\Kpoi=\mu^{-3}(\Kmoi-\tfrac{1}{3}\delta)+\tfrac{1}{3}\delta$ is on the line segment joining $\tfrac{1}{3}\delta$ and~$\Kmoi$.  Since the triangle defined by $\Tr K=1$ and $K>0$ is convex and $\tfrac{1}{3}\delta$ lies in it, we conclude that $\Kmoi>0$ implies $\Kpoi>0$.
Conversely, let us prove next that $\Sarc_{F,c,\epsilon}$ is otherwise not quiescence-preserving.
For this we consider a configuration on $\Hcal=\RR^3$ with constant $(\gmoi,\Kmoi,\phimoi_0,\phimoi_1)$ where $\gmoi$ is flat Euclidean or Minkowski and $\Kmoi=\diag(1-2\xi,\xi,\xi)$ for some $\xi\in(0,\tfrac{1}{2})$.  There are two cases to study.
\bei
\item If $\epsilon=-1$, we consider the limit $\xi\to 0$ with $\phimoi_1$ fixed.  In this limit, $\phimoi_0^2\to 0$ so $\mu(\phimoi_0,\phimoi_1)\to\epsilon=-1$ so $\Kpoi\to\diag(-\tfrac{1}{3},\tfrac{2}{3},\tfrac{2}{3})$, hence $\Kpoi\not>0$ for sufficiently small $\xi>0$.
\item If $\abs{F'(y)}>1$ for some $y\in\RR$, we take $\phimoi_1=y$ identically and $\xi\to\tfrac{1}{2}$: in this limit $12\pi\phimoi_0^2\to\tfrac{3}{4}>0$ so $\mu^{-3}>1$ and the first diagonal entry in $\Kpoi=\mu^{-3}(\Kmoi-\tfrac{1}{3}\delta)+\tfrac{1}{3}\delta$ tends to $\tfrac{1}{3}(1-\mu^{-3})<0$.  For $\xi$ sufficiently close to $\tfrac{1}{2}$ we get $\Kpoi\not>0$.
\eei

Recall that $\Sbf$ is {\sl idempotent} if $\Sbf\circ\Sbf$ is the identity.  Since $\Sirc_{\lambda,\varphi,\epsilon}(g,K,\phi_0,\phi_1)$ is independent of the eigenvectors of~$K$, it is not injective, let alone idempotent.  To determine when $\Sarc_{F,c,\epsilon}$ is idempotent, it is useful to note that
\[
\Sarc_{F_2,c_2,\epsilon_2} \circ \Sarc_{F_1,c_1,\epsilon_1}
= \Sarc_{F_3,c_3,\epsilon_3}
\]
where $c_3=c_2c_1$, $\epsilon_3=\epsilon_2\epsilon_1$, and $F_3=F_2\circ F_1$.
The resulting scattering map is the identity if and only if $F_3$ is the identity, $c_3=1$, and $\epsilon_3=+1$.
Thus, $\Sarc_{F,c,\epsilon}$ is idempotent if and only if $c=1$ and $F\circ F=\Id$.
The condition can be refined in the case $F'>0$.  In that case we can show $F=\Id$: indeed, if for any $y\in\RR$ we have $y<F(y)$ (resp.\@ $y>F(y)$) then applying the increasing function~$F$ implies the opposite inequality $F(y)<F(F(y))=y$ (resp.\@ $F(y)>F(F(y))=y$).  Once we know $F=\Id$ and $c=1$ there are only two maps, $\Sbf=\Sarc_{\Id,1,+}=\Id$ and $\Sbf=\Sarc_{\Id,1,-}$.  In contrast, the case $F'<0$ features a large family of idempotent scattering maps, as there are many strictly decreasing idempotent functions $F$ on~$\RR$.

Next, we consider maps that are {\sl quiescence-preserving and idempotent}.  The identity map $\Sarc_{\Id,1,+}=\Id$ clearly is.  For $F'<0$ we need to understand the interplay of $F\circ F=\Id$ and $\abs{F'}\leq 1$.  The latter condition states that $F$ is a map that reduces distances.  In order for it to be idempotent, it should thus preserve distances: such isometries of~$\RR$ are translations and reflections.  Due to $F'\leq 0$ we are left only with reflections $F(y)=\varphi-y$, as described in the proposition.

Finally, we study shift-covariant maps, such that shifting $\phimoi_1\to\phimoi_1+\varphi$ shifts $\phipoi_1\to\phipoi_1+a\varphi$ for some constant~$a$ and leaves $(\gpoi,\Kpoi,\phipoi_0)$ untouched.
In the case of $\Sirc_{\lambda,\varphi,\epsilon}$, this simply means that the conformal factor $\lambda(\thetamoi,\phimoi_0,\phimoi_1)$ multiplying the metric must be invariant under shifts of its last argument.
For $\Sarc_{F,c,\epsilon}$ we need $\mu$ defined by~\eqref{ex2-mu} to be $\phimoi_1$-independent, hence need $F'(\phimoi_1)$ to be a constant, as stated in the theorem.  It is easy to check that such affine~$F$ lead to a shift-covariant map.

The second part of \autoref{prop-conformal-ultralocal} is thus proven, while the proof of the classification of rigidly conformal ultralocal scattering maps will be done later in \autoref{section---5} together with the classification of all ultralocal scattering maps.

\subsection{General ultralocal scattering maps} 
\label{secti-24} 

\paragraph{Example 3. Isotropic ultralocal scattering.}

We now generalize the examples above to arbitrary ultralocal scattering maps, distinguishing again isotropic and anisotropic scattering maps.
For any three functions $\alpha_0,\alpha_1,\alpha_2$ of $(\theta,\phi_0,\phi_1)\in\RR\times I_0\times\RR$ that are even in~$\theta$ and $2\pi/3$-periodic (spacelike case) or $2\pi$-periodic (timelike case), and are such that $\del_\theta\alpha_0,\alpha_1,\alpha_2$ vanish at the boundary $\phi_0=\pm 1/\sqrt{12\pi}$ of~$I_0$, and for any constant $\varphi\in\RR$ and sign $\epsilon=\pm 1$, we introduce the map
\bel{Siso}
\Siso_{\alpha_0,\alpha_1,\alpha_2,\varphi,\epsilon}
\colon (g,K,\phi_0,\phi_1)
\mapsto \biggl( \exp\Bigl(\alpha_0\,\delta + \alpha_1 \cos\Theta + \alpha_2 \cos(2\Theta)\Bigr) g , \ \frac{1}{3}\delta , \ \frac{\epsilon}{\sqrt{12\pi}}, \ \varphi \biggr), 
\ee
where $\cos\Theta$ and $\cos(2\Theta)$ are defined through~\eqref{fTheta-notation}.
In the spacelike case the map can equivalently be written as
\bse
\label{Siso-space}
\begin{gather}
  \Siso_{\lambda,\varphi,\epsilon}
  \colon (g,K,\phi_0,\phi_1)
  \mapsto \biggl( \lambda(\Theta,\phi_0,\phi_1)^2 g , \ \frac{1}{3}\delta , \ \frac{\epsilon}{\sqrt{12\pi}}, \ \varphi \biggr) , \\
  \lambda(\theta,\phi_0,\phi_1)\coloneqq \exp\biggl(\frac{1}{2}\bigl(\alpha_0(\theta,\phi_0,\phi_1) + \alpha_1(\theta,\phi_0,\phi_1)\cos\theta + \alpha_2(\theta,\phi_0,\phi_1) \cos(2\theta)\bigr)\biggr) ,
\end{gather}
\ese
where $\lambda$ is an arbitrary positive $2\pi$-periodic even function that becomes $\theta$-independent along the boundary~$\del I_0$.  To retrieve the $2\pi/3$-periodic functions $\alpha_0,\alpha_1,\alpha_2$ from~$\lambda$ one can decompose $2\log\lambda$ into Fourier modes.

For both spacelike and timelike hypersurfaces, $\Siso$~reduces to its rigidly conformal case~$\Sirc$ upon setting $\alpha_1=\alpha_2=0$ in~\eqref{Siso}, or in the description~\eqref{Siso-space} imposing $2\pi/3$-periodicity of~$\lambda$ so as to make $\lambda(\Theta,\phi_0,\phi_1)$ into a multiple of the identity matrix.  In contrast to the rigidly conformal case, the metric gets generally both scaled and sheared upon crossing the singularity.

Our comments about $\Sirc$ in Example 1 are equally applicable to the general isotropic ultralocal scattering~$\Siso$.  The name ``isotropic'' stems from how the expansion after the bounce is isotropic,  given that the three Kasner exponents are equal.  Both $\phi$ and~$K$ are shielded by the singularity, except for their effect on how the metric transforms.  Again, the constant $\varphi$ and sign $\epsilon$ are mostly irrelevant due to how the Einstein-scalar field equations are invariant under mapping $\phi\to-\phi$ or $\phi\to(\phi+\text{constant})$.

\paragraph{Phase space and canonical transformation.}

The initial data for matter on the singularity hypersurface consists of two scalar fields $\phi_0,\phi_1$.  Since the evolution is ultralocal near the singularity, it is natural to consider the phase space $I_0\times\RR$ in which $(\phi_0,\phi_1)$ can take values at each point.  Here, $I_0\coloneqq[-1/\sqrt{12\pi},1/\sqrt{12\pi}]$ is the interval~\eqref{r-of-phi0} in which the momentum $\phi_0$ can vary.
We also recall from~\eqref{r-of-phi0} that $r(\phi_0)=(1-12\pi\phi_0^2)^{1/2}$ vanishes when momentum is maximal ($\phi_0=\pm 1/\sqrt{12\pi}$), namely along the boundary of $I_0\times\RR$.
Our construction below is based on the following symplectic form (or volume form) on the interior of the phase space:
\bel{phase-space-measure}
d\biggl(\frac{\phi_0}{r(\phi_0)}\biggr)\,d\phi_1
= \frac{d\phi_0 d\phi_1}{r(\phi_0)^3} .
\ee

\begin{definition}[Canonical transformation for matter]\label{def:canonical-transfo}
An {\bf $\epsilon$-canonical transformation} for the matter is a function $\Phi=(\Phi_0,\Phi_1)\colon \RR\times I_0\times\RR\to I_0\times\RR$ obeying the following properties for some sign~$\epsilon$:
\bei
\item[(i)] {\bf Periodic.} The image $\Phi(\theta,\phi_0,\phi_1)$ is $2\pi/3$ (spacelike case) or $2\pi$ (timelike case) periodic and even in~$\theta$, and at the boundary $\phi_0=\pm 1/\sqrt{12\pi}$ it is $\theta$-independent.

\item[(ii)] {\bf Maximal-momentum\footnote{We could also say that $\Phi$ is isotropy preserving since maximal momentum $\phi_0=\pm1/\sqrt{12\pi}$ corresponds to $K=\frac{1}{3}\delta$.} preserving.}
The function~$\Phi$ maps boundary to boundary and interior to interior, in the sense that $r(\Phi_0(\theta,\phi_0,\phi_1))=0$ if and only if $r(\phi_0)=0$.
Moreover, the ratio $r(\Phi_0)/r(\phi_0)$ has finite limits as $\phi_0\to\pm 1/\sqrt{12\pi}$ for each $(\theta,\phi_1)$, and these limits are $\theta$-independent and nowhere vanishing.

\item[(iii)] {\bf Volume preserving.} 
For each $\theta\in\RR$, the map $\Phi(\theta,\,.\,,\,.\,)$ is volume-preserving
 for the measure~\eqref{phase-space-measure},
 namely is a {\rm canonical transformation} of (the interior of) the phase space $I_0\times\RR$, up to the sign~$\epsilon$.  Explicitly,
 \[
 r(\Phi_0)^{-3} \bigl(\del_{\phi_0}\Phi_0\, \del_{\phi_1} \Phi_1 - \del_{\phi_1}\Phi_0\, \del_{\phi_0} \Phi_1\bigr) = \epsilon\, r(\phi_0)^{-3} .
 \]
 
\item[(iv)] {\bf Regular at boundaries.}
For each $(\theta,\phi_1)\in\RR^2$,
$\dfrac{\Phi_0}{r(\Phi_0)}\del_\theta\Phi_1\to 0$ and
$\dfrac{\Phi_0}{r(\Phi_0)} \del_{\phi_1}\Phi_1 - \epsilon \dfrac{\phi_0}{r(\phi_0)}\to 0$ at the boundaries $\phi_0\to\pm 1/\sqrt{12\pi}$.
In addition, the integral
$\displaystyle\int_{I_0}\frac{\Phi_0}{r(\Phi_0)} \del_{\phi_0} \Phi_1\, d\phi_0$, which is $\phi_1$-independent due to the other conditions, vanishes for all~$\theta$.

\eei
\end{definition}

\paragraph{Example 4. Anisotropic ultralocal scattering.}

We now define a scattering map
\bse\label{Sani}
\be
\Sani_{\Phi,c,\epsilon} \colon (\gmoi, \Kmoi, \phimoi_0, \phimoi_1)
\mapsto (\gpoi , \Kpoi , \phipoi_0 , \phipoi_1)
\ee
that depends on a constant {\bf scale factor} $c>0$, a sign $\epsilon=\pm 1$, and an $\epsilon$-canonical transformation $\Phi\colon \RR\times I_0\times\RR\to I_0\times\RR$ obeying conditions (i)--(iv) above with the sign~$\epsilon$.
First, the scalar fields are given by~$\Phi$ as
\be
(\phipoi_0,\phipoi_1) = \Phi(\thetamoi,\phimoi_0,\phimoi_1) ,
\ee
where the Kasner angle~$\thetamoi$ is given by the parametrization~\eqref{Kasner-rtheta} of eigenvalues of~$\Kmoi$.
From $\phi_{0\pm}$ one determines the Kasner radii
\bel{Sani-rpmoi}
r_{\pm} = r(\phi_{0\pm}) = \sqrt{1-12\pi\phi_{0\pm}^2} .
\ee
Second, the trace-free extrinsic curvature $\Kcirc=K-\tfrac{1}{3}\delta$ is continuous through the bounce, up to a scaling:
\bel{Sani-Kpoi}
\Kpoicirc = 0 \text{ for } \rmoi=0, \qquad \frac{\Kpoicirc}{\rpoi} = \epsilon \frac{\Kmoicirc}{\rmoi} \text{ for }\rmoi\neq 0.
\ee
The first case is imposed by condition (ii) since $\rmoi=0$ implies $\rpoi=0$ thus $\Kpoicirc=0$.  Condition (ii) also states that $\rmoi\neq 0$ implies $\rpoi\neq 0$, which makes the second equality well-defined.
Observe that the allowed scaling factors $\pm\rpoi/\rmoi$ are the only ones consistent with $\Tr(\mathring{K}_{\pm}^2)=2r_{\pm}^2/3$.
The proportionality has two consequences: $\Kpoi,\Kmoi$ share their eigenvectors and $\thetapoi-\thetamoi=0$ (for $\epsilon=+1$) or $\pi$ (for $\epsilon=-1$) modulo $2\pi$.
Importantly, even in the spacelike case where $\thetapoi$ and $\thetamoi$ are only defined modulo $2\pi/3$ (due to permuting eigenvectors), their difference is defined modulo $2\pi$ since one can compare eigenvalues of $\Kpoi$ and $\Kmoi$ on the same eigenvectors.
Third, an auxiliary function $\xi\colon \RR\times I_0\times\RR$ is given explicitly by
\bel{Sani-kappa}
\xi(\thetamoi,\phimoi_0,\phimoi_1)
= - \int_{-1/\sqrt{12\pi}}^{\phimoi_0} \frac{\Phi_0(\thetamoi,y,\phimoi_1)}{r(\Phi_0(\thetamoi,y,\phimoi_1))} \del_y \Phi_1(\thetamoi,y,\phimoi_1) \,dy.
\ee
Boundary regularity (iv) shows that $\xi$ vanishes at both boundaries $\phimoi_0=\pm 1/\sqrt{12\pi}$.
Just as $\Phi$~is, the function $\xi$~is manifestly even and ($2\pi/3$ or $2\pi$) periodic in~$\thetamoi$, and $\theta$-independent (in fact vanishing) at the boundary $\rmoi=0$.
Thanks to the volume-preserving condition~(iii), it is a solution of
\bel{Sani-dkappa}
\del_{\phimoi_0} \xi + \frac{\phipoi_0}{\rpoi} \del_{\phimoi_0} \phipoi_1 = 0 ,
\qquad
\qquad
\del_{\phimoi_1} \xi + \frac{\phipoi_0}{\rpoi} \del_{\phimoi_1} \phipoi_1 = \epsilon \frac{\phimoi_0}{\rmoi} .
\ee

Finally, for $\rmoi,\rpoi\neq 0$ the metric is scaled along the three eigenvectors of $\Kpoi,\Kmoi$:
\bel{Sani-gpoi}
\gpoi = c^2 \biggl(\frac{\rmoi}{\rpoi}\biggr)^{2/3}
\exp\biggl(16\pi\epsilon\xi\cos\Theta_-
-16\pi\epsilon\bigl( \del_{\thetamoi} \xi + \frac{\phipoi_0}{\rpoi} \del_{\thetamoi} \phipoi_1 \bigr)
\sin\Theta_-\biggr) \gmoi
\ee
where $\Theta_-=\diag(\thetamoi,\thetamoi+2\pi/3,\thetamoi+4\pi/3)$ in an eigenbasis of~$\Kmoi$.  As discussed above~\eqref{fTheta-notation}, $\Theta_-$~is rather ambiguous but $\cos\Theta_-=\tfrac{3}{2}\Kmoicirc/\rmoi$ is well-defined away from $\rmoi=0$.  While $\sin\Theta_-$ is ill-defined, being odd under $\thetamoi\to-\thetamoi$, this sign ambiguity is precisely fixed by the fact that the $\del_{\thetamoi}$ derivatives of the even functions~$\xi$ and $\phipoi_1$ are odd as well.
The value of $\gpoi$ when $\rmoi=\rpoi=0$ is determined as the $\phimoi_0=\pm 1/\sqrt{12\pi}$ limit of~\eqref{Sani-gpoi}.  The limit $(\lim\rmoi/\rpoi)$ is well-defined and non-zero by condition~(ii).  Exponentials drop out thanks to $\xi=0$ and condition~(iv) at the boundary, so $\gpoi=c^2(\lim\rmoi/\rpoi)^{2/3}\gmoi$.
\ese

\paragraph{Remarks on special cases and limits.}

We will see in \autoref{theorem-ultralocal} that there are no other ultralocal scattering maps. Let us make a few preliminary comments.
\bei

\item Example 3 reduces to Example 1 as follows.  The isotropic map
  $\Siso_{\alpha_0,\alpha_1,\alpha_2,\varphi,\epsilon}$ defined
  in~\eqref{Siso} is rigidly conformal if and only if the matrix
  multiplying the metric is a multiple of the identity matrix, namely if
  and only if $\alpha_1=\alpha_2=0$.  In this case it coincides with
  $\Sirc_{\lambda,\varphi,\epsilon}$ with $\lambda=\exp(\alpha_0/2)$.

\item Example 4 reduces to Example 2 as follows.
  The anisotropic map $\Sani_{\Phi,c,\epsilon}$ defined in~\eqref{Sani} is rigidly conformal if and only if the metric is scaled by a scalar namely the exponential in~\eqref{Sani-gpoi} is trivial.
  This requires the auxiliary function~$\xi$ to vanish and $\del_\theta\Phi_1=0$.
  The function $\xi$~\eqref{Sani-kappa} can only vanish identically if $\Phi_1$ is $\phi_0$-independent as well, namely $\phipoi_1=F(\phimoi_1)$ for some function~$F$.
  Then volume preservation gives that $(\phipoi_0/\rpoi)\del_{\phimoi_1}\phipoi_1 - \epsilon\phimoi_0/\rmoi$ is $\phimoi_0$-independent, but boundary regularity imposes that it vanishes at $\phimoi_0=\pm 1/\sqrt{12\pi}$ so it vanishes throughout.
  This fixes $\Phi$ completely in terms of~$F$.
  One then easily works out that the scattering map coincides with $\Sarc_{F,c,\epsilon}$ given in~\eqref{Sarc}.

\item As in its rigidly conformal special case $\Sarc_{F,c,\epsilon}$,
  the only effect of changing $\Phi\to-\Phi$ or shifting $\Phi_1$ by a
  constant in $\Sani_{\Phi,c,\epsilon}$ is to map $\phi\to-\phi$ or
  shift $\phi$ after the scattering, both of which are symmetries of
  the Einstein-scalar field equations.

\eei

\begin{figure}[t]\centering
  \begin{tikzpicture}[scale=2.5]
    \draw[thick] (0,0) circle (1);
    \fill[opacity=.5,color=black!20!white] (0:1) -- (120:1) -- (-120:1) -- cycle;
    \draw[dotted] (0:1) -- (120:1) -- (-120:1) -- cycle;
    \draw[->] (0,0) -- (1.3,0);
    \draw (0,-.05) -- (0,.05);
    \draw (0,0) -- (150:1.3);
    \draw ({-.5+.05*cos(60)},{-.5*tan(150)+.05*sin(60)}) -- ({-.5-.05*cos(60)},{-.5*tan(150)-.05*sin(60)});
    \draw[->] (3:1.1) arc (3:147:1.1) node [pos=.9,above] {\scriptsize $\theta$};
    \node at (.1,-.15) {\scriptsize $\aligned|\phi_0|&{=}1/\!\sqrt{12\pi}\\[-1ex]r&{=}0\endaligned$};
    \node at (-.2,.35) {\scriptsize $\aligned|\phi_0|&{=}\phimin(\theta)\\[-1ex]r&{=}1/(2\cos\widehat\theta)\endaligned$};
    \node at (-1.07,-.36) {\scriptsize $\phi_0{=}0$};
    \node at (-120:.6) {\scriptsize $K>0$};
  \end{tikzpicture}
  \caption{\label{fig:Kasner-triangle}{\bf Kasner disk and triangle.} The disk consists of points with $r^2=1-12\pi\phi_0^2\leq 1$, and the (shaded) triangle is characterized by $K>0$ or equivalently by $|\phi_0|>\phimin(\theta)$ with $\phimin$ defined in~\eqref{phimin}.}
\end{figure}
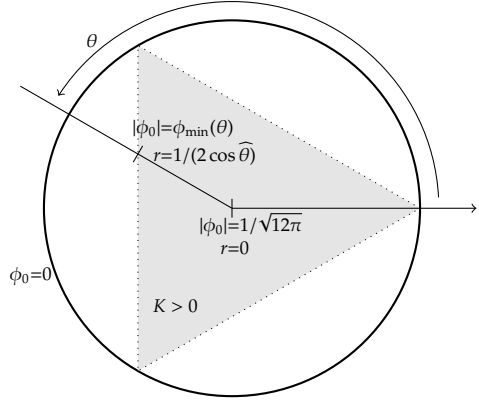

\paragraph{Classification of ultralocal scattering maps.}

To state the following theorem we introduce the notation~$\widehat{\theta}$ and the $2\pi/3$-periodic function~$\phimin$ that gives the value of~$\phi_0$ for which constraints impose that (at least) one eigenvalue of~$K$ vanishes:
\bel{phimin}
\phimin(\theta) = \biggl(\frac{1}{12\pi}\biggl(1-\frac{1}{4\cos^2\widehat{\theta}}\biggr)\biggr)^{1/2} ,
\qquad
\text{where} \quad
\widehat{\theta} = \theta - \frac{2\pi}{3} \biggl\lfloor\frac{3\theta}{2\pi}-1\biggr\rfloor \in \biggl[\frac{2\pi}{3},\frac{4\pi}{3}\biggr).
\ee
The geometric meaning of $\phimin$ is clarified by \autoref{fig:Kasner-triangle}.
The classification we find, and its refinements under various conditions, are the same for spacelike and timelike scattering maps, except for the $\theta$-periodicities: $2\pi/3$ in the spacelike case and $2\pi$ in the timelike case.

\begin{theorem}[Ultralocal scattering maps for self-gravitating scalar fields]
\label{theorem-ultralocal}
Spacelike or timelike, ultralocal scattering maps~$\Sbf$ are either $\Siso_{\alpha_0,\alpha_1,\alpha_2,\varphi,\epsilon}$ or $\Sani_{\Phi,c,\epsilon}$ defined in \eqref{Siso} and~\eqref{Sani}, respectively.  Among these maps, one distinguishes several subclasses:

\bei
\item Quiescence-preserving maps are $\Siso_{\alpha_0,\alpha_1,\alpha_2,\varphi,\epsilon}$ and $\Sani_{\Phi,c,\epsilon}$ under the condition that for all $\theta,\phi_0,\phi_1$ with $\abs{\phi_0} > \phimin(\theta+\pi\delta_{\epsilon=-1})$ one has $\abs{\Phi_0(\theta,\phi_0,\phi_1)} > \phimin(\theta)$.

\item Invertible maps are $\Sani_{\Phi,c,\epsilon}$ such that $\Phi(\theta,\,\cdot\,,\,\cdot\,)\colon I_0\times\RR\to I_0\times\RR$ is bijective for each~$\theta$.  Their inverse is $\Sani_{\Psi,1/c,\epsilon}$ with $\Psi(\theta,\,\cdot\,,\,\cdot\,)=\Phi(\theta+\pi\delta_{\epsilon=-1},\,\cdot\,,\,\cdot\,)^{-1}$ for each~$\theta$.
  They are idempotent if $c=1$ and $\Phi(\theta+\pi\delta_{\epsilon=-1},\,.\,,\,.\,)\circ\Phi(\theta,\,.\,,\,.\,)=\Id$ for all~$\theta$.

\item Shift-covariant maps are $\Siso_{\alpha_0,\alpha_1,\alpha_2,\varphi,\epsilon}$ with $\phi_1$-independent $\alpha_0,\alpha_1,\alpha_2$, and $\Sani_{\Phi,c,\epsilon}$ given in~\eqref{Sani-shift-cov}, below, which states that $\Phi_0/r(\Phi_0)=\epsilon a^{-1}\phi_0/r(\phi_0)$ and $\Phi_1=af+a\phi_1$ for some non-zero $a\in\RR$ and some suitably regular function $f=f(\theta,\phi_0)$.
  Among these, $\Siso$ are quiescence-preserving and non-invertible, while
  $\Sani$ are
  \bei
  \item invertible for any $f,a,c,\epsilon$,
  \item quiescence-preserving if and only if $|a|\leq 1$ and $\epsilon=+1$,
  \item quiescence-preserving and have quiescence-preserving inverse if and only if $a=\pm 1$ and $\epsilon=+1$,
  \item idempotent if and only if $a=\pm 1$, $c=1$, $\epsilon=\pm 1$, and $f(\theta,\phi_0)=-af(\theta+\pi\delta_{\epsilon=-1},\epsilon a \phi_0)$ for all $(\theta,\phi_0)$.
  \eei

\item Momentum-preserving $(\epsilon=+1)$ or momentum-reversing $(\epsilon=-1)$ maps are $\Sani_{\Phi,c,\epsilon}$ with $\Phi=\pm(\epsilon\phi_0,\phi_1+f(\theta,\phi_0))$ for some suitably regular function~$f$.
  They are invertible shift-covariant maps, with $a=\pm 1$.
\eei
\end{theorem}

This generalizes the classification of rigidly conformal ultralocal maps in \autoref{prop-conformal-ultralocal}.
We observe that both families of ultralocal scattering maps simply scale the densitized trace-free extrinsic curvature $\sqrt{|g|}\,\Kcirc$ by a constant factor~$\gamma$, with $\gamma=0$ for~$\Siso_{\alpha_0,\alpha_1,\alpha_2,\varphi,\epsilon}$ and $\gamma=\epsilon c^3$ for~$\Sani_{\Phi,c,\epsilon}$.

\begin{corollary}[Scaling of trace-free second fundamental form]
\label{cor:Kg-scaling}
Under an ultralocal singularity scattering map~$\Sbf$, the trace-free part of $\sqrt{|g|}\,K$ scales by some constant $\gamma\in\RR$ that depends only on~$\Sbf$:
\be
\sqrt{|\gpoi|}\,\Bigl( \Kpoi - \tfrac{1}{3}\delta \Bigr)
= \gamma\sqrt{|\gmoi|}\,\Bigl( \Kmoi - \tfrac{1}{3}\delta \Bigr) .
\ee
\end{corollary}

Different applications call for imposing different restrictions on the scattering maps.
Our local existence theory obtained in \autoref{theo:391} requires quiescence-preserving maps, to avoid BKL oscillations.
Our global existence theory in plane symmetry in~\cite{LLV-3}.
treats both sides of timelike singularity hypersurfaces on an equal footing, hence requires invertible maps, and we focus for definiteness on momentum-preserving maps, also characterized as quiescence-preserving shift-covariant maps whose inverse has the same properties.
As per \autoref{theorem-ultralocal} these are $\Sani_{\Phi,c,+}$ with $\Phi_0=\pm\phi_0$ and $\Phi_1=\pm(f(\theta,\phi_0)+\phi_1)$.

\paragraph{Proof of which maps are quiescence-preserving, idempotent, shift-covariant, etc.}

Proving the classification of scattering maps is somewhat involved, so we delay it until \autoref{section---5}.  For now we prove the second part of the theorem.
We restore the indices $\pm$ in the singularity scattering data, namely we denote $(\gpoi,\Kpoi,\phipoi_0,\phipoi_1)=\Sbf(\gmoi,\Kmoi,\phimoi_0,\phimoi_1)$

First consider $\Siso$.
It yields $\Kpoi=\frac{1}{3}\delta>0$ so it is quiescence-preserving.
It is manifestly not invertible (hence not idempotent).
For this class of scattering maps, shift-covariance states that the metric is unchanged upon shifting~$\phimoi_1$, which means that the functions $\alpha_0,\alpha_1,\alpha_2$ describing the change of metric are only functions of $\theta_-$ and~$\phimoi_0$.

We thus concentrate henceforth on~$\Sani$.
This anisotropic scattering map preserves quiescence provided it maps $\Kmoi>0$ to $\Kpoi>0$.
Given a Kasner angle~$\theta_-$, the corresponding extrinsic curvature $\Kmoi$ is positive if and only if its eigenvalues $1/3+(2/3)r_-\cos(\theta_-+2\pi j/3)$, $j=0,1,2$, are positive.  This holds for small enough~$r_-$, specifically $r_-<1/(2\cos\widehat\theta_-)$ with $\widehat\theta_-$ defined in~\eqref{phimin}.  Equivalently, the condition for $\Kmoi>0$ is $|\phimoi_0|>\phimin(\theta_-)$, and likewise $\Kpoi>0$ is equivalent to $|\phipoi_0|>\phimin(\theta_+)$.
Since $\theta_-=\theta_+$ for $\epsilon=+1$ and $\theta_++\pi$ for $\epsilon=-1$ we get the condition in the theorem.

We turn to invertibility or idempotence.
On general grounds, the inverse of an invertible scattering map $\Sani_{\Phi_1,c_1,\epsilon_1}$ is an invertible scattering map, which must be of the form $\Sani_{\Phi_2,c_2,\epsilon_2}$ because $\Siso_{\alpha_0,\alpha_1,\alpha_2,\varphi,\epsilon}$ is never invertible.
Let us check that (even without assuming invertibility)
\bel{Sani-compose}
\Sani_{\Phi_2,c_2,\epsilon_2} \circ \Sani_{\Phi_1,c_1,\epsilon_1} = \Sani_{\Phi_3,c_3,\epsilon_3}
\ee
with $c_3=c_2c_1$, $\epsilon_3=\epsilon_2\epsilon_1$, and $\Phi_3(\theta,\phi_0,\phi_1) = \Phi_2(\theta+\pi\delta_{\epsilon_1=-1},\Phi_1(\theta,\phi_0,\phi_1))$.
On general grounds, composing two ultralocal scattering maps gives an ultralocal scattering map, and we have the full classification available, so we simply need to fix parameters (the composition can manifestly not be of the form $\Siso$).
The sign $\epsilon_3$ is fixed by comparing Kasner angles: the phase $e^{i\theta}$ is multiplied by $\epsilon_1$ and then by $\epsilon_2$ upon applying the two scattering maps.
The scalar fields $(\phi_0,\phi_1)$ then manifestly transform according to the composition $\Phi_2(\theta+\pi\delta_{\epsilon_1=-1},\,\cdot\,,\,\cdot\,)\circ\Phi_1(\theta,\,\cdot\,,\,\cdot\,)$.
From volume factors of the metric we find $c_2^2c_1^2=c_3^2$.
This establishes~\eqref{Sani-compose}.
Imposing that $\Phi_3=\Id$, $c_3=1$, $\epsilon_3=1$ gives the characterization of invertible maps in the theorem, and their inverse.
Imposing further that $\Phi_2=\Phi_1$, $c_2=c_1$, $\epsilon_2=\epsilon_1$ gives the characterization of idempotent maps.

Next, we determine which $\Sani_{\Phi,c,\epsilon}$ are shift-covariant.
Shift-covariance sets $\del_{\phimoi_1}\phipoi_0=0$ and $\del_{\phimoi_1}\phipoi_1=a$ for some constant $a\in\RR$.
Additionally, $\Phi\colon(\thetamoi,\phimoi_0,\phimoi_1)\mapsto(\phipoi_0,\phipoi_1)$ must be an $\epsilon$-canonical transformation in the sense of \autoref{def:canonical-transfo}.
The condition of preserving volume in phase space is crucial: it gives
\[
\del_{\phimoi_0} \biggl( a \frac{\phipoi_0}{r(\phipoi_0)} - \epsilon \frac{\phimoi_0}{r(\phimoi_0)}\biggr) = 0,
\]
which is only possible provided $a\neq 0$.  We get $\phipoi_0/r(\phipoi_0) = \epsilon a^{-1} \phimoi_0/r(\phimoi_0) + b(\theta)$ where the integration constant~$b$ cannot depend on~$\phimoi_1$ because $\phipoi_0$~does not.
Boundary regularity requires $a\phipoi_0/r(\phipoi_0) - \epsilon\phimoi_0/r(\phimoi_0) = a\,b(\theta)$ to tend to zero as $r\to 0$.
This fixes $b=0$, hence gives the main characterization of shift-covariant maps in the theorem: $\phipoi_0/r(\phipoi_0) = \epsilon a^{-1} \phimoi_0/r(\phimoi_0)$.
The condition $\del_{\phimoi_1}\phipoi_1=a$ states that $a^{-1}\phipoi_1-\phimoi_1$ is a function $f(\theta,\phimoi_0)$.
We conclude that {\sl shift-covariant} $\Sani$~are characterized by $f,a,c,\epsilon$ and given explicitly by
\bse\label{Sani-shift-cov}
\be
\gathered
\phipoi_0 = \epsilon a^{-1} \mu^{-3} \phimoi_0 , \qquad
\phipoi_1 = a f(\theta_-,\phimoi_0) + a \phimoi_1 , \qquad
\Kpoicirc = \epsilon \mu^{-3} \Kmoicirc , \\
\gpoi = c^2 \mu^2 \exp\biggl(16\pi\epsilon\xi\cos\Theta_- - 16\pi\del_{\theta_-} \biggl(\epsilon\xi+ \frac{\phimoi_0}{r(\phimoi_0)}f\biggr)\sin\Theta_-\biggr) \gmoi
\endgathered
\ee
where we restored the $\pm$ indices that denote both sides of the singularity and where
\be
\mu=\mu(\phimoi_0) = (1 + 12\pi\phimoi_0^2(a^{-2}-1))^{1/6}, \qquad
\xi=\xi(\theta_-,\phimoi_0) = -\epsilon \int_{-1/\sqrt{12\pi}}^{\phimoi_0} \del_y f(\theta_-,y) \frac{y\,dy}{r(y)} .
\ee
It remains to translate the conditions on~$\Phi$ in \autoref{def:canonical-transfo} in terms of~$f$.
We find that~$f$ must be $2\pi/3$ (spacelike case) or $2\pi$ (timelike case) periodic and even in~$\theta$,
that $\del_\theta f(\theta,\phi_0)=o(r(\phi_0))$ as $\phi_0\to \pm 1/\sqrt{12\pi}$, and that
\bel{Sani-f-tech}
\xi\biggl(\theta_-,\frac{\pm 1}{\sqrt{12\pi}}\biggr) = 0 , \quad\text{namely}\quad
\int_{-1/\sqrt{12\pi}}^{1/\sqrt{12\pi}} \del_y f(\theta,y)\,\frac{y\,dy}{r(y)} = 0 .
\ee
\ese
It is easy to check from~\eqref{Sani-compose} that all shift-covariant~$\Sani$ are invertible, with inverse obtained by replacing $c\to 1/c$, $a\to 1/a$ and changing $f$ to the map
\[
(\theta,\phi_0)\mapsto -af\bigl(\theta+\pi\delta_{\epsilon=-1},\epsilon a \phi_0 \bigl(1+12\pi\phi_0^2(a^2-1)\bigr)^{-1/2} \bigr) .
\]
Idempotence then requires $c=1$ (recall $c>0$), $a=\pm 1$, and $f(\theta,\phi_0)=-af(\theta+\pi\delta_{\epsilon=-1},\epsilon a \phi_0)$.
If $\epsilon=a=+1$ this condition is that $f=0$.
If $\epsilon=+1$ and $a=-1$ this condition is that $f$ be an even function of~$\phi_0$.
If $\epsilon=-1$ and $a=+1$ this condition is that $f(\theta,\phi_0)=-f(\theta+\pi,-\phi_0)$.
If $\epsilon=-1$ and $a=-1$ this condition is that $f$ be $\pi$-periodic in~$\theta$.

We determine now under which condition on $f,a,c,\epsilon$ the shift-covariant maps~$\Sani$ are quiescence-preserving.
Since $\phi_0\mapsto\phi_0/r(\phi_0)$ is monotonic, the condition $\abs{\phi_0}>\phimin(\theta)$ translates to
\[
\biggl|\frac{\phi_0}{r(\phi_0)}\biggr|
> \biggl|\frac{\phimin(\theta)}{r(\phimin(\theta))}\biggr|
= \biggl(\frac{4\cos^2\widehat{\theta}-1}{12\pi}\biggr)^{1/2} ,
\qquad \text{where} \quad
\widehat{\theta} = \theta - \frac{2\pi}{3} \biggl\lfloor\frac{3\theta}{2\pi}-1\biggr\rfloor \in \biggl[\frac{2\pi}{3},\frac{4\pi}{3}\biggr).
\]
For $\epsilon=+1$, we want $\abs{\phipoi_0/r(\phipoi_0)}$ to be greater than this lower bound whenever $\abs{\phimoi_0/r(\phimoi_0)}$ is.  Since $\abs{\phipoi_0/r(\phipoi_0)}=\abs{a}^{-1}\abs{\phimoi_0/r(\phimoi_0)}$, the condition holds if and only if $\abs{a}\leq 1$.
For $\epsilon=-1$ the relevant angles~$\theta$ differ by~$\pi$ so the lower bounds are different.
Let us write the condition for $\theta_-=0$, $\theta_+=\pi$: then $\widehat{\theta}_-=2\pi/3$ so $\cos\widehat{\theta}_-=-1/2$ and we want the following implication
\[
\biggl|\frac{\phimoi_0}{r(\phimoi_0)}\biggr| > 0
\implies \bigg|\frac{\phipoi_0}{r(\phipoi_0)}\biggr| > \frac{1}{\sqrt{4\pi}} .
\]
There is no~$a$ that would ensure this, because the premise is obeyed by arbitrarily small $\phimoi_0/r(\phimoi_0)$, which lead to arbitrarily small $\phipoi_0/r(\phipoi_0)=\epsilon a^{-1}\phimoi_0/r(\phimoi_0)$.

Finally, we consider momentum-preserving and momentum-reversing maps, for which all of $\abs{\phi_0}$, $\abs{k_1-1/3}$, $\abs{k_2-1/3}$, $\abs{k_3-1/3}$ are continuous through the bounce.
In fact it is enough to require any one of them to be continuous: first, $\Siso$~is immediately ruled out, then, we observe for $\Sani$ that each $\abs{k_a-1/3}$ and $r(\phi_0)$ scales by the same factor $\rpoi/\rmoi$ upon crossing the singularity, so if any of them is continuous all of them are.
We can thus write $\phipoi_0=\epsilon a\phimoi_0$ for some sign $a=\pm 1$.
This sign is constant for $\phimoi_0>0$ and constant for $\phimoi_0<0$ in order for $\phipoi_0$~to remain continuous.
For $\phimoi_0\neq 0$, volume-preservation then imposes $\del_{\phimoi_1}\phipoi_1=a$, thus $\phipoi_1=a(f(\theta,\phimoi_0)+\phimoi_1)$.  Continuity of $\phipoi_1$ as we dial $\phimoi_0$ from positive to negative then forces the sign~$a$ to be the same for $\phimoi_0\gtrless 0$.
Altogether, the canonical transformation~$\Phi$ coincides with the particular case $a=\pm 1$ of what we found for shift-covariant maps:
  \bel{scattering-map-antipodal-K}
  \Kpoicirc = \epsilon \Kmoicirc , \qquad
  (\phipoi_0,\phipoi_1) = \pm \bigl(\epsilon \phimoi_0, \phimoi_1 + f(\theta,\phimoi_0)\bigr) ,
  \ee
  where $f$ is subject to the technical condition $\int_{-1/\sqrt{12\pi}}^{1/\sqrt{12\pi}} (yf'(y)dy/r(y))=0$,
  and $\gpoi$ is given as in~\eqref{Sani-gpoi}.

\subsection{Conformal scattering maps} 
\label{subsec-example-comparison}

We discuss here an interesting class of scattering maps in which the asymptotic profiles before and after the singularity are related by a conformal transformation.
In \autoref{secti-23} we introduced the notion of rigidly conformal scattering maps, which are such that $\gpoi$ and $\gmoi$ are in the same conformal class.  Note, however, that $g_{\pm}$~are simply convenient quantities to parametrize the asymptotic profiles $|s|^{2K_{\pm}}g_{\pm}$ for $s\gtrless 0$ by their values at $s=\pm 1$.  It is quite natural, thus, to compare the asymptotic profiles at other values of~$s$: this yields the more flexible notion of ``conformal scattering maps''.  For maps that are additionally ultralocal, as we will see, the whole asymptotic profiles are then conformally related in a suitable sense, reminiscent of the conformal cyclic cosmology proposal of Penrose, Tod, L\"ubbe, and others~\cite{LubbeTod,Lubbe,Tod:2002wd}.  An important difference is that we are considering here junctions of a Big Crunch with a Big Bang, while the conformal cyclic cosmology proposal maps (spacelike) future null infinity of an approximately de Sitter spacetime to the Big Bang of a new aeon.

\begin{definition}\label{def-conformal}
  A {\bf conformal scattering map} on a $3$-manifold~$\Hcal$ is a scattering map~$\Sbf$ such that for any data $(\gmoi,\Kmoi,\phimoi_0,\phimoi_1)$ and its image $(\gpoi,\Kpoi,\phipoi_0,\phipoi_1)$ under~$\Sbf$,
 for each point $x\in\Hcal$ there exists $s_-(x)<0<s_+(x)$ such that the metrics $|s_-|^{2\Kmoi}\gmoi$ and $|s_+|^{2\Kpoi}\gpoi$ are in the same conformal class.
\end{definition}

Let $\Sbf$ be a conformal scattering map that is ultralocal.
By \autoref{cor:Kg-scaling}, ultralocality implies $\Kpoi=\tfrac{1}{3}(1-\zeta)\delta+\zeta\Kmoi$ for $\zeta=\gamma\sqrt{|\gmoi|/|\gpoi|}$.
Conformality states that $|s_+|^{2\Kpoi}\gpoi=\Omega^2|s_-|^{2\Kmoi}\gmoi$ for some space-dependent $s_-<0<s_+$ and scalar factor~$\Omega$.  For any point $x\in\Hcal$ and $s>0$ we can thus rewrite the asymptotic profile after the singularity as
\bel{calc-profile-conf}
\aligned
|s|^{2\Kpoi}\gpoi & = \Bigl|\frac{s}{s_+(x)}\Bigr|^{2\Kpoi} |s_+(x)|^{2\Kpoi}\gpoi \\
& = \Bigl|\frac{s}{s_+(x)}\Bigr|^{2(1-\zeta)/3+2\zeta\Kmoi} \Omega^2 |s_-(x)|^{2\Kmoi}\gmoi 
= \biggl( \Omega \Bigl|\frac{s}{s_+(x)}\Bigr|^{(1-\zeta)/3}\biggr)^2 \biggl|\Bigl|\frac{s}{s_+(x)}\Bigr|^\zeta s_-(x)\biggr|^{2\Kmoi}\gmoi .
\endaligned
\ee
This means that the asymptotic spatial metric at $s>0$ is in the same conformal class as the one at $|s/s_+(x)|^\zeta s_-(x)<0$.  We change slightly the notation $s,s_\pm$ used in this derivation to state the following proposition.

\begin{proposition}[Conformal asymptotic profiles]\label{prop:conf-alls}
  Let $\Sbf$ be an ultralocal conformal scattering map on a $3$-manifold~$\Hcal$, and let $(g_{\pm},K_{\pm},\phi_{0\pm},\phi_{1\pm})$ be singularity data and its image under~$\Sbf$.  For any $x\in\Hcal$ and any $s_+>0$ there exists $s_-(x,s_+)<0$ such that $|s_{\pm}|^{2K_{\pm}}g_{\pm}$ are in the same conformal class.
\end{proposition}

The analogous statement exchanging $s_+$ and~$s_-$ only holds if we exclude the isotropic scattering map~$\Siso_{\alpha_0,\alpha_1,\alpha_2,\varphi,\epsilon}$.
Indeed, $\zeta=0$ in~\eqref{calc-profile-conf} in this case, so the asymptotic metric for any $s>0$ is conformal to the asymptotic metric at the same fixed proper time $s_-(x)<0$ before the singularity, while the asymptotic metrics for other $s<0$ are in a different conformal class.
The other family $\Sani_{\Phi,c,\epsilon}$ of ultralocal scattering maps, on the other hand, has $\zeta\neq 0$ so that $s_+\mapsto s_-(x,s_+)$ is a bijection and can be inverted to $s_-\mapsto s_+(x,s_-)$ for each~$x$.

To finish off our discussion of conformal ultralocal scattering maps we write down their classification, which is a straightforward consequence of \autoref{theorem-ultralocal}.
By \autoref{prop:conf-alls} it is enough to check whether the asymptotic metric at $s_+=1$ is conformal to some metric with $s_-<0$.
Thus, among $\Siso_{\alpha_0,\alpha_1,\alpha_2,\varphi,\epsilon}$ and $\Sani_{\Phi,c,\epsilon}$ we seek maps such that $\gpoi=\Omega^2|s_-|^{2\Kmoi}\gmoi$ for some scalars $\Omega,s_-$, or equivalently, $\exp(a\delta+b\cos\Theta_-)\gmoi$ for some scalars $a,b$ (we recall $\Kmoi=\tfrac{1}{3}\delta+\tfrac{2}{3}\rmoi\cos\Theta_-$).
For $\Siso_{\alpha_0,\alpha_1,\alpha_2,\varphi,\epsilon}$ we immediately find the condition to be that the coefficient $\alpha_2$ of $\cos(2\Theta_-)$ must vanish.
For $\Sani_{\Phi,c,\epsilon}$ the coefficient of $\sin\Theta_-$ must vanish.  Since it is known to vanish at the boundary $\rmoi=0$ (by isotropy), we simply write that its $\phimoi_0$ derivative vanishes.  To make the proposition self-contained we replace the auxiliary function $\xi$ in~\eqref{Sani-gpoi} using~\eqref{Sani-dkappa}.  After expanding derivatives and cancelling some terms the relation we find is surprisingly simple.

\begin{proposition}[Conformal ultralocal scattering maps for self-gravitating scalar fields]
  \label{conformal-ultralocal}
  Spacelike or timelike, conformal ultralocal scattering maps~$\Sbf$ are $\Siso_{\alpha_0,\alpha_1,\alpha_2,\varphi,\epsilon}$ with $\alpha_2=0$, and $\Sani_{\Phi,c,\epsilon}$ with
  \[
  \del_\theta \Phi_0(\theta,\phi_0,\phi_1) \del_{\phi_0} \Phi_1(\theta,\phi_0,\phi_1)
  = \del_{\phi_0} \Phi_0(\theta,\phi_0,\phi_1) \del_\theta \Phi_1(\theta,\phi_0,\phi_1) .
  \]
\end{proposition}
  

\section{Proof of the classification of ultralocal scattering maps}
\label{section---5} 

\subsection{A zoo of singularity data sets}
\label{ssec:zoo}

\paragraph{Prescribing values at points.}

The classification of ultralocal scattering maps requires several technical lemmas on the existence of singularity data sets, which essentially state that the momentum constraint is not so constraining after all.
We present these rather technical results in the present section and in \autoref{ssec:deriv}, which could be skipped on first reading.
We then move on to the classification proper: explaining the relevant tensor structures in \autoref{ssec:struct}, then showing in \autoref{ssec:scaling} that $\Kpoicirc$ is a multiple of~$\Kmoicirc$, and finally completing the classification in \autoref{ssec:completion}.

To state the lemmas we use the Kasner angle~$\theta$ introduced in~\eqref{Kasner-rtheta}, which is defined modulo $2\pi/3$ (spacelike case) or $2\pi$ (timelike case) or is completely ill-defined when the Kasner radius $r=((3/2)\Tr\Kcirc^2)^{1/2}=(1-12\pi\phi_0^2)^{1/2}$ vanishes, namely when $\phi_0=\pm 1/\sqrt{12\pi}$.
As shown by \autoref{lem:scalars}, any scalar quantity constructed from the singularity data without derivatives must be a function of $(\theta,\phi_0,\phi_1)$.
This triplet of scalars ranges over $\RR\times I_0\times\RR$, where $I_0=[-1/\sqrt{12\pi},1/\sqrt{12\pi}]$, modulo the ambiguity in~$\theta$.
When we say that a singularity data set is such that $(\theta,\phi_0,\phi_1)$ {\sl assumes some prescribed value} at a point~$x$, we mean this modulo the ambiguities in~$\theta$.
Note that prescribing the value of $(\theta,\phi_0,\phi_1)$ is equivalent to prescribing $\phi_0,\phi_1$ and eigenvalues $k_1,k_2,k_3$ of~$K$ compatible with the constraints (and with the convention that timelike eigenvectors have eigenvalue~$k_1$), so we often work directly with prescribed $(k_1,k_2,k_3,\phi_0,\phi_1)$, at the cost of having to impose the constraints explicitly rather than through the Kasner radius/angle parametrization.

\begin{lemma}[Singularity data sets with prescribed values at points]
  \label{lem:prescribed-values}
  Let $\Hcal$ be a $3$-manifold, let $x^{(1)},\dots,x^{(n)}\in\Hcal$ be distinct points, and let $(\theta^{(i)},\phi_0^{(i)},\phi_1^{(i)})\in\RR\times I_0\times\RR$ for any $i=1,\dots,n$.
Then there exists a smooth singularity data set $(g,K,\phi_0,\phi_1)$ such that $(\theta,\phi_0,\phi_1)$ assumes the prescribed value $(\theta^{(i)},\phi_0^{(i)},\phi_1^{(i)})$ at each~$x^{(i)}$, $i=1,\dots,n$.
\end{lemma}

\begin{proof} Consider non-intersecting neighborhoods $B^{(i)}\ni x^{(i)}$ each diffeomorphic to the unit ball in~$\RR^3$.  We construct below a data set in each ball that connects in a~$C^\infty$ manner with the following trivial data set: $K=\tfrac{1}{3}\delta$, $\phi_0=1/\sqrt{12\pi}$ and $\phi_1=0$ outside $\bigcup_i B^{(i)}$, with the metric being an arbitrary smooth metric. Here we chose the sign of~$\phi_0$ arbitrarily.

  It is now enough to construct a data set on the unit ball $B\subset\RR^3$ such that
\bei

\item $k_1,k_2,k_3,\phi_0,\phi_1$ takes a prescribed value $(k_1^{(i)},k_2^{(i)},k_3^{(i)},\phi_0^{(i)},\phi_1^{(i)})$ at $0\in B$, and
  
\item $K=\tfrac{1}{3}\delta$, $\phi_0=1/\sqrt{12\pi}$, and $\phi_1=0$ uniformly in a neighborhood of the boundary of~$B$.
\eei
We choose the metric to be conformally flat, specifically a scalar multiple of the standard (Euclidean or Minkowski) metric on~$\RR^3$, and we rescale the trace-free extrinsic curvature accordingly:
\bel{gK-rescaled}
g = \Omega^{-2/3} \diag(\pm 1,1,1) , \qquad \Kcirc = \Omega \Ht,
\ee
where we choose $\Omega$ to be a radial function, namely $\Omega=\Omega(|x|^2)$.
  As we saw in one instance in \autoref{subsec:param-data}, this scaling is convenient since the momentum constraint remains simple:
\bel{momentum-conformal-rescaled}
\nabla_a K^a_b = 8\pi\phi_0\del_b\phi_1 \iff
\del_a\Ht^a_b = \frac{8\pi\phi_0}{\Omega} \del_b\phi_1.
\ee
 The Hamiltonian constraint $\Tr K^2=1-8\pi\phi_0^2\leq 1$, on the other hand, imposes an upper bound on~$\Omega$, with equality if $\phi_0=0$:
\[
 \Omega \leq \bigl(\tfrac{3}{2}\Tr(\Ht^2)\bigr)^{-1/2}.
\]

\begin{figure}
  \centering
\begin{tikzpicture}
    \node at (0,0) {$x^{(1)}$}; \node at (-1.3,0) {$B^{(1)}$};
    \draw[densely dashed] (0,0) circle (.4 and .3) circle (.6 and .45) circle (.8 and .6);
    \draw (0,0) circle (1 and .75);
    \node at (1.2,1.8) {$x^{(2)}$}; \node at (-.1,1.8) {$B^{(2)}$};
    \draw[densely dashed] (1.2,1.8) circle (.4 and .5) circle (.6 and .75) circle (.8 and 1);
    \draw (1.2,1.8) circle (1 and 1.25);
    \node at (7,2.6) {$x^{(3)}$};
    \draw[densely dashed] (6,2.6) arc[start angle=-180, end angle=-60, radius=1] node [above] {\scriptsize $\phi_1=\text{constant}$} node [pos=.7,below] {\scriptsize $\phi_0=0,\phi_1$ varies};
    \draw[densely dashed] (5.4,2.2) arc[start angle=-175, end angle=-60, radius=1.6] node [pos=.7,below] {\scriptsize $\phi_0>0,\phi_1=0$};
    \draw[densely dashed] (4.8,1.8) arc[start angle=-170, end angle=-60, radius=2.4] node [pos=.7,below] {\scriptsize $\Omega=\phi_0$};
    \draw (4.2,1.4) arc[start angle=-165, end angle=-60, radius=3.2] node [pos=.1,left] {$B^{(3)}$} node [pos=.4,below left] {\scriptsize trivial};
  \end{tikzpicture}
  \caption{\label{fig:layers}{\bf Singularity data set constructed in \autoref{lem:prescribed-values}.}
    It has prescribed values at~$x^{(i)}$ and is trivial away from ball-shaped neighborhoods.
    {\bf Left:} global structure.  {\bf Right:} some properties of the concentric layers around one~$x^{(i)}$.}
\end{figure}
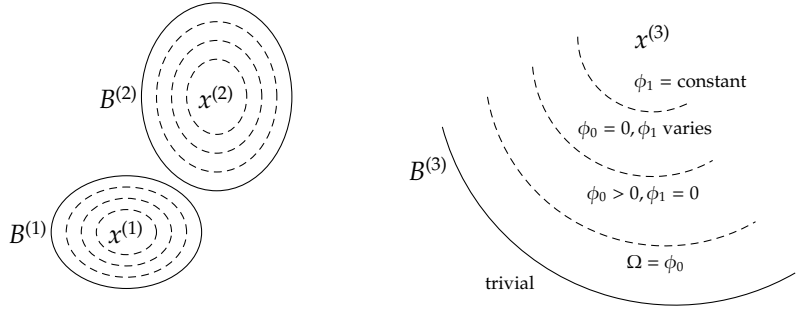

We construct the data in four concentric layers. To avoid any issue with regularity at the junction between layers (or at the center $0\in B$) we simply arrange for the singularity data to be constant in a neighborhood of each such junction (or of~$0$).
The whole construction is summarized in \autoref{fig:layers}.
For the three inner layers, in $0\leq|x|^2\leq 3/4$ (say), we choose $\Ht^a_b=\Kucirc^a_b$ to be a constant diagonal matrix
\bel{Ht-constant}
\Kucirc \coloneqq \diag\bigl(k_1^{(i)},k_2^{(i)},k_3^{(i)}\bigr) - \tfrac{1}{3} \delta .
\ee
This reproduces the prescribed data at~$0$ provided $\Omega(0)=1$. Since $\del_a\Kucirc^a_b=0$, the momentum constraint states $\phi_0\del_b\phi_1=0$, which we satisfy in three successive layers by imposing $\del_b\phi=0$, $\phi_0=0$, and $\del_b\phi=0$.
The purpose of these layers is two-fold: to allow us to tune the value of $\phi_1$ to zero, and to ensure a particular value $\phi_0>0$ at the boundary in order to connect to the last layer.
\bei

\item First, for $0\leq|x|^2\leq 1/4$, we keep $\phi_1=\phi_1^{(i)}$ constant and vary $\Omega$ smoothly from $\Omega=1$ for small~$|x|^2$ to $\Omega=(\tfrac{3}{2}\Tr(\Kucirc^2))^{-1/2}$ for $|x|^2$ close to $1/4$.  In this layer, the Hamiltonian constraint sets
\[
\phi_0 = \pm \sqrt{(2/3-\Omega^2\Tr\Kucirc^2)/(8\pi)}, 
\]
where the sign is that of~$\phi_0^{(i)}$.  In particular, for $|x|^2$ close to $1/4$ we have $\phi_0=0$.

\item Second, for $1/4\leq|x|^2\leq 1/2$, we keep $\phi_0=0$ and $\Omega$~constant, while varying $\phi_1$~smoothly until $\phi_1=0$ for $|x|^2$ near~$1/2$.
  
\item Third, for $1/2\leq|x|^2\leq 3/4$, we keep $\phi_1=0$ and vary~$\Omega$ like in the first layer, ensuring that $\phi_0>0$.  To simplify the construction of the next layer we vary $\Omega$ and $\phi_0$ until they become equal, which occurs for $\phi_0$~given in~\eqref{phi0-eq-Omega} below.  We choose $\Omega=\phi_0$ equal to this value for $|x|^2$ close to~$3/4$, where we recall that $\Ht=\Kucirc$ is still constant and given by~\eqref{Ht-constant}.

\eei

For the last layer, located in the interval $|x|^2\in[3/4,1]$, we take $\phi_0>0$ throughout, which allows us to choose $\Omega=\phi_0$ in~\eqref{gK-rescaled}.  The Hamiltonian constraint (together with $\phi_0>0$) can be solved:
\bel{phi0-eq-Omega}
\Omega = \phi_0 \quad \text{ at }  \phi_0 = \sqrt{\frac{2/3}{8\pi + \Tr\Ht^2}}.
\ee
 The momentum constraint~\eqref{momentum-conformal-rescaled} simplifies to $\del_a\Ht^a_b=8\pi\del_b\phi_1$ and we consider the following class of explicit solutions:
\[
\Ht^a_b = \alpha(|x|^2) \Kucirc^a_b , \qquad \phi_1 = \frac{1}{8\pi}\alpha'(|x|^2) x_a \Kucirc^a_b x^b,
\]
where $\Kucirc^a_b$ is~\eqref{Ht-constant}, $\alpha=\alpha(|x|^2)$ is a radial function and $\alpha'$ its $|x|^2$~derivative.
In order for the layer to properly join with the previous one and with the trivial data outside the ball, we choose $\alpha=1$ for $|x|^2$ close to~$3/4$, and $\alpha=0$ for $|x|^2$ close to~$1$.

This concludes the construction of the singularity data set on $B\subset\RR^3$ that is trivial near the boundary and has a prescribed value at~$0$.  Patching such balls around each point~$x^{(i)}$ into a trivial singularity data set on an arbitrary $3$-manifold $\Hcal$ is then immediate.
\end{proof}

\autoref{lem:prescribed-values} has a straightforward consequence, stated now.

\begin{lemma}[Always-constant ultralocal scalars can only depend on the signature]
\label{lem:constants}
Let $\Hcal$ be a $3$-manifold and $A$ be an ultralocal scalar function of singularity data on~$\Hcal$. If $A(x)$ is independent of $x\in\Hcal$ for any smooth singularity data $(g,K,\phi_0,\phi_1)\in\Ibf(\Hcal)$, then $A$~is an overall constant independent of the data itself, except for the signature of~$g$.
\end{lemma}

\begin{proof} Fix a signature (spacelike or timelike) once and for all. We wish to show that $A$ is constant for data with this signature.
By \autoref{lem:scalars} we know that the ultralocal scalar~$A$ can be written as
$
A(x)=\widehat{A}(\theta(x),\phi_0(x),\phi_1(x))
$
for some function $\widehat{A}\colon\RR\times I_0\times\RR\to\RR$ that is even and ($2\pi/3$ or $2\pi$) periodic in~$\theta$. Our goal is to show that $\widehat{A}(\theta^{(1)},\phi_0^{(1)},\phi_1^{(1)})=\widehat{A}(\theta^{(2)},\phi_0^{(2)},\phi_1^{(2)})$ for any pair of values in $\RR\times I_0\times\RR$. \autoref{lem:prescribed-values} provides a singularity data set assuming the values $(\theta^{(i)},\phi_0^{(i)},\phi_1^{(i)})$, $i=1,2$ at two different points.
Since $A(x)$ is $x$-independent, it takes the same value at these points hence $\widehat{A}$ takes the same values for the two given $(\theta,\phi_0,\phi_1)$.  Altogether, $\widehat{A}$ is constant so $A$ only depends on the signature of~$g$ (and on the scattering map of course).
\end{proof}

\paragraph{Restriction to non-degenerate data.}

We now prove that ultralocal scalars and ultralocal scattering maps are characterized by their value on data for which $\phi_0\neq 0$ and the eigenvalues $k_1,k_2,k_3$ are distinct.  Some care is needed when stating the result, because our constructions of scattering data sets (such as the one used for \autoref{lem:prescribed-values}) involve regions where the extrinsic curvature~$K$ is in fact degenerate.

\begin{lemma}[Non-degenerate data distinguish ultralocal scalars]
\label{lem:continuity}
Let $\Hcal$ be a $3$-manifold and let $\sigma$ be a continuous ultralocal scalar.
Assume that for any scattering data set $(g,K,\phi_0,\phi_1)$ on~$\Hcal$, and any $x\in\Hcal$ such that $\phi_0(x)\neq 0$ and $K(x)$ has three distinct eigenvalues, one has $\sigma(g,K,\phi_0,\phi_1)(x) = 0$.
Then $\sigma=0$ identically.
\end{lemma}

\begin{lemma}[Non-degenerate data distinguish ultralocal scattering maps]
\label{lem:non-degenerate}
Let $\Hcal$ be a $3$-manifold and let $\Sbf_1,\Sbf_2$ be two ultralocal scattering maps on~$\Hcal$.
Assume that for any scattering data set $(g,K,\phi_0,\phi_1)$ on~$\Hcal$, and any $x\in\Hcal$ such that $\phi_0(x)\neq 0$ and $K(x)$ has three distinct eigenvalues, one has $\Sbf_1(g,K,\phi_0,\phi_1)(x) = \Sbf_2(g,K,\phi_0,\phi_1)(x)$.
Then $\Sbf_1=\Sbf_2$.
\end{lemma}

\begin{proof}[Proof of Lemmas \ref{lem:continuity} and \ref{lem:non-degenerate}]
We establish the two lemmas simultaneously. Assume that we are given a scalar~$\sigma$ and two scattering maps $\Sbf_1,\Sbf_2$ satisfying the conditions in the two lemmas, respectively.

We construct singularity data sets $(g,K,\phi_0,\phi_1)$ on~$\Hcal$ taking any prescribed value at some $x\in\Hcal$, and show that $\sigma(g,K,\phi_0,\phi_1)(x)=0$ and that $\Sbf_1(g,K,\phi_0,\phi_1)(x)=\Sbf_2(g,K,\phi_0,\phi_1)(x)$ for these specific data sets.  Ultralocality extends this equality to any other data set taking the same value at $x\in\Hcal$, and diffeomorphism-invariance shows the choice of point does not matter.  Thus $\sigma=0$ and $\Sbf_1=\Sbf_2$ at all points for arbitrary data sets.

The key to show $\sigma(g,K,\phi_0,\phi_1)(x)=0$ for the data sets constructed below is to use continuity of $\sigma(g,K,\phi_0,\phi_1)$.  Since we know that it vanishes in the set of points $x\in\Hcal$ such that $\phi_0(x)\neq 0$ and $K(x)$ has three distinct eigenvalues, it must vanish as well in the closure of that set inside~$\Hcal$.  We simply need to ensure that the point of interest is in this closure.  The same reasoning applies to $\Sbf_1(g,K,\phi_0,\phi_1)(x)=\Sbf_2(g,K,\phi_0,\phi_1)(x)$ because, by definition, scattering maps send smooth data to (at least) continuous data.

Rather than constructing a complicated singularity data set that covers all cases at the same time, we first show that $\sigma$ vanishes and $\Sbf_1,\Sbf_2$ agree for all data with $\phi_0\neq 0$ and $K\neq\tfrac{1}{3}\delta$; in other words we treat the case where two eigenvalues of~$K$ coincide.
 We follow the construction used in the proof of \autoref{lem:prescribed-values}, building data on the unit ball $B\subset\RR^3$ and embedding it inside trivial data for the rest of~$\Hcal$.
 We choose $g,K$ as in~\eqref{gK-rescaled} with $\Omega=\abs{\phi_0}$ (where $\phi_0\neq 0$ will be specified later) so that the momentum constraint simplifies:
\bse
\label{data-lem-non-degenerate}
\be
g = \abs{\phi_0}^{-2/3} \diag(\pm 1,1,1) , \qquad \Kcirc = \abs{\phi_0} \Ht, \qquad
\del_a\Ht^a_b = 8\pi(\sgn\phi_0) \del_b\phi_1.
\ee
We consider the following class of explicit solutions:
\be
\aligned
\Ht^a_b 
& = \alpha(|x|^2) \,\Kucirc^a_b + \beta(|x|^2) \,\bigl(x^a x_b - \tfrac{1}{3} \delta^a_b |x|^2\bigr), 
\qquad
\quad 
  8\pi(\sgn\phi_0)\phi_1 = \alpha'(|x|^2) \,x_a \Kucirc^a_b x^b + \int^{|x|^2} \bigl(\tfrac{5}{3}\beta(q) + \tfrac{2}{3}\beta'(q)\bigr) \, dq,
\endaligned
\ee
\ese
  where $\alpha=\alpha(|x|^2)$ and $\beta=\beta(|x|^2)$ are radial functions and $\Kucirc\neq 0$ is a prescribed non-zero value.
  We choose a smooth function $\alpha=\beta$ whose derivatives (of all orders) vanish at $0$ and~$1$, with prescribed values $\alpha(0)=\beta(0)=1$ and $\alpha(1)=\beta(1)=0$.
  It is then easy to check (using $\Kucirc\neq 0$) that for $x$ approaching~$0$ in a generic direction the eigenvalues of~$\Ht$ are all distinct, hence those of~$K$ also are.
  The function $\sigma(g,K,\phi_0,\phi_1)$ thus vanishes at points $x$ approaching~$0$ hence at~$0$.  For the same reason, the continuous functions $\Sbf_1(g,K,\phi_0,\phi_1)$ and $\Sbf_2(g,K,\phi_0,\phi_1)$ agree at~$0$.
  The data at~$0$ that can be achieved using this construction is arbitrary except for the conditions $\phi_0\neq 0$ and $K\neq\tfrac{1}{3}\delta$.

Next, we show that $\sigma$ vanishes and $\Sbf_1,\Sbf_2$ agree for data with $\phi_0=0$. Consider the singularity data sets constructed in the proof of \autoref{lem:prescribed-values} and shift $\phi_1$ by some arbitrary overall constant. Since $\sigma(g,K,\phi_0,\phi_1)$ vanishes and $\Sbf_i(g,K,\phi_0,\phi_1)$, $i=1,2$, agree on the region $\{\phi_0\neq 0\}$, this must still be the case on its boundary.  It is easy to check that the data at such boundary points has $\phi_0=0$ of course, but no other restriction, namely it has arbitrary $\theta$ and~$\phi_1$.

The same singularity data sets also show that $\sigma$ vanishes and $\Sbf_1,\Sbf_2$ agree for $K=\tfrac{1}{3}\delta$: simply consider a point with $K=\tfrac{1}{3}\delta$ on the boundary of the region $\{K\neq\tfrac{1}{3}\delta\}$. This concludes the proof of Lemmas~\ref{lem:continuity} and~\ref{lem:non-degenerate}.
\end{proof}

\subsection{On derivatives of singularity data sets}
\label{ssec:deriv}

\paragraph{Scalars with vanishing derivatives.}

We continue our forays into constructing singularity data sets, but this time we additionally impose conditions on derivatives of the scalars $(\theta,\phi_0,\phi_1)$ at a point.  The saving grace is that we do not need to distinguish various special cases according to how many eigenvalues coincide: in applications later on, \autoref{lem:non-degenerate} allows us to assume $k_1,k_2,k_3$ are pairwise distinct.
This translates to two restrictions on the scalars: $\theta\neq 0\bmod{\pi/3}$ and $\phi_0\neq\pm 1/\sqrt{12\pi}$.
We denote
\bel{non-degenerate-scalars}
\Delta_\neq \coloneqq
(\RR\setminus\tfrac{\pi}{3}\ZZ)\times \bigl(\tfrac{-1}{\sqrt{12\pi}},\tfrac{1}{\sqrt{12\pi}}\bigr)\times\RR, 
\ee
so that $(\theta,\phi_0,\phi_1)\in\Delta_\neq$ means that the corresponding eigenvalues are pairwise distinct.
As before, when stating that a singularity data set $(g,K,\phi_0,\phi_1)$ assumes at some point~$x$ a certain value in~$\Delta_\neq$, the angle $\theta$ is understood up to $\theta\to-\theta$ and modulo $2\pi/3$ (spacelike case) or $2\pi$ (timelike case).

\begin{lemma}[Non-trivial data with locally constant scalars]
  \label{lem:derivatives}
  Let $\Hcal$ be a $3$-manifold and $x\in\Hcal$ be a point.
  For any prescribed value in~$\Delta_\neq$ (defined above) there exists a singularity data set $(g,K,\phi_0,\phi_1)$ on~$\Hcal$ such that, throughout a neighborhood of~$x$, $(\theta,\phi_0,\phi_1)$ assumes this prescribed value and $\nabla_a(\Kcirc^2)^a_b$ is nowhere vanishing.
\end{lemma}

\begin{proof}
  As in previous proofs we construct data with the desired properties on the unit ball $B\subset\RR^3$, such that the data smoothly reaches the trivial values $\phi_1=0$ and $K=\tfrac{1}{3}\delta$ at the boundary.  This yields a singularity data set on~$\Hcal$ by mapping the data through a diffeomorphism from $B$ to a neighborhood of $x\in\Hcal$ and extending it to $\Hcal$ using trivial data: constant $(K,\phi_0,\phi_1)$ and an arbitrary metric.  In fact, for the data set we construct on~$B$, $\nabla_a(\Kcirc^2)^a_b$ vanishes at~$0$ and is non-zero in a neighborhood of~$0$, so the diffeomorphism identifying $B$ with a neighborhood of $x\in\Hcal$ should be chosen to map $x$ close to $0\in B$ but not exactly at~$0$.

  We construct data on the ball $B\subset\RR^3$ in three layers and work in the standard basis of~$\RR^3$.
  Denote by $\phiu_0,\phiu_1$ and $\Ku=\diag(\ku_1,\ku_2,\ku_3)$ the prescribed data.
\bei
  \item The first layer, for $|x|^2\leq 1/2$, is described below.  It has constant $(K,\phi_0,\phi_1)=(\Ku,\phiu_0,\phiu_1)$, and has a variable metric that smoothly goes to $g=\diag(\pm 1,1,1)$ at $|x|^2=1/2$, with all derivatives vanishing.
  \item The second layer, for $1/2\leq|x|^2\leq 3/4$, is essentially the same as the third layer used in the proof of \autoref{lem:prescribed-values}.  It has $\phi_1 = \phiu_1$ and
\bel{gK-scaling-7}
  g = \Omega^{-2/3} \diag(\pm 1,1,1) , \qquad \Kcirc = \Omega \Kucirc ,
\ee
  where $\Omega$ interpolates from $1$ at $|x|^2=1/2$ to the value~\eqref{phi0-eq-Omega} at which $\Omega=\abs{\phi_0}$.
  \item The outermost layer, for $3/4\leq|x|^2\leq 1$ coincides with the outermost layer used in the proof of \autoref{lem:prescribed-values}, except for an overall sign of~$\phi_0$ and constant shift of~$\phi_1$.  It has~\eqref{gK-scaling-7} with $\Omega=\abs{\phi_0}$ and it interpolates from the previous layer to trivial data $\Kcirc=0$, $\phi_0=\pm 1/\sqrt{12\pi}$, and $\phi_1=\phiu_1$.
\eei
  In contrast to \autoref{lem:prescribed-values}, since we only want to prescribe data in one ball rather than multiple ones, there is no need to normalize the sign of~$\phi_0$ or the constant value of~$\phi_1$ in order to complete the data into data on~$\Hcal$.  One would otherwise need two additional layers for this purpose.

  We now construct the first layer, for $|x|^2\leq 1/2$.  We keep constant $K=\diag(k_1,k_2,k_3)$, $\phi_0$, and $\phi_1=0$, but we consider a diagonal metric with entries $\pm\exp(h_a)$, that is, 
$
  g = \diag(\pm e^{h_1},e^{h_2},e^{h_3}) ,
$
  where $h_a=h_a(x)$, $a=1,2,3$ are general functions.
  To work out the momentum constraint we compute
\bel{someq513}
  \nabla_a K^a_b = \del_a K^a_b + \Gamma^a{}_{ac} K^c_b - \Gamma^c{}_{ab} K^a_c
  = \del_a K^a_b + \del_c\bigl(\log|g|^{1/2}\bigr) K^c_b - \frac{1}{2} K^{ac} \del_b g_{ac}
\ee
  where we simply wrote the Christoffel symbols in terms of derivatives of the metric and used that $K^{ac}$ is symmetric to cancel two terms.
  For the data we are considering, the derivative term vanishes.
  Using that $g,K$ are diagonal we find that~\eqref{someq513} is equal to
\bel{momentum-diagonal}
  \sum_a \nabla_a K^a_b = \frac{1}{2} \sum_a (k_b - k_a) \del_b h_a , \qquad b=1,2,3,
\ee
  where we explicited the sum in $\nabla_a K^a_b$ to avoid confusion.
  Given that $\phi_1=0$, the momentum constraint states that this sum should vanish for all~$b$.  One rather symmetric solution is to choose
$
  h_a(x) = \lambda(|x|^2)\,(k_{a+1} - k_{a-1})$ 
  for $a=1,2,3$, 
  where indices of $k$ are understood modulo~$3$ and~$\lambda$ is some radial function that vanishes at $|x|^2=0$ and $|x|^2=1/2$ together with its derivatives of all orders (so as to keep the data smooth).

  We are free to impose that $\lambda'(|x|^2)\neq 0$ for all other values of~$|x|^2$: as we now show, this ensures that $\nabla_a(\Kcirc^2)^a_b\neq 0$ for $0<|x|^2<1/2$.
  Thanks to the momentum constraint $\nabla_a K^a_b=0$,
\[
  \nabla_a(\Kcirc^2)^a_b
  = (\nabla_a K^a_c) K^c_b + K^a_c \nabla_a K^c_b - \frac{2}{3} \nabla_a K^a_b + \frac{1}{9} \nabla_a \delta^a_b
  = K^a_c \nabla_a K^c_b.
\]
  Next, using the fact that $K$ is constant (in the given coordinates), and inserting the explicit form of Christoffel symbols we compute
\[
  \nabla_a(\Kcirc^2)^a_b
  = K^a_c \Gamma^c{}_{ad} K^d_b - K^a_c \Gamma^d{}_{ab} K^c_d 
  = \tfrac{1}{2} K^{ac} (g_{ca,d}+g_{cd,a}-g_{ad,c}) K^d_b
  - \tfrac{1}{2} (K^2)^{ad}(g_{da,b}+g_{db,a}-g_{ab,d}) .
\]
  In both terms the second and third derivatives of~$g$ cancel by symmetry of~$K^{ac}$ and of~$(K^2)^{ad}$.  We now write sums explicitly after using that $g$ and $K$ are diagonal:
\[
  \nabla_a(\Kcirc^2)^a_b
  = \frac{1}{2}K^d_b K^{ac} g_{ac,d} - \frac{1}{2}(K^2)^{ac} g_{ac,b}
  = \frac{1}{2}\sum_a (k_b-k_a) k_a \del_b h_a , \quad b=1,2,3 .
\]
  Using that $\del_bh_a=2x_b\lambda'(k_{a+1}-k_{a-1})$ with indices understood modulo~$3$, we compute by explicitly writing down the terms in the sum and factorizing to get
\[
  \nabla_a(\Kcirc^2)^a_b
  = x_b\lambda' (k_1-k_2) (k_2-k_3) (k_3-k_1) .
\]
  Since $x_b\lambda'$ is non-zero for $0<|x|^2<1/2$ and since we assumed that the prescribed data has distinct eigenvalues, we find as announced that $\nabla_a(\Kcirc^2)^a_b\neq 0$ for $0<|x|^2<1/2$.
\end{proof}

\paragraph{Data with non-trivial derivatives.}

In the course of proving the classification the momentum constraint reduces to an equation of the form
\bel{dchi-gammadzeta}
\del_a \gamma\,\Kcirc{}^a_b = \sum_I \chi_I \del_b \zeta_I ,
\ee
where $\gamma,\chi_I,\zeta_I$ are some scalar functions of the singularity data.
The following lemma states that this equation implies that both sides vanish separately.
Its proof is analogous to that of \autoref{lem:prescribed-values} and we give it in \autoref{app:scattering}.

\begin{lemma}[Extrinsic curvature and derivatives of scalars]
  \label{lem:eigen-deriv}
  Let $\Hcal$ be a $3$-manifold.
  Let $\gamma,\chi_I,\zeta_I$ be a finite collection of continuous ultralocal scalar fields such that for any singularity data $(g,K,\phi_0,\phi_1)$ on~$\Hcal$ the relation~\eqref{dchi-gammadzeta} holds
  at all points $x\in\Hcal$ such that $K(x)$ has three distinct eigenvalues.
  Then $\gamma$ is a constant: it only depends on the data through the signature of~$g$.
\end{lemma}

\paragraph{Derivatives of scalars are independent.}

Thanks to \autoref{lem:eigen-deriv} the momentum constraint reduces from~\eqref{dchi-gammadzeta} down to the vanishing of a sum of terms $\chi_I\del_b\zeta_I$, $b=1,2,3$, with $\chi_I,\zeta_I$ being ultralocal scalars.
Expanding the derivatives gives a linear combination of $\del_b\theta,\del_b\phi_0,\del_b\phi_1$, and the following lemma states that these derivatives are linearly independent in a suitable sense.
Its proof is analogous to that of \autoref{lem:prescribed-values} and we give it in \autoref{app:scattering}.

\begin{lemma}[Linear independence of derivatives of scalars]
  \label{lem:independent-derivatives}
  Let $\Hcal$ be a $3$-manifold and let $\mu,\nu,\varkappa$ be continuous ultralocal scalar fields such that for any singularity data $(g,K,\phi_0,\phi_1)$ on~$\Hcal$ one has 
\bel{munuxi}
  \mu \del_b \theta + \nu \del_b \phi_0 + \varkappa \del_b \phi_1 = 0 , \qquad b=1,2,3,
\ee
  at all points $x\in\Hcal$ such that $K(x)$ has three distinct eigenvalues.
  Then one has $\mu=\nu=\varkappa=0$.
\end{lemma}

\subsection{Technical steps for the classification of singularity scattering maps}
\label{app:scattering}

\paragraph{Proof of \autoref{lem:eigen-deriv}.}

Fix a signature for the metric once and for all.
We use the singularity data set on~$\Hcal$ that we constructed in the proof of \autoref{lem:prescribed-values} for the case of a single prescribed value $(\Ku,\phiu_0,\phiu_1)$.  We recall now solely the aspects that we need in this proof.
Outside a ball $B\subset\Hcal$ which we identify (by a diffeomorphism) with the unit ball in~$\RR^3$, the data are trivial.  Inside the ball, we have four layers, in which all scalar functions are radial, in the sense that they only depend on~$|x|^2$.  The data at~$0$ matches the prescribed value $(\Ku,\phiu_0,\phiu_1)$.  In the first three layers,
\be
g = \Omega^{-2/3} \diag(\pm 1,1,1) , \qquad \Kcirc = \Omega \Kucirc ,
\ee
with a suitably chosen radial function $\Omega>0$.
In particular, $\Omega(0)=1$ to reproduce the prescribed data at~$0$, while at the outer edge of the third layer $\Omega$ is equal to a value~\eqref{phi0-eq-Omega} for which $\Omega=\phi_0$.
In the last layer,
\be
g = \phi_0^{-2/3} \diag(\pm 1,1,1) , \qquad \Kcirc = \phi_0 \alpha \Kucirc ,
\ee
with a suitably chosen radial function~$\alpha$ interpolating smoothly from $\alpha=1$ at the inner edge of the layer to $\alpha=0$ in a neighborhood of the boundary of the ball, say for $|x|^2\geq 5/6$.  We can select~$\alpha$ so that it is positive for $|x|^2<5/6$.

Crucially, $\Kcirc$ is proportional to $\Kucirc$ throughout, and tends to~$0$ at $|x|^2=5/6$.  Choose now $\Ku$ to have pairwise distinct eigenvalues, so that $K(x)$ also does in the region $|x|^2<5/6$.  By assumption, we thus have
\be
\sum_I \chi_I \del_b \zeta_I = \del_a \gamma\,\Kcirc{}^a_b = (k_b -\tfrac{1}{3}) \del_b \gamma , \quad b=1,2,3,
\ee
in this region.  Since all scalars are radial we find
\be
2x_b \sum_I \chi_I \zeta_I' = 2x_b (k_b -\tfrac{1}{3}) \gamma' , \quad b=1,2,3,
\ee
where primes denote $|x|^2$~derivatives.  Away from the coordinate planes we can divide by $2x_b$ and take the difference of two of these equations to get $(k_b-k_a)\gamma'=0$ for all $a,b=1,2,3$.  Since $K(x)$ is non-degenerate for $|x|^2<5/6$, we learn that $\gamma'=0$ in this region minus the coordinate planes.  Since $\gamma$ is a radial function we finally get that $\gamma$ is a constant on $0<|x|^2<5/6$.  By continuity, $\gamma(0)$ and $\gamma(5/6)$ also take the same value.  The scalar field thus takes the same value for the prescribed data as for the trivial data $\Kcirc=0$, $\phi_0=1/\sqrt{12\pi}$, $\phi_1=0$.  Let us call this constant value~$\gamma_0$
Now $\gamma-\gamma_0$ obeys the conditions of \autoref{lem:continuity} hence $\gamma=\gamma_0$ throughout~$\Hcal$ for arbitrary data sets, as we wanted to prove.

\paragraph{Proof of \autoref{lem:independent-derivatives}}

As usual, \autoref{lem:scalars} implies that $\mu,\nu,\varkappa$ are functions of $(\theta,\phi_0,\phi_1)$.
We prove $\varkappa=0$, $\nu=0$, and $\mu=0$, in this order, by applying the continuity \autoref{lem:continuity} after proving these identities for data with any prescribed value $(\Ku,\phiu_0,\phiu_1)$ such that $\phiu_0\neq 0$ and $\Ku$~has pairwise distinct eigenvalues.
We take $\Ku$~diagonal without loss of generality.
Let us show $\varkappa=0$. As in previous proofs, we construct the data in layers in the unit ball $B\subset\RR^3$ and work in the standard basis.
We only describe the first layer, as the construction can easily be completed, using the same layers as in \autoref{lem:prescribed-values}, to a data set on~$B$ that is trivial near the boundary.
In the first layer, say $|x|^2\leq 1/2$, we keep $K=\Ku$ and $\phi_0=\phiu_0$ constant but vary $\phi_1$ and the metric, which we choose to be diagonal and of determinant $\pm1$, that is, 
$
  g = \diag(\pm e^{h_1},e^{h_2},e^{h_3})$
with $ h_1+h_2+h_3=0$. 
For this data, the momentum constraint simplifies further than~\eqref{momentum-diagonal}, namely we have 
$
\frac{-1}{2} \sum_a \ku_a \del_b h_a = 8\pi\phiu_0 \del_b \phi_1$ ($b=1,2,3$), 
whose solution is
\bel{eq23094}
\phi_1 = \frac{-1}{16\pi\phiu_0} \sum_a \ku_a h_a + \text{constant} .
\ee
The functions $h_a(|x|^2)$ are arbitrary except for $h_1+h_2+h_3=0$ and for the fact that they must vanish together with all their derivatives at~$0$ and $1/2$, so we can arrange that $\phi_1'\neq 0$ in ther interval $0<|x|^2<1/2$.
On the other hand, \eqref{munuxi} reads
$
0 = \varkappa \del_b \phi_1 = 2x_b \varkappa \phi_1' ,
$
so we conclude that $\varkappa$ vanishes near~$0$, except along the coordinate planes.
By continuity, $\varkappa$ vanishes at~$0$, namely it vanishes for the prescribed data.
By \autoref{lem:continuity}, $\varkappa=0$ identically.

Next, to prove $\nu=0$, we consider exactly the data set on $B\subset\RR^3$ used for \autoref{lem:prescribed-values}, but specify further the conformal factor~$\Omega$ used there.  Recall that in the first layer $\Kcirc=\Omega\Kucirc$ where $\Ku$ is the prescribed value and $\Omega=\Omega(|x|^2)$ interpolates between $\Omega(0)=1$ and some value at the outer boundary $|x|^2=1/4$ of the layer, with all derivatives vanishing at these end-points.  We choose $\Omega$ such that it takes the value $1$ again for some $|x|^2\in(0,1/4)$, but with a non-zero derivative.  Then the data at this point is equal to the prescribed value, but $\del_b\phi_0\neq 0$.  On the other hand, by construction of the data set, $\Kcirc$~is everywhere a non-negative multiple of the given data, in other words $\theta$ is constant.
Thus, \eqref{munuxi} reads $\nu \del_b \phi_0=0$ hence $\nu=0$.
Since this holds for arbitrary prescribed data, we conclude that $\nu=0$ identically.

Finally, proving $\mu=0$ requires building singularity data sets with variable~$\theta$, and we can ignore how $\phi_0,\phi_1$ vary since we already showed $\nu=\varkappa=0$.
We use the conformally flat data constructed in~\eqref{data-lem-non-degenerate}.
In particular, this data set has $\Kcirc=\abs{\phi_0}\Ht$ with
\be
\Ht^a_b = \alpha(|x|^2) \,\Kucirc^a_b + \beta(|x|^2) \,\bigl(x^a x_b - \tfrac{1}{3} \delta^a_b |x|^2\bigr).
\ee
Contrarily to what we do below~\eqref{data-lem-non-degenerate}, we now take $\alpha=\alpha(|x|^2)$ and $\beta=\beta(|x|^2)$ to be different radial functions.
Specifically, we fix some generic point $y\in B$ (we determine later the genericity condition) and impose some values for $\alpha,\beta,\alpha',\beta'$ at that particular point:
\be
\alpha(|y|^2) = 1 , \quad \beta(|y|^2) = 0 , \quad
\alpha'(|y|^2) = \Kucirc^a_b y^b y_a , \quad \beta'(|y|^2) = - \Tr(\Kucirc^2) .
\ee
Then, $\Ht(y) = \Kucirc$ while $\del_b\Ht(y)^a_c=2y_b\bigl(\alpha'\Kucirc^a_c+\beta'(y^a y_c - \tfrac{1}{3}\delta^a_c|y|^2)\bigr)$, so
\be
\del_b \Tr(\Ht^2) = 2\Tr(\Ht\del_b\Ht) = 4x_b\bigl(\alpha'\Tr(\Kucirc^2)+\beta'y^a \Kucirc_a^c y_c\bigr) = 0 .
\ee
On the other hand, we have 
\be
\aligned
\del_b \Tr(\Ht^3) = 3\Tr(\Ht^2\del_b\Ht)
& = 6y_b\bigl(\alpha'\Tr(\Kucirc^3) + \beta'\bigl(y^a (\Kucirc^2)^c_a y_c - \tfrac{1}{3}|y|^2\Tr(\Kucirc^2)\bigr)\bigr) \\
& = 2y_b(\ku_1-\ku_2)(\ku_2-\ku_3)(\ku_3-\ku_1)\bigl((\ku_2-\ku_3)y^1y_1 + (\ku_3-\ku_1)y^2y_2 + (\ku_1-\ku_2)y^3y_3\bigr)
\endaligned
\ee
where we obtained the second line by explicitly writing down all terms, using $\Tr\Kucirc=0$, and factorizing.  For generic~$y$ the result is nonzero provided eigenvalues of $\Kucirc$ are pairwise distinct.

We are interested in the Kasner angle, which one can get from \eqref{Kcirc-powers-tr}, and using that $\Kcirc$ is a positive multiple of~$\Ht$:
$
\cos(3\theta) = \frac{\Tr\Kcirc^3}{\sqrt{6}(\Tr\Kcirc^2)^{3/2}}
= \frac{\Tr\Ht^3}{\sqrt{6}(\Tr\Ht^2)^{3/2}} .
$
Taking a derivative and evaluating at~$y$ we get
\bel{dcos3theta}
-3 \sin(3\theta)\del_b\theta = \frac{\del_b\Tr(\Ht^3)}{\sqrt{6}(\Tr\Ht^2)^{3/2}} \neq 0.
\ee
As a consistency check we note that $\sin(3\theta)$ is indeed non-zero when $\theta\neq 0\bmod{\pi/3}$, namely when eigenvalues of~$K$ are pairwise distinct.
An important consequence of~\eqref{dcos3theta}, however, is that $\del_b\theta\neq 0$.
Then \eqref{munuxi} $\mu\del_b\theta=0$ implies that $\mu=0$ at the point~$y$, where data can take any prescribed value with $\phi_0\neq 0$ (so that the conformal scaling makes sense) and pairwise distinct eigenvalues for~$K$.  We conclude by \autoref{lem:continuity} that $\mu=0$ identically.
This establishes \autoref{lem:independent-derivatives}.

\subsection{Structure of scattering maps}
\label{ssec:struct}

\paragraph{Reduction to pointwise scattering maps.}

Let us consider a spacelike or timelike ultralocal singularity scattering map~$\Sbf$ on some (unimportant) $3$-manifold~$\Hcal$.  Its restriction $\Sbf_x$ to any one point $x\in\Hcal$ can be described as follows.
Any choice of local coordinates near~$x$ identifies the space of possible values of $(\gmoi, \Kmoi, \phimoi_0, \phimoi_1)(x)$ to the finite-dimensional space $\Ipoint$ of tuples $(\gpoint,\Kpoint,\phipoint_0,\phipoint_1)$ such that $\phipoint_0, \phipoint_1\in\RR$, $\gpoint$~is a quadratic form on $\RR^3$ with signature ${\pm}{+}{+}$ in the spacelike or timelike case, and $\Kpoint$~is a matrix that is symmetric with respect to~$\gpoint$ and that obeys $\Tr\Kpoint=1$ and $1-\Tr(\Kpoint^2)=8\pi\phipoint_0^2$.
Under this identification, $\Sbf$ yields a map $\Spoint\colon\Ipoint\to\Ipoint$ that is independent of the choice of~$x$ and of local coordinates thanks to diffeomorphism invariance.
Changing local coordinates acts with a matrix $A\in GL(3,\RR)$ on both sides of the singularity, namely
\[
\Spoint(A\cdot(\gpoint,\Kpoint))=A\cdot\Spoint(\gpoint,\Kpoint),
\]
where $A$~acts in the obvious manner $\gpoint_{ab}\mapsto A_a^cA_b^d\gpoint_{cd}$ and $\Kpoint^a_b\mapsto (A^{-1})^a_c A_b^d \Kpoint^c_d$.  We arrive at a first useful description of ultralocal singularity scattering maps.

\begin{lemma}[Reduction to pointwise scattering maps]\label{prop-ultralocal}
Specifying an ultralocal singularity scattering map~$\Sbf$ is equivalent to specifying a $GL(3,\RR)$-covariant map $\Spoint\colon\Ipoint\to\Ipoint$ that preserves the momentum constraint in the following sense.
  For any data $(\gmoi,\Kmoi,\phimoi_0,\phimoi_1)$ on some three-manifold $\Hcal$, and $(\gpoi,\Kpoi,\phipoi_0,\phipoi_1)$ its image under pointwise application of~$\Spoint$, one has:
\[
  \text{if } \nablamoi_a \Kmoi^a_b = 8 \pi \, \phimoi_0 \del_b \phimoi_1
  \text{ then } \nablapoi_a \Kpoi^a_b = 8 \pi \, \phipoi_0 \del_b \phipoi_1 .
\]
\end{lemma}

\paragraph{Polynomial structure of extrinsic curvature.}

We have seen in \autoref{lem:scalars} that scalars such as $\phipoi_0$ and~$\phipoi_1$ are simply functions of $\theta,\phi_0,\phi_1$ that are even and $2\pi/3$-periodic (spacelike case) or $2\pi$-periodic (timelike case) in~$\theta$, and that are $\theta$-independent for $\phi_0=\pm 1/\sqrt{12\pi}$.
The tensors $\gpoi$ and~$\Kpoi$ are likewise constrained by covariance under $GL(3,\RR)$ (change of basis).  We focus first on~$\Kpoi$ for definiteness, then we apply the same arguments to $(\gmoi)^{-1}\gpoi$, and finally to its logarithm after showing it exists and is real.

Let us work in an orthonormal basis of eigenvectors of~$\Kmoi$, namely a basis $v_1,v_2,v_3$ in which $\gmoi=\diag(\pm 1,1,1)$ and $\Kmoi$ is diagonal.
The change of basis mapping one of the eigenvectors~$v_a$ to its opposite does not affect $\gmoi$ and~$\Kmoi$ hence the image $(\gpoi,\Kpoi,\phipoi_0,\phipoi_1)$ of $(\gmoi,\Kmoi,\phimoi_0,\phimoi_1)$ under the scattering map is also unaffected.
However, off-diagonal components of $\Kpoi$ in the basis $v_1,v_2,v_3$ change sign under such a change of basis, so they must vanish.
We learn that $\Kpoi$ is diagonal in the same basis as $\gmoi$ and~$\Kmoi$.
If $k_{1-},k_{2-},k_{3-}$ are all distinct the three matrices $\delta=\diag(1,1,1)$, $\Kmoi=\diag(k_{1-},k_{2-},k_{3-})$, and $K_-^2=\diag(k_{1-}^2,k_{2-}^2,k_{3-}^2)$ span the space of all diagonal matrices, so $\Kpoi$ is a linear combination of them.
It is most convenient later on to work with the traceless $\Kpoicirc=\Kpoi-\frac{1}{3}\delta$ and powers of $\Kmoicirc=\Kmoi-\frac{1}{3}\delta$, and write
\bel{Kpoib0delta}
\Kpoicirc = \beta_0 \delta + \beta_1 \Kmoicirc + \beta_2 \Kmoicirc^2
\ee
for some functions $\beta_0,\beta_1,\beta_2$ of $(\thetamoi,\phimoi_0,\phimoi_1)$.
Tracelessness of~$\Kpoicirc$ imposes $\beta_0=\frac{-1}{3}\beta_2\Tr\Kmoicirc^2$, of course, but it is more convenient for us to keep all three functions.
At this stage of the argument, these functions are only defined when eigenvalues are all distinct, namely when $\thetamoi\neq 0\bmod{\pi/3}$.
Let us comment on periodicity.
Exchanging the eigenvectors $v_2$ and~$v_3$ maps $\thetamoi\to-\thetamoi$ and swaps $k_{2\pm}\leftrightarrow k_{3\pm}$, which must leave~\eqref{Kpoib0delta} invariant, so the functions $\beta_0,\beta_1,\beta_2$ are even in~$\thetamoi$.
Likewise, they are $2\pi/3$ periodic in the spacelike case because the cyclic permutation $v_1\to v_2\to v_3\to v_1$ permutes eigenvalues $k_{a-}$ and $k_{a+}$ in the same way and maps $\thetamoi\to\thetamoi+2\pi/3$.
In the timelike case this cyclic permutation is not available because $v_1$~is singled out as being timelike with respect to the metric~$\gmoi$.

Whenever two eigenvalues of $\Kmoi$ coincide (say, $k_{1-}=k_{2-}$ for definiteness), the corresponding eigenvalues of~$\Kpoi$ also do, as we now prove.
For this, change basis in $\operatorname{Span}(v_1,v_2)$ from $v_1,v_2$ to another orthonormal pair of vectors $v'_1,v'_2$ with the same timelike/spacelike nature, namely with $\gmoi(v_i,v_j)=\gmoi(v'_i,v'_j)$ for $1\leq i,j\leq 2$.
This is an $O(2,\RR)$ or $O(1,1,\RR)$ transformation depending on signature.
The change of basis leaves $\gmoi,\Kmoi$ invariant hence must leave $\gpoi,\Kpoi$ invariant.
In particular, $\Kpoi$~remains diagonal, namely $v'_1,v'_2$ are also eigenvectors of~$\Kpoi$, which implies that $v_1$ and~$v_2$ have the same eigenvalue under~$\Kpoi$.

From this fact, and assuming that singularity scattering maps map smooth data to twice differentiable data, it would be possible to prove that $\beta_0,\beta_1,\beta_2$ extend to continuous functions for all $(\thetamoi,\phimoi_0,\phimoi_1)$.
The analysis is somewhat tedious but we will not need it:
indeed, \autoref{lem:non-degenerate} ensures that studying a singularity scattering map restricted to non-degenerate data is enough to fully characterize it.
We will simply impose at the end that the scattering maps we find have well-defined limits when two Kasner exponents coincide.

\paragraph{Polynomial structure of scattering maps.}

The arguments above apply if we replace $\Kpoi$ by the matrix $g_-^{-1}\gpoi$ with components $\gmoi^{ab}\gpoi_{bc}$, and they lead to expressing this matrix as a linear combination of $\delta,\Kmoi,\Kmoi^2$ with coefficients that are possibly singular at $r_-=0$.
The real matrix $g_-^{-1}\gpoi$ is diagonal in the real basis $v_1,v_2,v_3$ hence it has real eigenvalues.
They are non-zero since the matrix is invertible (with inverse $g_+^{-1}\gmoi$).
Consider briefly the special case where the two spacelike eigenvalues $k_{2-}=k_{3-}$ of $\Kmoi$ coincide.
Then as proven above for~$\Kpoi$, the entries $(2,2)$ and $(3,3)$ of $g_-^{-1}\gpoi$ are equal, from which we deduce $\gpoi(v_2,v_2)=\gpoi(v_3,v_3)$.
Because $\gpoi$~is diagonal and has signature ${-}{+}{+}$, exactly one of its diagonal entries must be negative, so by elimination $\gpoi(v_1,v_1)<0$.
We thus learn that eigenvalues of $g_-^{-1}\gpoi$ are all positive in this case $k_{2-}=k_{3-}$.
To extend the result to any $\Kmoi$, consider a smooth singularity data set $(\gmoi,\Kmoi,\phimoi_0,\phimoi_1)$ interpolating between a point where $k_{2-}=k_{3-}$ and a point with the desired value of~$\Kmoi$.
Continuity of the metric~$\gpoi$ implies that eigenvalues of $g_-^{-1}\gpoi$ vary continuously.  Since they are positive at a point and cannot vanish, they are positive everywhere.
We conclude that the matrix $g_-^{-1}\gpoi$ has {\sl positive} eigenvalues only.
The matrix thus admits a logarithm, to which the arguments above apply as well.
We conclude that the matrix $\log\bigl(g_-^{-1}\gpoi\bigr)$ is a linear combination of $\delta,\Kmoi,K_-^2$ too, as long as the $k_{a-}$ are pairwise distinct.

In practice, instead of $\delta,\Kmoi,K_-^2$ we write matrices as linear combinations of two other sets of matrices.
For $\Kpoi$ we write $\Kmoi=\tfrac{1}{3}+\Kmoicirc$ and express $\delta,\Kmoi,K_-^2$ as linear combinations of $\delta,\Kmoicirc,\Kmoicirc^2$ as stated above.
For the metric we express these matrices further in terms of $\delta,\cos(\Theta_-),\cos(2\Theta_-)$ where $\Theta_-=\diag(\theta_-,\theta_-+2\pi/3,\theta_-+4\pi/3)$: we recall the relations~\eqref{Kcirc-powers}
\bel{Kcirc-powers-copy}
\Kmoicirc = \frac{2\rmoi}{3} \cos\Theta_- , \qquad
\Kmoicirc^2 = \frac{2\rmoi^2}{9} \bigl(\delta+\cos(2\Theta_-)\bigr) .
\ee
Importantly, the angle $\thetamoi$, hence the matrices $\Theta_-$, $\cos\Theta_-$, and $\cos(2\Theta_-)$, are ill-defined at $\rmoi=0$.  Thus, a linear combination $\alpha_0\delta+\alpha_1\cos\Theta_-+\alpha_2\cos(2\Theta_-)$ only has a well-defined limit if the scalar fields $\del_{\thetamoi}\alpha_0,\alpha_1,\alpha_2$ vanish at $\rmoi=0$.
We deduce the following lemma.

\begin{lemma}[Polynomial structure of scattering maps]\label{lem:poly}
Any ultralocal singularity scattering map obeys
\bel{ultralocal-poly}
\aligned
\gpoi & = \exp\bigl(\alpha_0\,\delta + \alpha_1\cos\Theta_- + \alpha_2\cos(2\Theta_-)\bigr) \gmoi, 
\qquad
\qquad
\Kpoicirc = \beta_0 + \beta_1\Kmoicirc + \beta_2\Kmoicirc^2 , 
\endaligned
\ee
in which $\alpha_0,\alpha_1,\alpha_2,\beta_0,\beta_1,\beta_2$ are scalar functions, like $\phipoi_0,\phipoi_1$, namely functions of $\thetamoi,\phimoi_0,\phimoi_1$ that are even and periodic in~$\thetamoi$ with period $2\pi/3$ in the spacelike case and $2\pi$ in the timelike case.
In addition, $\del_{\thetamoi}\alpha_0,\alpha_1,\alpha_2$ vanish at $\rmoi=0$.
\end{lemma}

The tracelessness $\Tr\Kpoicirc=0$ translates to $\beta_0=-\beta_2\Tr(\Kmoicirc^2)/3$ but we do not need this for now.
We also do not analyse yet how $\beta_0,\beta_1,\beta_2$ behave at $\rmoi=0$.
The exponential in~\eqref{ultralocal-poly} is defined by its power series; it yields a matrix, whose upper index we lower using $\gmoi$, so as to obtain the $(0,2)$ tensor~$\gpoi$.  We easily compute the inverse metric and the ratio~$\omega$ of volume factors, which simplifies because $\cos\Theta_-$ and $\cos(2\Theta_-)$ are traceless:
\bel{poly-omega}
g_+^{-1} = \exp\bigl(-\alpha_0\,\delta - \alpha_1\cos\Theta_- - \alpha_2\cos(2\Theta_-)\bigr) g_-^{-1},
\qquad
\omega \coloneqq \sqrt{|\gpoi|/|\gmoi|} = e^{3\alpha_0/2} .
\ee

\subsection{Scaling of trace-free extrinsic curvature}
\label{ssec:scaling}

\paragraph{Simplifying the momentum constraint.}

Our next step is to plug the polynomial form~\eqref{ultralocal-poly} into the momentum constraint.
Since $\Kpoi$ has a constant trace, its trace-free part $\Kpoicirc=\Kpoi-\tfrac{1}{3}\delta$ has the same divergence as~$\Kpoi$ and the constraint reads 
$
\nablapoi_a \Kpoicirc{}^a_b = 8 \pi \, \phipoi_0 \del_b \phipoi_1 .
$
The Levi-Civita connections of two metrics $\gpoi$ and~$\gmoi$ differ by a tensor, whose components are
\[
(\Gammapoi-\Gammamoi)^c{}_{ab}
= \Gammapoi^c{}_{ab} - \Gammamoi^c{}_{ab}
= \frac{1}{2} (g_+^{-1})^{cd} \bigl( \nablamoi_a \gpoi_{bd} + \nablamoi_b \gpoi_{da} - \nablamoi_d \gpoi_{ab} \bigr) ,
\]
where we wrote the inverse metric of~$\gpoi$ explicitly as~$g_+^{-1}$ for emphasis.
Using the well-known identity $\Gamma^a{}_{ac}=\del_c(\log|g|^{1/2})$ we find
$
(\Gammapoi-\Gammamoi)^a{}_{ac}
= \Gammapoi^a{}_{ac} - \Gammamoi^a{}_{ac}
= \del_c(\log\omega)$.
Combining the above, we compute
\bse
\bel{ultralocal-div-Kpoi}
\aligned
\nablapoi_a \Kpoicirc{}^a_b
& = \nablamoi_a \Kpoicirc{}^a_b
+ \Kpoicirc{}^c_b (\Gammapoi-\Gammamoi)^a{}_{ac}
- \Kpoicirc{}^a_c (\Gammapoi-\Gammamoi)^c{}_{ab}
\\
& = \nablamoi_a \Kpoicirc{}^a_b
+ \Kpoicirc{}^a_b \del_a(\log\omega)
- \frac{1}{2} \Kpoicirc{}^a_c (g_+^{-1})^{cd} \nablamoi_b \gpoi_{da}
 = \omega^{-1}\, \nablamoi_a\bigl(\omega \Kpoicirc{}^a_b\bigr)
- \frac{1}{2} X_b. 
\endaligned
\ee
Here, we used the symmetry of $\Kpoicirc{}^a_c (g_+^{-1})^{cd}$ and cancelled two terms in $(\Gammapoi-\Gammamoi)^c{}_{ab}$,
while we introduced a notation for the last term: 
\bel{proof-def-X}
X_b \coloneqq \Kpoicirc{}^a_c (g_+^{-1})^{cd} \nablamoi_b \gpoi_{da} .
\ee
\ese

\paragraph{Most terms involve derivatives of scalars.}

The term~$X_b$ defined in~\eqref{proof-def-X} can be recast (using $\nablamoi\gmoi=0$) as the trace of a product of matrices:
\[
X_b = \Kpoicirc{}^a_c (g_+^{-1}\gmoi)^c_d \nablamoi_b (g_-^{-1}\gpoi)^d_a
= \Tr\bigl( \Kpoicirc{} (g_+^{-1}\gmoi) \nablamoi_b (g_-^{-1}\gpoi) \bigr) .
\]
Given their explicit polynomial forms in \autoref{lem:poly}, $\Kpoicirc{}$ and $g_+^{-1}\gmoi$ commute, so $(g_+^{-1}\gmoi) \nablamoi_b (g_-^{-1}\gpoi)$ can be replaced within the trace by $\nablamoi_b \log(g_-^{-1}\gpoi)$.
Explicitly,
\bel{proof-X-form}
X_b = \Tr\bigl( \Kpoicirc \nablamoi_b \log(g_-^{-1}\gpoi) \bigr)
= \Tr\bigl( \bigl(\beta_0+\beta_1\Kmoicirc+\beta_2\Kmoicirc^2\bigr)
\nablamoi_b \bigl(\alpha_0\,\delta + \alpha_1\cos\Theta_- + \alpha_2\cos(2\Theta_-)\bigr) \bigr) .
\ee
By writing $\cos\Theta_-$ and $\cos(2\Theta_-)$ in terms of $\delta,\Kmoicirc,\Kmoicirc^2$ using~\eqref{Kcirc-powers-copy} we obtain polynomial expressions in~$\Kmoicirc$.
Expanding further, the derivative $\nablamoi_b$ can either act on scalars~$\alpha_n$, giving terms of the form $\Tr({\dots})\del_b\alpha_n$, or act on powers of~$\Kmoicirc$, giving terms of the form $\Tr(\Kmoicirc^n\nablamoi_b\Kmoicirc)$ times a scalar.  Since $\Tr(\Kmoicirc^n\nablamoi_b\Kmoicirc)=\del_b\Tr(\Kmoicirc^{n+1})/(n+1)$ is the derivative of a scalar, all terms in~$X_b$ take the form $(\text{scalar})\del_b(\text{scalar})$.

The momentum constraint on the ``$+$'' side of the singularity states that $\nablapoi_a\Kpoicirc{}^a_b$ is also of the form $(\text{scalar})\del_b(\text{scalar})$, so \eqref{ultralocal-div-Kpoi}~can be written as
\bel{proof-momentum-simplified}
\nablamoi_a\bigl(\omega\Kpoicirc{}^a_b\bigr) = \sum_I \chi_I\del_b\zeta_I
\ee
for some collection of scalar fields $\chi_I$ and~$\zeta_I$ whose precise expression is not useful yet.

\paragraph{Scaling of trace-free extrinsic curvature.}

To get rid of derivatives of scalar fields in~\eqref{proof-momentum-simplified}, we consider particular configurations $(\gmoi,\Kmoi,\phimoi_0,\phimoi_1)$ constructed in \autoref{lem:derivatives}.  These data sets are such that $(\thetamoi,\phimoi_0,\phimoi_1)$ is constant in some domain $\Omega\subset\Hcal$ and is equal to any prescribed value in~$\Delta_\neq$.  This set, defined in~\eqref{non-degenerate-scalars}, consists of values such that the corresponding eigenvalues $k_1,k_2,k_3$ are pairwise distinct and $\phimoi_0\neq 0$.
Since all scalars are functions of $\thetamoi,\phimoi_0,\phimoi_1$, the first derivative $\del_a$ of any scalar then vanishes at~$x$.
In addition, the data sets are such that $\nablamoi_a(\Kmoicirc^2)^a_b\neq 0$ on~$\Omega$.

For these data sets, the right-hand side of~\eqref{proof-momentum-simplified} vanishes in the domain~$\Omega$.
We compute its left-hand side in~$\Omega$ by plugging the polynomial form~\eqref{ultralocal-poly}, then dropping all derivatives of scalar fields since they vanish for this configuration:
\[
\nablamoi_a\bigl(\omega\Kpoicirc{}^a_b\bigr)
= \nablamoi_a\bigl(\omega\beta_0\delta^a_b+\omega\beta_1\Kmoicirc{}^a_b + \omega\beta_2(\Kmoicirc^2)^a_b\bigr) 
 = \omega\beta_1\,\nablamoi_a\Kmoicirc{}^a_b + \omega\beta_2\,\nablamoi_a(\Kmoicirc^2)^a_b .
\]
The momentum constraint is $\nablamoi_a\Kmoicirc{}^a_b=8\pi\phimoi_0\del_b\phimoi_1$, which vanishes at~$x$ in the given configuration.  This eliminates the first term above and we learn that
\[
\omega\beta_2\,\nablamoi_a(\Kmoicirc^2)^a_b = 0 .
\]
For the data sets given by \autoref{lem:derivatives}, $\nablamoi_a(\Kmoicirc^2)^a_b\neq 0$, so we learn that $\beta_2=0$.  Because $\Kpoicirc$ is traceless we deduce $\beta_0=0$.  Altogether,
\[
\beta_0=\beta_2=0, \qquad \Kpoicirc=\beta_1\Kmoicirc \qquad
\text{when $\Kmoi$ has three different eigenvalues.}
\]
By a continuity argument identical to the proof of Lemmas~\ref{lem:continuity} and~\ref{lem:non-degenerate} we could prove that this conclusion holds even when eigenvalues are degenerate, but we do not need this.

\paragraph{Constant scaling of densitized trace-free extrinsic curvature.}

Now that we know $\Kpoicirc=\beta_1\Kmoicirc$ (for non-degenerate data) we can recalculate the left-hand side of~\eqref{proof-momentum-simplified} without assuming that scalar fields have vanishing derivative.  We get
\[
\nablamoi_a\bigl(\omega\Kpoicirc{}^a_b\bigr)
= \nablamoi_a\bigl(\omega\beta_1\,\Kmoicirc{}^a_b\bigr)
= \del_a(\omega\beta_1)\,\Kmoicirc{}^a_b + 8\pi\omega\beta_1\phimoi_0\del_b\phimoi_1 ,
\]
so \eqref{proof-momentum-simplified} takes the form
\[
\del_a(\omega\beta_1)\,\Kmoicirc{}^a_b
= - 8\pi\omega\beta_1\phimoi_0\del_b\phimoi_1 + \sum_I \chi_I \del_b \zeta_I .
\]
This identity takes the form~\eqref{dchi-gammadzeta} analyzed in \autoref{lem:eigen-deriv}, so we learn that the scalar coefficient $\omega\beta_1$ in front of $\Kmoicirc{}^a_b$ is an overall constant that only depends on the signature (and the scattering map), so
\bel{proof-Kg-scaling}
\Kpoicirc = \gamma \omega^{-1} \Kmoicirc
\ee
for some constant $\gamma\in\RR$.
Note that the conclusion of \autoref{lem:eigen-deriv} does not involve any non-degeneracy assumption: the identity holds for all data.  This gives an alternate proof of our \autoref{cor:Kg-scaling} that does not rely on the full classification.

A useful consequence of~\eqref{proof-Kg-scaling} is
\bel{rpoirmoi}
\rpoi = \sqrt{\tfrac{2}{3}\Tr\Kpoicirc^2} = \abs{\gamma}\omega^{-1} \sqrt{\tfrac{2}{3}\Tr\Kmoicirc^2} = \abs{\gamma}\omega^{-1}\rmoi .
\ee
For $\rpoi\neq 0$ (hence $\gamma\neq 0$ and $\rmoi\neq 0$ due to the above equation), one can write
\bel{Thetapoimoi}
\cos\Theta_+ = \frac{3\Kpoicirc}{2\rpoi} = (\sgn\gamma)\frac{3\Kmoicirc}{2\rmoi} = (\sgn\gamma)\cos\Theta_- ,
\quad \text{ hence } \begin{cases}
  \theta_+ = \theta_- & \text{if $\gamma>0$,} \\
  \theta_+ = \theta_- + \pi & \text{if $\gamma<0$.}
\end{cases}
\ee
We emphasize that while angles are only defined modulo $2\pi/3$ in the spacelike case because the three eigenvectors are indistinguishable, their difference $\theta_+-\theta_-$ is actually well-defined modulo~$2\pi$ in both the spacelike and timelike case because one can compare eigenvalues of $\Kpoi$ and~$\Kmoi$ on the same eigenvectors.

\subsection{Completion of the classification}
\label{ssec:completion}

\paragraph{Rigidly conformal case.}

As a warmup, we derive the classification of rigidly conformal and ultralocal maps announced in \autoref{prop-conformal-ultralocal}.
Specifically, we temporarily restrict ourselves to ultralocal scattering maps for which $\gpoi$ and $\gmoi$ have the same conformal class, namely $\alpha_1=\alpha_2=0$.
Then the expression~\eqref{proof-X-form} vanishes thanks to $\Tr\Kmoicirc=0$,
\[
X_b = \Tr\bigl(\beta_1 \Kmoicirc \nablamoi_b(\alpha_0\delta)\bigr)
= \Tr(\Kmoicirc \delta) \beta_1\del_b\alpha_0 = 0, 
\]
so~\eqref{ultralocal-div-Kpoi}, together with the momentum constraints, gives
\[
8\pi\phipoi_0\del_b\phipoi_1
= \nablapoi_a \Kpoicirc{}^a_b
= \gamma\omega^{-1}\nablamoi_a\Kmoicirc{}^a_b
= 8\pi\gamma\omega^{-1}\phimoi_0\del_b\phimoi_1 .
\]
The scalars $\omega,\phipoi_0,\phipoi_1$ are some functions of the scalars $(\thetamoi,\phimoi_0,\phimoi_1)$ given by the data.
The chain rule for $\phipoi_1=\phipoi_1(\thetamoi,\phimoi_0,\phimoi_1)$ yields
\[
8\pi\bigl(\phipoi_0\,\del_{\thetamoi}\phipoi_1\bigr)\,\del_b\thetamoi
+ 8\pi\bigl(\phipoi_0\,\del_{\phimoi_0}\phipoi_1\bigr)\,\del_b\phimoi_0
+ 8\pi\bigl(\phipoi_0\,\del_{\phimoi_1}\phipoi_1-\gamma\omega^{-1}\phimoi_0\bigr)\,\del_b\phimoi_1
= 0.
\]
By \autoref{lem:independent-derivatives}, the coefficients of $\del_b\thetamoi,\del_b\phimoi_0,\del_b\phimoi_1$ must vanish separately, namely
\bel{proof-conformal-coefs-vanish}
\phipoi_0\,\del_{\thetamoi}\phipoi_1
= \phipoi_0\,\del_{\phimoi_0}\phipoi_1
= \phipoi_0\,\del_{\phimoi_1}\phipoi_1-\gamma\omega^{-1}\phimoi_0 = 0.
\ee

Then, there are two very different cases, $\gamma=0$ and $\gamma\neq 0$.
\bei
\item
  If $\gamma=0$, we have $\Kpoicirc=0$ so $r(\phipoi_0)=0$ namely $\phipoi_0=\epsilon/\sqrt{12\pi}$ with $\epsilon=\pm 1$.
  This sign is constant since we require scattering maps to map sufficiently regular data to (at least) continuous data.
  Since $\phipoi_0\neq 0$, \eqref{proof-conformal-coefs-vanish} simply states that $\phipoi_1$ is a constant, while $\omega$ is completely unconstrained.
  This yields the isotropic scattering map given in~\eqref{Sirc}, with $\lambda^3=\omega$:
\[
  \Sirc_{\lambda,\varphi,\epsilon} \colon (g, K, \phi_0, \phi_1)
  \mapsto \biggl(\lambda^2 g, \ \frac{1}{3}\delta,\ \frac{\epsilon}{\sqrt{12\pi}},\ \varphi \biggr) .
\]

\item
  If $\gamma\neq 0$, then the last equation in~\eqref{proof-conformal-coefs-vanish} prevents $\phipoi_0$ from vanishing unless $\phimoi_0=0$.  Thus, we learn that $\del_{\thetamoi}\phipoi_1=\del_{\phimoi_0}\phipoi_1=0$ for $\phimoi_0\neq 0$, and, by continuity of~$\phipoi_1$, for $\phimoi_0=0$ as well.  In other words, $\phipoi_1=F(\phimoi_1)$ for some $F\colon\RR\to\RR$.
  The last equation in~\eqref{proof-conformal-coefs-vanish} reads
\bel{proof-conf-2}
  \phipoi_0\,F'(\phimoi_1) = \gamma\omega^{-1}\phimoi_0 ,
\ee
  which implies that $F'$~is nowhere vanishing (since it is independent of~$\phimoi_0$).

  We then have to solve~\eqref{proof-conf-2} and the Hamiltonian constraint
\[
  1 - 12 \pi \phipoi_0^2 = \frac{3}{2} \Tr(\Kpoicirc^2)
  = \frac{3}{2} \gamma^2 \omega^{-2} \Tr(\Kmoicirc^2)
  = \gamma^2 \omega^{-2} (1 - 12 \pi \phimoi_0^2)
\]
  for $\phipoi_0$ and~$\omega$.  Eliminating $\phipoi_0$ using~\eqref{proof-conf-2} gives
\bel{proof-conf-3}
  \gamma^{-2}\omega^2 = 1 + 12 \pi \phimoi_0^2 \bigl(F'(\phimoi_1)^{-2} - 1\bigr) .
\ee
  It is then immediate to solve~\eqref{proof-conf-2} for~$\phipoi_0$.
  Denoting $\mu\coloneqq\abs{\omega/\gamma}^{1/3}$, given in terms of $\phimoi_0,\phimoi_1$ by~\eqref{proof-conf-3}, and denoting $\epsilon=\sgn\gamma=\pm 1$, we find
\be
  \aligned
  \gpoi & = \omega^{2/3} \gmoi = |\gamma|^{2/3} \mu^2 \gmoi,
 \qquad
  &&\Kpoicirc = \epsilon \mu^{-3} \Kmoicirc,
 \\
  \phipoi_0 & = \epsilon \mu^{-3} \frac{\phimoi_0}{F'(\phimoi_1)},
 \qquad
  &&\phipoi_1  = F(\phimoi_1),
  \endaligned
\ee
  which is nothing by the anisotropic rigidly conformal scattering map defined in~\eqref{Sarc}.
\eei
This concludes the classification in \autoref{prop-conformal-ultralocal} of ultralocal scattering maps that are rigidly conformal.

\paragraph{Isotropic case.}
We return to general ultralocal scattering maps, in which $\alpha_1,\alpha_2$ may be nonzero.
The value of the constant~$\gamma$ plays a key role again in the classification.
We treat in this paragraph the case $\gamma=0$, namely $\Kpoicirc=0$: the asymptotic profile on the ``$+$''~side of the singularity undergoes {\sl isotropic scaling.}

In this case, the Hamiltonian constraint forces $\phipoi_0=\epsilon/\sqrt{12\pi}$ for some $\epsilon=\pm 1$.  Because we require scattering maps to map smooth data to (at least) continuous data, for such data $\phipoi_0$~cannot jump between the values $\pm 1/\sqrt{12\pi}$, namely $\epsilon(x)$ is independent of $x\in\Hcal$.  By \autoref{lem:constants} we learn that $\epsilon$~only depends on the scattering map and not on the data.
Next, since $\Kpoi=\tfrac{1}{3}\delta$ is constant and $\phipoi_0\neq 0$, the momentum constraint states that $\del_b\phipoi_1=0$.  Again we have a space-independent scalar~$\phipoi_1$, which by \autoref{lem:constants} can only depend on the scattering map.
Finally, the metric is not constrained beyond the polynomial structure given in \autoref{lem:poly}.
This yields the isotropic scattering $\Siso_{\alpha_0,\alpha_1,\alpha_2,\varphi,\epsilon}$ of~\eqref{Siso}:
\be
(\gpoi,\Kpoi,\phipoi_0,\phipoi_1)
= \biggl( \exp\Bigl(\alpha_0\,\delta + \alpha_1 \cos\Theta_- + \alpha_2 \cos(2\Theta_-)\Bigr) \gmoi , \ \frac{1}{3}\delta , \ \frac{\epsilon}{\sqrt{12\pi}}, \ \varphi \biggr), 
\ee
where $\del_{\thetamoi}\alpha_0=\alpha_1=\alpha_2=0$ for $\rmoi=0$ (namely $\phimoi_0=\pm 1/\sqrt{12\pi}$).

\paragraph{Anisotropic case.}

We now turn to the case $\gamma\neq 0$, using the same method as for the rigidly conformal maps.
A convenient form for the trace-free part $\Kpoicirc$ is
\bel{eq0934}
\Kpoicirc = \tfrac{2}{3}\rpoi\cos\Theta_+ = \tfrac{2}{3}\epsilon\rpoi\cos\Theta_- , \qquad
\text{with } \epsilon = \sgn\gamma = \pm 1,
\ee
where the second equality is obvious for $\rpoi=0$ and is~\eqref{Thetapoimoi} otherwise.
As we will see momentarily, inserting this expression of~$\Kpoicirc$ in the momentum constraint reduces it down to a short sum of terms of the form $(\text{scalar})\del_b(\text{scalar})$.
The chain rule rewrites the sum as a linear combination of $\del_b\thetamoi$, $\del_b\phimoi_0$, $\del_b\phimoi_1$, whose coefficients must all vanish by \autoref{lem:independent-derivatives}.
This vanishing gives three equations on derivatives of $\alpha_0,\alpha_1,\alpha_2,\phipoi_0,\phipoi_1$ with respect to $\thetamoi,\phimoi_0,\phimoi_1$, and we eventually get the solutions $\Sani_{\Phi,c,\epsilon}$ defined by~\eqref{Sani}.

Let us begin by calculating using~\eqref{eq0934} the remainder term~$X_b$ given in~\eqref{proof-X-form}:
\[
X_b = \tfrac{2}{3}\epsilon\rpoi\Tr\bigl( \cos\Theta_- \nablamoi_b \bigl(\alpha_0\,\delta+\alpha_1\cos(\Theta_-)+\alpha_2\cos(2\Theta_-)\bigr) \bigr) .
\]
Upon expanding derivatives we encounter the traces: $\Tr(\cos\Theta_-)=0$ and 
\[
\aligned
\Tr\bigl((\cos\Theta_-)^2\bigr)&=\tfrac{1}{2}\Tr\bigl(\cos(2\Theta_-)+\delta\bigr)=\tfrac{3}{2}, \\
\Tr\bigl(\cos\Theta_-\cos(2\Theta_-)\bigr) &= \tfrac{1}{2}\Tr\bigl(\cos(3\Theta_-)+\cos\Theta_-\bigr) = \tfrac{3}{2}\cos(3\theta_-) , \\
\Tr\bigl(\cos\Theta_-\,\del_b\cos\Theta_-\bigr) &= \tfrac{1}{2}\del_b \Tr\bigl((\cos\Theta_-)^2\bigr) = \tfrac{1}{2} \del_b(\tfrac{3}{2})=0 , \\
\Tr\bigl(\cos\Theta_-\,\del_b\cos(2\Theta_-)\bigr) &= \tfrac{4}{3}\del_b\Tr\bigl((\cos\Theta_-)^3\bigr) = \del_b\cos(3\theta_-) .
\endaligned
\]
Then $X_b$~simplifies to
\[
\aligned
X_b &= \tfrac{2}{3}\epsilon\rpoi\Bigl( \tfrac{3}{2}\del_b\alpha_1 + \tfrac{3}{2}\cos(3\theta_-)\del_b\alpha_2 + \alpha_2\,\del_b\cos(3\theta_-)\Bigr) \\
& = \epsilon\rpoi\Bigl( \del_b\bigl(\alpha_1 + \cos(3\theta_-)\alpha_2\bigr) - \tfrac{1}{3} \alpha_2\,\del_b\cos(3\theta_-)\Bigr) 
 = 16\pi \rpoi \del_b\xi + \epsilon\rpoi \alpha_2 \sin(3\theta_-)\del_b\theta_- ,
\endaligned
\]
where we introduced (with a factor chosen to simplify later expressions)
\bel{def-kappa}
\xi = \frac{\epsilon}{16\pi} \bigl(\alpha_1 + \cos(3\theta_-)\alpha_2\bigr) .
\ee

Using the divergence given in~\eqref{ultralocal-div-Kpoi} and our calculation of~$X_b$, the momentum constraint reads
\[
8\pi\phipoi_0\del_b\phipoi_1 - 8\pi\gamma\omega^{-1}\phimoi_0\del_b\phimoi_1
= \frac{-1}{2}X_b = \frac{-1}{2} \Bigl( 16\pi\rpoi\del_b\xi + \epsilon\rpoi\alpha_2 \sin(3\theta_-)\del_b\theta_-\Bigr) ,
\]
hence, dividing by $\rpoi=|\gamma|\omega^{-1}\rmoi$ provided it is nonzero,
\[
\del_b\xi + \frac{\phipoi_0}{\rpoi}\del_b\phipoi_1
= \epsilon \frac{\phimoi_0}{\rmoi}\del_b\phimoi_1 - \frac{1}{16\pi}\epsilon\alpha_2 \sin(3\theta_-)\del_b\theta_- .
\]

We then use the chain rule to write all $\del_b(\text{scalar})$ in terms of $\del_b\thetamoi$, $\del_b\phimoi_0$, $\del_b\phimoi_1$ and we write down the three equations stating that coefficients of these three derivatives must match due to \autoref{lem:independent-derivatives}:
\bel{derivatives-of-kappa}
\del_{\thetamoi}\xi + \frac{\phipoi_0}{\rpoi}\del_{\thetamoi}\phipoi_1 = - \frac{1}{16\pi}\epsilon\alpha_2 \sin(3\theta_-),
\quad
\del_{\phimoi_0}\xi + \frac{\phipoi_0}{\rpoi}\del_{\phimoi_0}\phipoi_1 = 0 ,
\quad
\del_{\phimoi_1}\xi + \frac{\phipoi_0}{\rpoi}\del_{\phimoi_1}\phipoi_1
 = \epsilon \frac{\phimoi_0}{\rmoi} .
\ee
The first equation lets us rewrite in terms of~$\xi$ the terms that appear in the polynomial form~\eqref{ultralocal-poly} of~$\gpoi$: first express $\cos(2\Theta_-)$ as $\cos(x-y)=\cos x\cos y+\sin x\sin y$ for $x=3\Theta_-$ and $y=\Theta_-$, then use \eqref{def-kappa} and~\eqref{derivatives-of-kappa}.  This yields
\[
\aligned
\alpha_0 \,\delta + \alpha_1 \cos\Theta_- + \alpha_2 \cos(2\Theta_-) 
& = \alpha_0 \,\delta + (\alpha_1 + \alpha_2\cos(3\theta_-)) \cos\Theta_-
+ \alpha_2 \sin(3\theta_-)\sin\Theta_- \\
& = \alpha_0 \,\delta + 16\pi\epsilon\xi \cos\Theta_-
-16\pi\epsilon\bigl(\del_{\thetamoi}\xi + \frac{\phipoi_0}{\rpoi}\del_{\thetamoi}\phipoi_1\bigr) \sin\Theta_- .
\endaligned
\]
On the other hand, \eqref{poly-omega} and \eqref{rpoirmoi} relate $\alpha_0$ to Kasner radii as
$\exp(3\alpha_0/2) = \omega = \abs{\gamma} \rmoi/\rpoi$.
Overall, for nonzero $\rmoi,\rpoi$,
\bel{proof-Sani-gpoi}
\aligned
\gpoi & = \exp\bigl(\alpha_0 \,\delta + \alpha_1 \cos\Theta_- + \alpha_2 \cos(2\Theta_-)\bigr)\,\gmoi\\
& = \biggl|\frac{\gamma\rmoi}{\rpoi}\biggr|^{2/3} \exp\biggl(16\pi\epsilon\xi\cos\Theta_- - 16\pi\epsilon\bigl( \del_{\thetamoi} \xi + \frac{\phipoi_0}{\rpoi} \del_{\thetamoi} \phipoi_1 \bigr) \sin\Theta_-\biggr)\,\gmoi.
\endaligned
\ee

The expressions we obtained for $(\gpoi,\Kpoi,\phipoi_0,\phipoi_1)$ coincide with those of the anisotropic scattering $\Sani_{\Phi,c,\epsilon}$ \eqref{Sani-gpoi} with $c=\abs{\gamma}^{1/3}$ and $\Phi\colon(\thetamoi,\phimoi_0,\phimoi_1)\mapsto (\phipoi_0,\phipoi_1)$, but there remains to show that the map~$\Phi$ is indeed an $\epsilon$-canonical transformation in the sense of \autoref{def:canonical-transfo}.
We prove the conditions in turn.
\bei
\item[(i)] {\bf Periodic.}  This condition simply states that $\phipoi_0,\phipoi_1$ are scalar fields.

\item[(ii)] {\bf Maximal-momentum preserving.}  We know $\rpoi = \abs{\gamma}\omega^{-1}\rmoi$ from~\eqref{rpoirmoi}, and $\abs{\gamma}\omega^{-1}$ is nowhere vanishing since $\omega$ is the ratio of volume factors of two non-degenerate metrics.  Thus, $\rpoi=0\iff\rmoi=0$ and $\rpoi/\rmoi=\abs{\gamma}\omega^{-1}$ remains finite and non-zero and becomes $\theta$-independent at the boundary (the last point being because $\omega$ is a scalar field).

\item[(iii)] {\bf Volume preserving.}
Imposing that the last two equations in~\eqref{derivatives-of-kappa} are compatible in the sense that $\del_{\phimoi_0}\del_{\phimoi_1}\xi = \del_{\phimoi_1}\del_{\phimoi_0}\xi$, we obtain for $\rpoi,\rmoi\neq 0$ that
\be
\del_{\phimoi_0}\Bigl(\frac{\phipoi_0}{\rpoi}\Bigr)\del_{\phimoi_1}\phipoi_1 - \del_{\phimoi_1}\Bigl(\frac{\phipoi_0}{\rpoi}\Bigr)\del_{\phimoi_0}\phipoi_1
= \epsilon \del_{\phimoi_0}\Bigl(\frac{\phimoi_0}{\rmoi}\Bigr) \del_{\phimoi_1} \phimoi_1 .
\ee
We included here the trivial factor $\del_{\phimoi_1} \phimoi_1$ to illustrate that this equation states preservation of the two-form $d(\phi_0/r)d\phi_1$ up to an overall sign~$\epsilon$.
 
\item[(iv)] {\bf Regular at boundaries.} One conclusion in \autoref{lem:poly} is that $\alpha_1=\alpha_2=0$ at the boundaries $\phimoi_0=\pm 1/\sqrt{12\pi}$ (because $\Kmoicirc=0$ has no preferred directions).  We deduce $\xi=0$, and its derivatives $\del_{\thetamoi}\xi$, $\del_{\phimoi_1}\xi$ along the boundaries thus vanish.  Inserting this fact (and $\alpha_2=0$) into~\eqref{derivatives-of-kappa} yields
\be
\frac{\phipoi_0}{\rpoi}\del_{\thetamoi}\phipoi_1 \to 0 , \qquad
\frac{\phipoi_0}{\rpoi}\del_{\phimoi_1}\phipoi_1 - \epsilon \frac{\phimoi_0}{\rmoi} \to 0 .
\ee
Since $\xi$ vanishes on both boundaries we can integrate on~$I_0$ the second equation in~\eqref{derivatives-of-kappa} to get
\be
\int_{-1/\sqrt{12\pi}}^{1/\sqrt{12\pi}} \frac{\phipoi_0}{\rpoi} \del_{\phimoi_0}\phipoi_1 \, d\phimoi_0 = 0 .
\ee
\eei
This concludes the proof of the first part of \autoref{theorem-ultralocal}, that is, 
the only ultralocal scattering maps are $\Siso_{\alpha_0,\alpha_1,\alpha_2,\varphi,\epsilon}$ and $\Sani_{\Phi,c,\epsilon}$.  The second part was proven already in \autoref{secti-24}.

This also concludes our study of scattering maps per se.
In the companion paper in~\cite{LLV-3}, we study the class of plane-symmetric spacetimes \cite{FloresSanchez:2003,FloresSanchez:2008,KhanPenrose,PLFStewart,Penrose0,Penrose,Szekeres70,Yurtsever88} 
and we 
apply our theory to the particular scenario of colliding plane-symmetric gravitational waves.


\paragraph*{Acknowledgments.} 

The authors wish to thank Paul Tod for pointing out to them the relevance of the reference~\cite{Barrow}. The authors gratefully acknowledge support from the Simons Center for Geometry and Physics, Stony Brook University. 
This research was initiated when the first author (BLF) was a research fellow at Princeton University, while the second author (PLF)
 was a visiting fellow at the Courant Institute for Mathematical Sciences, New York University.
The second author is also grateful to the Institut Mittag-Leffler where he attended the Program: ``General relativity, geometry and analysis: beyond the first 100 years after Einstein''. 


\end{document}